\documentclass[acmsmall,screen,authorversion]{acmart}

\usepackage{algorithm}
\usepackage[noend]{algpseudocode}
\usepackage{amssymb,amsthm}
\usepackage{anyfontsize}
\usepackage{booktabs} 
\usepackage{enumitem}
\usepackage{graphicx}
\usepackage{mathtools}
\usepackage{numprint} \npthousandsep{,}  
\usepackage{subfig}
\usepackage{tikz}
\usepackage{xcolor}

\newif\ifcomment
\commentfalse  
\newcommand{\bluecomment}[1]{\ifcomment\color{blue} #1 \color{black}\fi}
\newcommand{\redcomment}[1]{\ifcomment\color{red} #1 \color{black}\fi}

\newcommand{\bigO}[1]{\ensuremath{\mathcal{O}(#1)}}
\newcommand{\p}{\mathbb{P}}

\newcommand{\pmin}{p_{min}}
\newcommand{\pmax}{p_{max}}
\newcommand{\M}{\mathcal{M}}
\newcommand{\A}{\mathcal{A}}

\newcommand{\W}{\mathcal{W}}
\newcommand{\cbound}{2+\sqrt{2}}  
\newcommand{\ebound}{2.17}        

\newcommand{\ex}{x}
\newcommand{\lin}{m}
\newcommand{\Ohf}{\Omega^*}
\newcommand{\Gtri}{G_{\Delta}}
\newcommand{\Ghex}{G_{hex}}
\newcommand{\sthree}{1.732050808}  

\newtheorem{thm}{Theorem}[section]
\newtheorem{lem}[thm]{Lemma}

\newtheorem{cor}[thm]{Corollary}

\newtheorem{defn}[thm]{Definition}

\newtheorem{property}{Property}

\acmYear{2019}
\acmMonth{2}

\setcopyright{none}
\copyrightyear{2019}

\begin{document}

\title[A Markov Chain Algorithm for Compression in Self-Organizing Particle Systems]{A Markov Chain Algorithm for Compression in Self-Organizing Particle Systems}

\author{Sarah Cannon}
\orcid{0000-0001-6510-4669}
\affiliation{%
	\institution{University of California, Berkeley}
	\city{Berkeley}
	\state{CA}
	\country{USA}}
\email{sarah.cannon@berkeley.edu}

\author{Joshua J. Daymude}
\orcid{0000-0001-7294-5626}
\affiliation{%
    \institution{Arizona State University}
    \department{Computer Science, CIDSE}
    \city{Tempe}
    \state{AZ}
    \country{USA}}
\email{jdaymude@asu.edu}

\author{Dana Randall}
\affiliation{%
	\institution{Georgia Institute of Technology}
	\city{Atlanta}
	\state{GA}
	\country{USA}}
\email{randall@cc.gatech.edu}

\author{Andr\'ea W. Richa}
\affiliation{%
    \institution{Arizona State University}
    \department{Computer Science, CIDSE}
    \city{Tempe}
    \state{AZ}
    \country{USA}}
\email{aricha@asu.edu}

\begin{abstract}
In systems of programmable matter, we are given a collection of simple computation elements (or {\it particles}) with limited (constant-size) memory. We are interested in when they can self-organize to solve system-wide problems of movement, configuration and coordination.
Here, we initiate a {\it stochastic approach} to developing  robust distributed algorithms for programmable matter systems using Markov chains.
We are able to leverage the wealth of prior work in Markov chains and related areas to design and rigorously analyze our distributed algorithms and show that they have several desirable properties.

We study the {\it compression} problem, in which a particle system must gather as tightly together as possible, as in a sphere or its equivalent in the presence of some underlying geometry.
More specifically, we seek fully distributed, local, and asynchronous algorithms that lead the system to converge to a configuration with small boundary.
We present a Markov chain-based algorithm that solves the compression problem under the {\em geometric amoebot model}, for particle systems that begin in a connected configuration.
The algorithm takes as input a bias parameter $\lambda$, where $\lambda > 1$ corresponds to particles favoring having more neighbors.
We show that for all $\lambda > \cbound$, there is a constant $\alpha > 1$ such that eventually with all but exponentially small probability the particles are {\it $\alpha$-compressed}, meaning the perimeter of the system configuration is at most $\alpha \cdot p_{{min}}$, where $p_{{min}}$ is the minimum possible perimeter of the particle system.
Surprisingly, the same algorithm can also be used for {\it expansion} when $0 < \lambda < \ebound$, and we prove similar results about expansion for values of $\lambda$ in this range.
This is counterintuitive as it shows that particles preferring to be next to each other ($\lambda > 1$) is not sufficient to guarantee compression.
Since its first appearance in the conference version of this paper, we have further validated this new stochastic approach in subsequent works by using it to provably accomplish a variety of other objectives in programmable matter.

\end{abstract}



\maketitle

\section{Introduction} \label{sec:intro}

Many programmable matter systems have recently been proposed and realized --- modular and swarm robotics, synthetic biology, DNA and molecular programming, and smart materials form an incomplete list --- and each is often tailored toward a particular task domain or physical setting.
We abstract away from specific settings and instead describe programmable matter as a collection of simple computational elements (called {\it particles}) with limited computational power.
These particles individually execute fully distributed, local, asynchronous algorithms to collectively solve system-wide problems of movement, configuration, and coordination.
We assume a formal model of programmable matter
known as the \emph{geometric amoebot model} (discussed in Section~\ref{subsec:pm-relatedwork} and formalized in Section~\ref{subsec:model}), where particles occupy vertices of the triangular lattice~$\Gtri$ (Fig.~\ref{fig:model_gtri}) and move along its edges.

We desire for a particle system to {\it compress}, gathering tightly together, approaching a sphere or its equivalent in the presence of some underlying geometry.
Formally, the {\it compression problem} seeks a reorganization of a particle system (via asynchronous local particle movements) such that the system converges to a {configuration} with small boundary, where we refer to the total length of this boundary as the {\it perimeter} of the configuration.
This compression phenomenon is often found in natural systems: fire ants form floating rafts by gathering in such a manner, and honey bees communicate foraging patterns by swarming within their hives.
While each individual ant or bee cannot view the group as a whole when soliciting information, it can take cues from its immediate neighbors to achieve cooperation.

It is with this motivation that we present a distributed algorithm for compression under the geometric amoebot model that is derived from a Markov chain process.
Because our distributed algorithm comes from a Markov chain, we are able to leverage established techniques from stochastic processes to analyze it and provide guarantees about its behavior.
The stochasticity of Markov chains also implies that our distributed algorithm is inherently robust and oblivious, two desirable properties of distributed algorithms.

\subsection{Our Approach} \label{subsec:approach}

We solve the compression problem via the {\it stochastic approach to self-organizing particle systems}.
We introduced this approach in the conference version of this paper~\cite{Cannon2016}, and have since successfully applied it to other problems (e.g.,~\cite{Arroyo2018,Cannon2018}).
At a high level, we first define an energy function that captures our objectives for the particle system.
We then design a Markov chain that, in the long run, favors particle system configurations with desirable energy values.
Care is taken to ensure this Markov chain can be translated to a fully distributed, local, asynchronous algorithm run by each particle individually.

The motivation underlying the design of this Markov chain is from statistical physics, where ensembles of particles similar to those we consider represent physical systems and demonstrate that local micro-behavior can induce global, macro-scale changes to the system~\cite{Baxter1980,Blanca2018,Restrepo2013}.
Like a spring relaxing, physical systems favor configurations that minimize energy.
Each system configuration~$\sigma$ has energy determined by a \emph{Hamiltonian} $H(\sigma)$ and a corresponding weight $w(\sigma) = \text{e}^{-B \cdot H(\sigma)}$, where $B = 1/T$ is inverse temperature.
Markov chains have been well-studied as a tool for sampling system configurations with probabilities proportional to their weight $w(\sigma)$, where configurations with the least energy $H(\sigma)$ have the highest weight and are thus most likely to be sampled.

In our stochastic approach to self-organizing particle systems, we introduce a Hamiltonian $H(\sigma)$ over particle system configurations $\sigma$ that assigns the lowest values to desirable configurations; we then design a Markov chain algorithm to favor these configurations with small Hamiltonians.
To solve the compression problem, we let $H(\sigma) = -e(\sigma)$, where $e(\sigma)$ is the number of edges induced by~$\sigma$, i.e., the number of lattice edges in $\Gtri$ with both endpoints occupied by particles.
Setting $\lambda = \text{e}^B$, we get $w(\sigma) = \lambda^{e(\sigma)}$.
In Section~\ref{subsec:compression}, we will show that favoring induced edges is equivalent to favoring shorter perimeter.
Thus, as $\lambda$ gets larger (by increasing $B$, effectively lowering temperature), we increasingly favor configurations with a large number of induced edges, which are those that are compressed.

Using a Metropolis Filter (Section~\ref{subsec:mc}), we can design a Markov chain $\M$ that uses only local moves and eventually reaches a distribution that favors configurations proportional to their weight $w(\sigma)$.
That is, we design $\M$ such that the eventual probability of the particle system being in configuration $\sigma$ is $w(\sigma) / Z$, where $Z = \sum_{\sigma'} w(\sigma')$ is a normalizing constant known as the {\it partition function}.
In this {\it stationary distribution} of $\M$, a configuration with more edges --- and smaller $H(\sigma)$ --- occurs with higher probability.
We then use tools from Markov chain analysis to prove the even stronger statement that non-compressed configurations occupy only an exponentially small fraction of this stationary distribution.
Because we carefully design $\M$ to use only local moves, it can be implemented in a distributed, asynchronous setting by a (decentralized) self-organizing particle system.

This stochastic approach to developing distributed algorithms for programmable matter is applicable beyond the compression problem; it has the potential to solve any problem where the objective can be described as minimizing some energy function, provided changes in that energy function can be calculated with only local information.
In Section~\ref{sec:concl}, we give a more detailed discussion of what properties this energy function needs to be amenable to our approach.

\subsection{Our Results and Techniques} \label{subsec:techniques}

Formally, we present a {\it Markov chain $\M$} for compression under the geometric amoebot model. We show that $\M$ can be directly translated into a {\em fully distributed, local, asynchronous algorithm} $\A$: when each particle independently executes the steps of $\A$, the overall behavior of the system is equivalent to that of the process $\M$.
Both $\M$ and $\A$ start in an arbitrary system configuration $\sigma_0$ of $n$ particles that is connected.
A {\it bias parameter} $\lambda$ is given as input, where $\lambda > 1$ corresponds to particles favoring having more neighbors.
Markov chain $\M$ is carefully designed so that the particle system always remains connected and no new holes form.
Furthermore, to allow us to apply many of the standard tools of Markov chain analysis, we prove  $\M$ is  eventually {\it ergodic}.
While the proofs of these two facts mostly use only first principles, we emphasize they are far from trivial; working in a distributed setting necessitates carefully defined protocols for local moves that make these proofs challenging.

As the particles individually execute distributed algorithm $\A$, the unique stationary distribution~$\pi$ of $\M$ is eventually reached.
In this stationary distribution, the probability that the system is in any particular configuration $\sigma$ is given by $\pi(\sigma)$.
We prove that for all $\lambda > \cbound$, there is a constant $\alpha = \alpha(\lambda) > 1$ such that at stationarity, with all but a probability that is exponentially small (in $n$, the number of particles), the particle system is {\it $\alpha$-compressed}, meaning the perimeter is at most $\alpha$ times the minimum perimeter for $n$ particles $\pmin = \Theta(\sqrt{n})$.
In fact, for any $\alpha > 1$, our algorithms can accomplish $\alpha$-compression by setting $\lambda$ to be large enough.

We additionally show the counterintuitive result that $\lambda > 1$ is not sufficient to guarantee compression, even when there are a large number of particles.
In fact, for all $0 < \lambda < \ebound$, there is a constant $\beta < 1$ such that at stationarity, with all but exponentially small probability, the perimeter is at least a $\beta$ fraction of the maximum perimeter $\pmax = \Theta(n)$.
We call such a configuration {\it $\beta$-expanded}.
This implies that for any $0 < \lambda < \ebound$, the probability that the particle system is $\alpha$-compressed is exponentially small for any constant $\alpha > 1$.

The key tool used to establish compression and expansion is a careful \emph{Peierls argument}, used in statistical physics to study non-uniqueness of limiting Gibbs measures and to determine the presence of phase transitions (see, e.g.,~\cite{Dobrushin1968}), and in computer science to establish slow mixing of Markov chains (see, e.g.,~\cite{Borgs1999}).
We carefully design $\M$ and $\A$ to ensure the particle system stays connected and eventually all holes are eliminated and do not reform. This means our Peierls arguments are significantly simpler than many standard Peierls arguments on configurations that are not required to be connected and can have holes.


\subsection{Related Work}

We now discuss related work, which spans three general areas: programmable matter, distributed compression/clustering behavior, and random particle processes on graphs and lattices.

\vspace{2mm}\noindent{\bf Programmable Matter Systems and Models.} \label{subsec:relwork-pm}
\label{subsec:pm-relatedwork}
To develop a system of \emph{programmable matter}, one endeavors to create a material or substance that utilizes user input or stimuli from its environment to change its physical properties in a programmable fashion.
Many such systems have been realized; a non-exhaustive list includes:
\begin{itemize}
\item {\it DNA computing}, where strands of DNA programmed with specific base sequences combine in solution to form specific arrangements~\cite{Adleman1994};
\item {\it Smart materials}, such as 3D-printed wood that bends in a preprogrammed way when wet~\cite{David2015};
\item {\it Modular robots} that can self-reconfigure to accomplish different tasks, such as the ReBiS robot which can switch between bipedal and snake-like movement~\cite{Thakker2014};
\item {\it Swarm robotics}, where large groups of robots collectively perform tasks, like the kilobots of~\cite{Rubenstein2014}.
\end{itemize}


Models and realizations of programmable matter can be divided into {\it active} and {\it passive} systems.
In passive programmable matter systems, which include most instances of DNA computing and smart materials, the individual elements composing the matter 
have little to no control over their movements and how they respond to their environment.
In contrast, in active programmable matter systems, individual computational units are capable of making decisions and acting on those decisions.
For example, in self-reconfigurable modular robots, each robotic module can adjust its connections to other modules in order to form different structures~\cite{Moubarak2012}, and in distributed swarms each robot makes independent decisions about what to do~\cite{Rubenstein2014}.

We will focus on active programmable matter.
Because instances of such systems are incredibly varied, we will examine an abstraction that captures features common to many different active programmable matter systems.
In a {\it self-organizing particle system}, individual units called {\it particles} with limited computational and communication abilities occupy the vertices and move along the edges of some graph (representing real space) in a distributed, asynchronous way.
More specifically, in the {\it geometric amoebot model}, these particles have constant-size memory, communicate only with their immediate neighbors, and exist on the triangular lattice (see Section~\ref{subsec:model} for details).
Since it was first introduced in 2014~\cite{Derakhshandeh2014}, the geometric amoebot model has been used to understand phenomena observed in physical robot systems~\cite{Savoie2018} and to study fundamental problems such as shape formation~\cite{Derakhshandeh2016}, leader election~\cite{Daymude2017}, and fault tolerance~\cite{DiLuna2018}.

The amoebot model is not the only abstraction of active programmable matter currently in use.
For example, \emph{metamorphic robot} systems~\cite{Chirikjian1994} model dynamically reconfiguring robots on the hexagonal lattice, and have yielded some rigorous algorithmic work (e.g.,~\cite{Walter2005}).
The {\it nubot} model~\cite{Woods2013} for molecular-scale self-assembly represents active programmable matter as monomers on the 2D triangular grid, and work has largely focused on efficient shape formation (e.g.,~\cite{Chen2014}).
While the nubot model allows rigid attachments between monomers that result in non-local movement and interactions, our amoebot model only allows local interactions between particles and yet still manages to accomplish global objectives.

\vspace{2mm}\noindent{\bf Compression, Clustering, and Gathering.}\label{subsec:relwork-comp}
Nature offers a variety of examples in which gathering and cooperative behavior is apparent.
For example, social insects often exhibit compression-like characteristics in their collective behavior: fire ants form floating rafts~\cite{Mlot2011}, cockroach larvae perform self-organizing aggregation~\cite{Jeanson2005,Rivault1998}, and honey bees choose hive locations based on a decentralized process of swarming and recruitment~\cite{Camazine1999}.
Compression is also seen in other species, for example in the slime mold {\it Dictyoselium}, whose natural life cycle includes a phase where about \numprint{100000} single-celled organisms gather together into a cluster known as a ``slug''~\cite{Devreotes1989}.

Our work on compression was originally inspired by the {\it Ising model} of statistical physics~\cite{Ising1925}, a fundamental model of ferromagnetism that has been widely studied.
In this model, all vertices of some graph are assigned a positive or negative spin, and a {\it temperature} parameter governs how likely it is for neighboring vertices to have the same spin.
For certain temperatures, we see {\it clustering}, where large regions of the graph have the same spin.
In an analogy to the Ising model, we consider the locations of our triangular lattice $\Gtri$ as having positive spin if they are occupied by particles and negative spin otherwise.
Our bias parameter $\lambda$ corresponds to inverse temperature in the Ising model, and thus governs the likelihood of having adjacent particles.
Solving the compression problem corresponds to forming a cluster of positive spins in the Ising model with {\it fixed magnetization}, where the total number of vertices with each spin does not change.
Our work diverges from the fixed magnetization Ising model by requiring that particles only move to adjacent locations and the particle system configuration remains connected, constraints not typically considered for Ising models but necessary for distributed implementations in self-organizing particle systems.

Works in a variety of areas of computer science have considered compression-type problems.
In distributed computing, the rendezvous (or gathering) problem seeks to gather mobile agents together on some node of a graph (see, e.g.,~\cite{Bampas2010} and the references within).
In comparison, our particles follow the exclusion principle, and hence are unable to gather at a single node.
Our particles are also computationally simpler than the mobile agents considered.
In swarm robotics, different variations of shape formation and aggregation problems have been studied (e.g.~\cite{Flocchini2008,Rubenstein2014,Gauci2014}), but always with robots that have more computational power or global knowledge/vision of the system than our particles do.
Similarly, pattern formation and creation of convex structures has been studied in the {\it cellular automata} domain (e.g.~\cite{Chavoya2006,Deutsch2017}), but differs from our model by assuming more powerful computational capabilities.

Lastly, in~\cite{Derakhshandeh2015-shape,Derakhshandeh2016}, algorithms for hexagon shape formation in the amoebot model were presented.
Although a hexagon satisfies our definition of a compressed configuration, these algorithms critically rely on a leader particle that coordinates the rest of the particle system.
In comparison, the Markov chain-based algorithm we present takes a fully decentralized and local approach, forgoing the need for a leader, and is naturally self-stabilizing.

\vspace{2mm}\noindent{\bf Random Particle Exclusion Processes.} \label{subsec:relwork-mc}
As opposed to earlier work in the amoebot model, we use randomization to determine particle movements.
The resulting random dynamics are an example of a {\it particle exclusion process}, where a fixed number of particles  move among the vertices of some graph by traversing its edges such that two particles  never occupy the same vertex at the same time.
There is a significant body of work analyzing Markov chains that are particle exclusion processes.
In fact, the widely-used Comparison Theorem for bounding the mixing time of Markov chains was first presented in a paper analyzing the mixing time of an unbiased exclusion process~\cite{Diaconis1993}.
However, our (distributed) setting and goals require us to diverge from many common assumptions made about exclusion processes.
For example, in our work, particle movement probabilities are not fixed ahead of time but are calculated anew in each iteration, and our random particle dynamics are constrained to ensure the particle system remains connected.
The first of these is necessary to control the probability distribution our process converges to, and the second is necessary because the amoebot model restricts communication to immediate neighbors.


\section{Background and Model} \label{sec:background}

We begin with the geometric amoebot model for programmable matter.
We then define some properties of particle systems and discuss what it means for a particle system to be compressed.

\subsection{The Amoebot Model} \label{subsec:model}

In the {\it amoebot model}, first introduced in~\cite{Derakhshandeh2014} and described in full in~\cite{Daymude2019}, programmable matter consists of individual, homogeneous computational elements called \emph{particles}.
The structure of a particle system is represented as a subgraph of an infinite, undirected graph $G = (V,E)$ where $V$ represents all positions a particle can occupy relative to its structure and $E$ represents all atomic movements a particle can make.
Each node in $V$ can be occupied by at most one particle at a time.
For compression (and many other problems), we further assume the \emph{geometric} amoebot model, in which $G = \Gtri$, the \emph{triangular lattice}\footnote{Previous works, such as the conference version of this paper~\cite{Cannon2016}, refer to $\Gtri$ as the triangular lattice $\Gamma$ or the infinite triangular grid graph $G_{\text{eqt}}$.} (Fig.~\ref{fig:model_gtri}).

\begin{figure}
\centering
\subfloat[]{
	\begin{tikzpicture}[scale=0.6]
    \clip (0.5,-0.25) rectangle (5.5,3.25);
    \foreach \i in {0,...,10} \draw[black,line width=.5pt] (\i*\sthree / 2,-5)--(\i*\sthree / 2,5);
    \foreach \i in {-10,...,10}
    {
    \draw[black,line width=.5pt] (0,\i)--(5*\sthree,\i + 5);
    \draw[black,line width=.5pt] (0,\i)--(5*\sthree,\i - 5);
    }
    \foreach \i in {0,2,...,10}
    \foreach \j in {-5,...,5}
    \draw[fill] (\i*\sthree / 2,\j) circle (0.13);
    \foreach \i in {1,3,...,10}
    \foreach \j in {-4.5,...,4.5}
    \draw[fill] (\i*\sthree / 2,\j) circle (0.13);
\end{tikzpicture}
	\label{fig:model_gtri}
} \hfil
\subfloat[]{
	\begin{tikzpicture}[scale=0.6]
    \clip (0.5,-0.25) rectangle (5.5,3.25);
    \foreach \i in {0,...,10}
    \draw[lightgray,line width=.5pt] (\i*\sthree / 2,-5)--(\i*\sthree / 2,5);
    \foreach \i in {-10,...,10}
    {
    \draw[lightgray,line width=.5pt] (0,\i)--(5*\sthree,\i + 5);
    \draw[lightgray,line width=.5pt] (0,\i)--(5*\sthree,\i - 5);
    }
    \draw[fill] (1*\sthree,1) circle (0.125);
    \draw[black,line width=1.5pt](1*\sthree,1)--(1*\sthree,2);
    \draw[fill] (1*\sthree,2) circle (0.125);
    \draw[fill] (2*\sthree,1) circle (0.125);
    \draw[black,line width=1.5pt](2*\sthree,1)--(2.5*\sthree,1.5);
    \draw[fill] (2.5*\sthree,1.5) circle (0.125);
    \draw[fill] (1.5*\sthree,2.5) circle (0.125);
    \draw[black,line width=1.5pt](1.5*\sthree,2.5)--(2*\sthree,2);
    \draw[fill] (2*\sthree,2) circle (0.125);
    \draw[fill] (1.5*\sthree,1.5) circle (0.125);
    \draw[fill] (2*\sthree,0) circle (0.125);
\end{tikzpicture}
	\label{fig:model_particles}
}
\caption{(a) A section of the triangular lattice $\Gtri$. (b) Expanded and contracted particles; $\Gtri$ is shown in gray and particles are depicted as black circles. Particles with a black line between their nodes are expanded.}
\label{fig:model}
\end{figure}
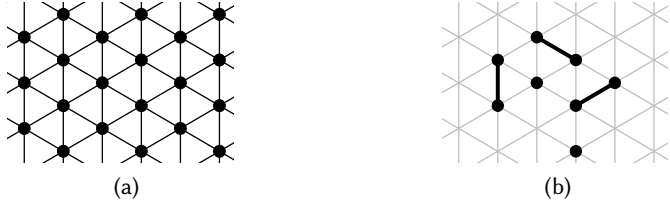

Each particle occupies either a single node in $V$ (i.e., it is \emph{contracted}) or a pair of adjacent nodes in $V$ (i.e., it is \emph{expanded}), as in Fig.~\ref{fig:model_particles}.
Particles move via a series of \emph{expansions} and \emph{contractions}: a contracted particle can expand into an unoccupied adjacent node to become expanded, and completes its movement by contracting to once again occupy a single node.
An expanded particle's \emph{head} is the node it last expanded into and the other node it occupies is its \emph{tail}; a contracted particle's head and tail are both the single node it occupies.

Two particles occupying adjacent nodes are said to be \emph{neighbors}.
Each particle is \emph{anonymous}, lacking a unique identifier, but can locally identify each of its neighboring locations and can determine which of these are occupied by particles.
Each particle has a constant-size local memory that it can write to and its neighbors can read from for communication.
In particular, a particle stores whether it is contracted or expanded in its memory.
Particles do not have any access to global information such as a shared compass, coordinate system, or estimate of the size of the system.

The system progresses through \emph{atomic actions} according to the standard asynchronous model of computation from distributed computing (see, e.g.,~\cite{Lynch1996}).
A classical result under this model states that for any concurrent asynchronous execution of atomic actions, there is a sequential ordering of actions producing the same end result, provided conflicts that arise in the concurrent execution are resolved.
In the amoebot model, an atomic action is an activation of a single particle.
Once activated, a particle can $(i)$ perform an arbitrary, bounded amount of computation involving information it reads from its local memory and its neighbors' memories, $(ii)$ write to its local memory, and $(iii)$ perform at most one expansion or contraction.
Conflicts involving simultaneous particle expansions into the same unoccupied node are assumed to be resolved arbitrarily such that at most one particle moves to some unoccupied node at any given time.
No conflicts of concurrent writes to the same memory location are possible because a particle only writes into its own memory.
Thus, while in reality many particles may be active concurrently, it suffices when analyzing our algorithm to consider a sequence of activations where only one particle is active at a time.
The resulting activation sequence is assumed to be \emph{fair}: for any inactive particle $P$ at time $t$, $P$ will be activated again at some time $t' > t$.
An \emph{asynchronous round} is complete once every particle has been activated at least once.

\subsection{Terminology for Particle Systems} \label{subsec:terminology}

In addition to the formal model, we introduce some terminology specific to the compression problem.
A particle system \emph{arrangement} is the collection of locations in $\Gtri$ that are occupied by tails of particles;\footnote{Lattice locations occupied by heads of expanded particles are not considered part of a configuration, since the states of our Markov chain consider only contracted particles. This is for technical reasons that will be explained in Section~\ref{subsec:localalga}.} note that an arrangement does not distinguish which particle occupies which location within the arrangement.
Two arrangements are equivalent if one is a translation of the other, and an equivalence class of arrangements is called a particle system \emph{configuration}.
Note that two configurations differing by rotation are distinct from a global perspective, even though each individual particle has no sense of global orientation.

An \emph{edge} of a configuration is an edge of $\Gtri$ where both endpoints are occupied by tails of particles.
Similarly, a \emph{triangle} of a configuration is a triangular face of $\Gtri$ with all three vertices occupied by tails of particles.
The number of edges (resp., triangles) of a configuration $\sigma$ is denoted $e(\sigma)$ (resp., $t(\sigma)$).
When referring to a \emph{path}, we mean a path of configuration edges.
Two particles are \emph{connected} if there is a path between them, and a configuration is \emph{connected} if all pairs of particles are.

A \emph{boundary} of a configuration $\sigma$ is a minimal closed walk $\W$ on edges of $\sigma$ that separates all particles of $\sigma$ from a connected, unoccupied subgraph of $\Gtri$ that has at least one vertex; for each boundary $\W$, let $S_\W$ be the maximal such subgraph.
If $S_\W$ is finite, we say it is a \emph{hole}.
If $S_\W$ is infinite, then $\W$ is the unique external boundary of $\sigma$.
The {\it perimeter} $p(\sigma)$ of a configuration $\sigma$ is the sum of the lengths of all boundaries of $\sigma$.
Note an edge may appear twice in the same boundary (if it is a cut-edge of $\sigma$) or in two different boundaries (e.g., if it separates two holes).
In these cases, the edge is counted twice in $p(\sigma)$.

We specifically consider connected particle system configurations.
Starting from a connected configuration (possibly with holes), our algorithm will keep the system connected, eliminate all holes, and prohibit any new holes from forming, a fact we prove in Section~\ref{subsec:invariants}.
Allowing a particle system to disconnect is generally undesirable.
Because particles can only communicate with their immediate neighbors and do not have any global orientation, disconnected components have no way of knowing their relative positions and thus cannot intentionally move toward one another to reconnect.
Furthermore, our current proof techniques require hole-free configurations.

\subsection{Compression of Particle Systems} \label{subsec:compression}

Our objective is to solve the compression problem.
There are many ways to formalize what it means for a particle system to be {\it compressed}.
For example, one could try to minimize the diameter of the system, maximize the number of edges, or maximize the number of triangles.
We choose to define compression in terms of minimizing the perimeter.
Here, we prove that for connected configurations with no holes (the states we eventually reach), minimizing perimeter, maximizing the number of edges, and maximizing the number of triangles are all equivalent and are stronger notions of compression than minimizing the diameter.

The perimeter of a connected, hole-free configuration of $n$ particles ranges from a maximum value $\pmax(n) = 2n-2$ when the system is in its least compressed state (a tree with no induced triangles) to some minimum value $\pmin(n) = \Theta(\sqrt{n})$ when the system is in its most compressed state.
It is easy to see $\pmin(n) \leq 4\sqrt{n}$, and we now prove any configuration $\sigma$ of $n$ particles has $p(\sigma) \geq \sqrt{n}$; this bound is not tight but suffices for our proofs.

\begin{lem} \label{lem:pmin}
A connected configuration with $n \geq 2$ particles has perimeter at least $\sqrt{n}$.
\end{lem}
\begin{proof}
We argue by induction on $n$.
A connected particle system with two particles necessarily has perimeter $2 \geq \sqrt{2}$.
So let $\sigma$ be any connected particle system configuration with $n > 2$ particles, and suppose the lemma holds for all connected configurations with less than $n$ particles.

First, suppose there is a particle $Q \in \sigma$ not incident to any triangles of $\sigma$.
This implies $Q$ has one, two, or three neighbors, none of which are adjacent.
If $Q$ has one neighbor, removing $Q$ from $\sigma$ yields a configuration $\sigma'$ with $n-1$ particles and, by induction, perimeter at least $\sqrt{n-1}$.
Thus,
\[p(\sigma) = p(\sigma') + 2 \geq \sqrt{n-1} + 2 \geq \sqrt{n}.\]
If $Q$ has two neighbors, removing $Q$ from $\sigma$ produces two connected particle configurations $\sigma_1$ and $\sigma_2$, where $\sigma_1$ has $n_1$ particles, $\sigma_2$ has $n_2$ particles, and $n_1 + n_2 = n-1$.
Thus,
\[p(\sigma) \geq \sqrt{n_1} + \sqrt{n_2} + 4  > \sqrt{n-1} + 4 > \sqrt{n}.\]
Similarly, if $Q$ has three neighbors, its removal produces three particle configurations with $n_1$, $n_2$, and $n_3$ particles, respectively, where $n_1 + n_2 + n_3 = n-1$.
We conclude:
\[p(\sigma) \geq \sqrt{n_1} +\sqrt{n_2} + \sqrt{n_3} + 6 > \sqrt{n-1} + 6 > \sqrt{n}.\]

Now, suppose every particle in $\sigma$ is incident to some triangle of $\sigma$, implying there are at least $\lceil n/3 \rceil$ triangles in $\sigma$.
An equilateral triangle with side length $1$ has area $\sqrt{3}/4$, so the external boundary of $\sigma$ encloses an area of at least $A = \lceil n/3 \rceil (\sqrt{3}/4) \geq \sqrt{3}n / 12$.
By the isoperimetric inequality, the minimum perimeter shape enclosing this area (without regard to lattice constraints) is a circle of radius $r$ and perimeter $p$, where
\[r = \sqrt{\frac{A}{\pi}} = \sqrt{ \frac{n\sqrt{3}}{12\pi}},
\hspace{.75cm} p = 2\pi r = \sqrt{\frac{\pi n}{\sqrt{3}}} > \sqrt{n}.\]
As the external boundary of $\sigma$ also encloses an area of at least $\sqrt{3}n/12$, we have $p(\sigma) > \sqrt{n}$.
\end{proof}

When $n$ is clear from context, we omit it and refer to $\pmin := \pmin(n)$ and $\pmax := \pmax(n)$.
We now formalize what it means for a particle system to be compressed.

\begin{defn} \label{defn:compression}
For any $\alpha > 1$, a connected configuration $\sigma$ is {\it $\alpha$-compressed} if $p(\sigma) \leq \alpha \cdot \pmin$.
\end{defn}

We prove in Section~\ref{sec:results} that our algorithm, when executed for a sufficiently long time, achieves $\alpha$-compression with all but exponentially small probability for any constant $\alpha > 1$, provided $n$ is sufficiently large.
We note that $\alpha$-compression implies the diameter of the particle system is also $\bigO{\sqrt{n}}$, so our notion of $\alpha$-compression is stronger than defining compression in terms of diameter.

In order to minimize perimeter using only simple local moves, we exploit the following relationship.
Because our algorithm eventually reaches and remains in the set of connected configurations with no holes (Section~\ref{subsec:invariants}), we only consider this case.

\begin{lem} \label{lem:edge=perim}
For a connected particle system configuration $\sigma$ with no holes, $e(\sigma) = 3n - p(\sigma) - 3$.
\end{lem}
\begin{proof}
We count particle-edge incidences, of which there are $2e(\sigma)$.
Counting another way, every particle has six incident edges, except for those on the (unique) external boundary $\W$.
At each particle on $\W$, the exterior angle is $120$, $180$, $240$, $300$, or $360$ degrees.
These correspond to the particle ``missing'' $1$, $2$, $3$, $4$, or $5$ of its six possible incident edges, respectively, or $degree/60 - 1$ missing edges.
If $\W$ visits the same particle multiple times, we count the appropriate exterior angle and corresponding missing edges each time.
From a well-known result about simple polygons with $p(\sigma)$ sides, the sum of exterior angles along $\W$ is $180 p(\sigma) + 360$ degrees.
Summing the number of ``missing'' edges	over all particles on $\W$, we find the total number of missing edges to be:
\[(180 p(\sigma) + 360) / 60 - p(\sigma) = 2p(\sigma) + 6.\]
This implies there are $6n - (2p(\sigma) + 6)$ total particle-edge incidences, so $2e(\sigma) = 6n - 2p(\sigma) - 6$.
\end{proof}

We briefly note that minimizing perimeter is also equivalent to maximizing triangles.

\begin{lem} \label{lem:tri=perim}
For a connected particle system configuration $\sigma$ with no holes, $t(\sigma) = 2n - p(\sigma) - 2$.
\end{lem}
\begin{proof}
The proof is nearly identical to that of Lemma~\ref{lem:edge=perim} but counts particle-triangle incidences, of which there are $3  t(\sigma)$.
Counting another way, every particle has six incident triangles, except for those on the external boundary $\W$.
Consider any traversal of $\W$; at each particle, the exterior angle is $120$, $180$, $240$, $300$, or $360$ degrees.
These correspond to the particle ``missing'' $2$, $3$, $4$, $5$, or $6$ of its six possible incident triangles, respectively, or $degree/60$ missing triangles.
If $\W$ visits the same particle multiple times, we count the appropriate exterior angle at each visit.
The sum of exterior angles along $\W$ is $180 p(\sigma) + 360$, so in total particles on the perimeter are missing $3p(\sigma)+ 6$ triangles.
This implies there are $6n - 3p(\sigma) - 6$ particle-triangle incidences, so $3t(\sigma) = 6n - 3p(\sigma) - 6$.
\end{proof}


The above lemmas imply the following corollary.

\begin{cor} \label{cor:p-e-t}
A connected particle system configuration with no holes and minimum perimeter is also a configuration with the maximum number of edges and the maximum number of triangles.
\end{cor}

Because these three notions of compression are equivalent, we will state our algorithm in terms of maximizing the number of edges but prove our compression results in terms of minimizing perimeter, for ease of presentation.
In the conference version of this paper~\cite{Cannon2016}, we stated our algorithm in terms of maximizing the number of triangles, but do not do so here.

\subsection{Markov Chains}\label{subsec:mc}

Our distributed algorithm for compression in self-organizing particle systems is based on a Markov chain.
A thorough treatment of Markov chains can be found in the standard textbook~\cite{Levin2009}; here, we present the necessary terminology relevant to our results.

A {\it Markov chain} is a memoryless random process on a state space $\Omega$; in this paper, $\Omega$ will always be finite and discrete.
In particular, a Markov chain randomly transitions between the states of $\Omega$ in a time-independent, or {\it stochastic}, fashion.
The probability with which the chain transitions to its next state depends only on its current state.
The chain's past states, how long it has been running, and other such factors have no effect on these probabilities.
We focus on discrete time Markov chains, where one transition occurs per {\it iteration} (or {\it step}) of the Markov chain.
Because of its stochasticity, we can completely describe a Markov chain by its transition matrix $M$, which is an $|\Omega| \times |\Omega|$ matrix indexed by the states of $\Omega$, defined such that for any pair $x,y \in \Omega$, $M(x,y)$ is the probability, if in state $x$, of transitioning to state $y$ in one step of the Markov chain.
The $t$-step transition probability $M^t(x,y)$ is the probability of transitioning from $x$ to $y$ in exactly $t$ steps.


A Markov chain is \emph{irreducible} if there is a sequence of valid transitions from any state to any other state, that is, if for all $x,y \in \Omega$ there is a $t$ such that $M^t(x,y) > 0$.
A Markov chain is \emph{aperiodic} if for all $x \in \Omega$, $\gcd\{t : M^t(x,x) > 0\} = 1$.
A Markov chain is \emph{ergodic} if it is both irreducible and aperiodic, or equivalently, if there exists $t$ such that for all $x,y \in \Omega$, $M^t(x,y) > 0$.


A {\it stationary distribution} of a Markov chain is a probability distribution $\pi$ over $\Omega$ such that $\pi M = \pi$.
Any finite, ergodic Markov chain converges to a unique stationary distribution given by $\pi(y) = \lim_{t \to \infty} M^t(x,y)$ for any $x,y \in \Omega$; importantly, for such chains this stationary distribution is completely independent of the starting state $x$.
To verify a distribution $\pi'$ is the unique stationary distribution of a finite ergodic Markov chain, it suffices to check that $\pi'(x)M(x,y) = \pi'(y)M(y,x)$ for all $x,y \in \Omega$ (the \emph{detailed balance condition}; see, e.g.,~\cite{Feller1968}).
Detailed balance will be the key to connecting our global objective, captured in the stationary distribution $\pi$, to the local moves executed by our particles, which occur with probabilities described by the transition matrix $M$.


Given a state space $\Omega$, a set of allowable transitions between states, and a desired stationary distribution $\pi$ on $\Omega$, the  Metropolis-Hastings algorithm~\cite{Hastings1970} gives a Markov chain on $\Omega$ that uses only allowable transitions and has stationary distribution $\pi$.
This is accomplished by carefully setting the probabilities of the state transitions as follows.
For a state $x\in \Omega$, its {\it neighbors} $N(x)$ are the states it can transition to and its {\it degree} is its number of neighbors.
Starting at $x \in \Omega$, the Metropolis-Hastings algorithm picks $y \in N(x)$  uniformly with probability $1/(2\Delta)$, where $\Delta$ is the maximum degree of any state, and moves to $y$ with probability $\min\{1, \pi(y)/\pi(x)\}$; with all the remaining probability, it stays at $x$ and repeats.
Using this probability calculation to decide whether or not to make a transition is known as a \emph{Metropolis filter}.
If the allowable transitions connect $\Omega$ (i.e., if the chain is irreducible), then $\pi$ must be the stationary distribution by detailed balance.
While calculating $\pi(y)/\pi(x)$ seems to require global knowledge, this ratio can often be calculated easily using only local information when many terms cancel out, as will be the case for us.
The states of the Markov chain $\M$ we consider are particle system configurations, and its transitions correspond to moves of one particle.
Each particle will calculate the Metropolis probabilities for $\M$ using only the difference in the number of neighbors (incident edges) it has before and after it moves, which can be observed locally without global information. The resulting stationary distribution of $\M$ will favor configurations with more edges and thus, by Corollary~\ref{cor:p-e-t}, smaller perimeter.

\section{Algorithms for Compression} \label{sec:markov_chain}

Our objective is to give a stochastic algorithm enabling a self-organizing particle system on the triangular lattice $\Gtri$ to provably solve the compression problem.
Remarkably, our algorithm achieves this goal despite only using one bit of information per particle for communication, even though the amoebot model allows for significantly more sophisticated communication.
Moreover, our algorithm relies only on local information: each particle only needs to know which of its adjacent locations are occupied by neighboring particles and which neighbors, if any, are expanded.

In order to leverage powerful tools from Markov chain analysis to prove our algorithm's correctness, our algorithm is designed to maintain several necessary properties.
First, assuming the particle system is initially connected, our algorithm will ensure it stays connected, eventually eliminates any holes it may contain, and prohibits any new holes from forming --- all using only local information.
Second, any moves allowed by our algorithm after all holes have been eliminated are ensured to be {\it reversible}: if a particle moves from its current location to a new location in one step, then in the next step there is a nonzero probability that it moves back to its original location.
Finally, the moves allowed by our algorithm suffice to transform any connected, hole-free particle system configuration into any other connected, hole-free configuration.

Our algorithm achieves compression by biasing particles towards moves that gain them more neighbors; i.e., where more edges with neighboring particles are formed.
Specifically, a bias parameter $\lambda$ controls how strongly the particles favor having more neighbors: $\lambda > 1$ corresponds to favoring neighbors, while $\lambda < 1$ corresponds to disfavoring neighbors.
As Lemma~\ref{lem:edge=perim} shows, locally favoring more neighbors is equivalent to globally favoring a shorter perimeter; this is the relationship we exploit to obtain particle compression.

\subsection{The Markov Chain \texorpdfstring{$\M$}{M}} \label{subsec:markovchainm}

We begin by presenting two key properties that enable a particle to move from location $\ell$ to adjacent location $\ell'$ without disconnecting the particle system or forming a hole.
We will let capital letters refer to particles and lower case letters refer to locations on the triangular lattice $\Gtri$, e.g., ``particle $P$ at location $\ell$.''
For a particle $P$ (resp., location~$\ell$), we use $N(P)$ (resp., $N(\ell)$) to denote the set of particles adjacent to $P$ (resp., to~$\ell$), where by {\it adjacent} we mean connected by a lattice edge.
For adjacent locations $\ell$ and~$\ell'$, by $N(\ell \cup \ell')$ we mean $(N(\ell) \cup N(\ell')) \setminus\{\ell, \ell'\}$.
Let ${\mathbb{S} = N(\ell) \cap N(\ell')}$ be the set of particles adjacent to both $\ell$ and $\ell'$; note $|\mathbb{S}| \in \{0, 1, 2\}$.

\begin{property} \label{prop:1}
$|\mathbb{S}| \in \{1,2\}$ and every particle in $N(\ell\cup \ell')$ is connected to a particle in $\mathbb{S}$ by a path through $N(\ell\cup \ell')$.
\end{property}
\begin{property} \label{prop:2}
$|\mathbb{S}| = 0$, $\ell$ and $\ell'$ each have at least one neighbor, all particles in $N(\ell) \setminus \{\ell'\}$ are connected by paths within this set, and all particles in $N(\ell') \setminus \{\ell\}$ are connected by paths within this set.
\end{property}

These properties capture precisely the structure required to maintain particle connectivity and prevent certain new holes from forming.
Additionally, both are symmetric for $\ell$ and $\ell'$, necessary for particle moves to be reversible.
However, they are not so restrictive as to limit the movement of particles and prevent compression from occurring.
We will see that after a burn-in phase to eliminate any holes, moves satisfying these properties suffice to transform any configuration into any other.

We now define our Markov chain $\M$ for compression.
The state space $\Omega$ of $\M$ is the set of all connected configurations of $n$ contracted particles, and the rules and probabilities given in Algorithm $\M$ define the transitions between states.
Later, in Section~\ref{subsec:localalga}, we will show how to view this Markov chain as a local, distributed, asynchronous algorithm $\A$.
Both $\M$ and $\A$ take as input a bias parameter $\lambda > 1$ and begin at an arbitrary connected starting configuration $\sigma_0 \in \Omega$.

\begin{algorithm}
\caption*{{\bf Algorithm $\M$}: Markov Chain for Compression}
\begin{algorithmic}[1]
\Statex {\bf Beginning at any connected configuration $\sigma_0$ of $n$ contracted particles, repeat:}
\State Select particle $P$ uniformly at random from among all particles; let $\ell$ be its location. \label{alg:M:begin}
\State Choose neighboring location $\ell'$ and $q \in (0,1)$ uniformly at random. \label{alg:M:q}
\If {$\ell'$ is unoccupied}
    \State $P$ expands to simultaneously occupy $\ell$ and $\ell'$.\label{alg:M:expand}
    \State Let $e = |N(\ell)|$ be the number of neighbors $P$ had when it was contracted at $\ell$, and let $e' = |N(\ell')|$ be the number of neighbors $P$ would have if it contracts to $\ell'$.
    \If {(1) $e \neq 5$, (2) $\ell$ and $\ell'$ satisfy Property~\ref{prop:1} or Property~\ref{prop:2}, and (3) $q < \lambda^{e' - e}$} \label{alg:M:conds}
    \State $P$ contracts to $\ell'$.\label{alg:M:contractto}
    \Else {} $P$ contracts back to $\ell$.\label{alg:M:contractback}
    \EndIf
\EndIf
\end{algorithmic}
\end{algorithm}

In Markov chain $\M$, note that a constant number of random bits suffice to generate $q$ in Step~\ref{alg:M:q}, as only a constant precision is required (given that $e'-e$ is an integer in $[-3,3]$ and $\lambda$ is a constant).
In Step~\ref{alg:M:conds}, Condition (1) ensures no holes form, Condition (2) ensures the particle system stays connected and $\M$ is eventually ergodic, and Condition (3) ensures the particle moves happen with probabilities such that $\M$ converges to the desired distribution.

In practice, Markov chain $\M$ yields good compression.
We simulated $\M$ for $\lambda = 4$ on $100$ particles that began in a line; the configurations after 1, 2, 3, 4, and 5 million iterations of $\M$ are shown in Fig.~\ref{fig:line100_bias4}.
In Section~\ref{sec:results}, we will rigorously prove that Markov chain $\M$ achieves compression with all but exponentially small probability whenever $\lambda > \cbound$ (Theorem~\ref{thm:compress_alpha}).

\begin{figure}[ht]
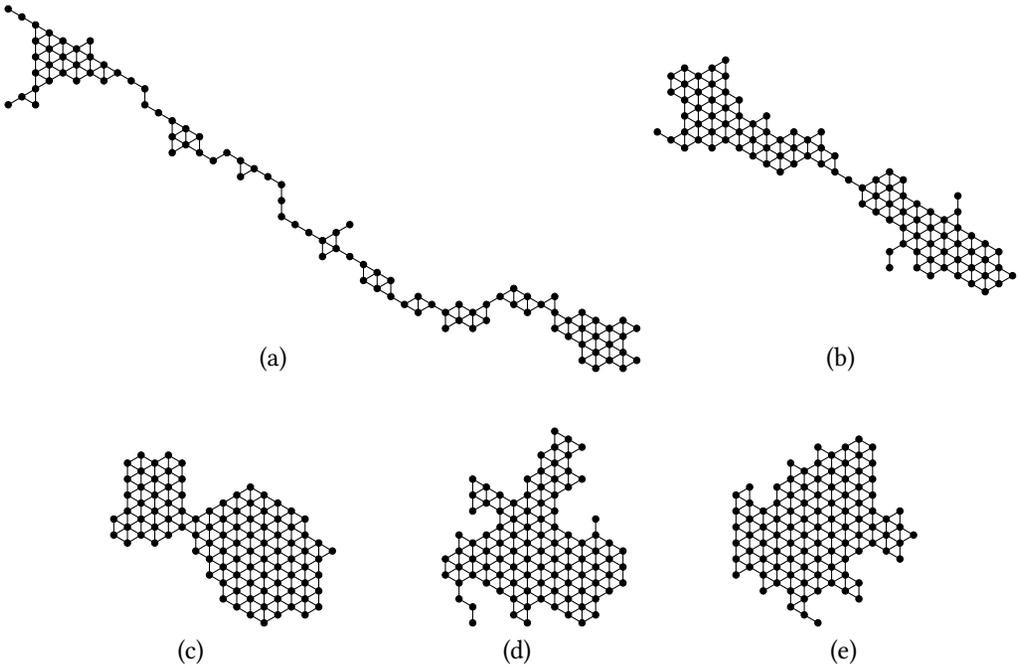

\centering
\begin{tikzpicture}[scale=0.21]
\draw[fill](0,22.5)circle(0.2);
\draw(0,22.5)--(0.866025,22);
\draw[fill](17.3205,10.5)circle(0.2);
\draw(17.3205,10.5)--(17.3205,11.5);
\draw[fill](4.33013,20)circle(0.2);
\draw(4.33013,20)--(5.19615,19.5);
\draw(4.33013,20)--(5.19615,20.5);
\draw[fill](2.59808,21)circle(0.2);
\draw(2.59808,21)--(3.4641,20.5);
\draw[fill](12.1244,13.5)circle(0.2);
\draw(12.1244,13.5)--(12.9904,13);
\draw(12.1244,13.5)--(12.1244,14.5);
\draw[fill](1.73205,20.5)circle(0.2);
\draw(1.73205,20.5)--(2.59808,20);
\draw(1.73205,20.5)--(2.59808,21);
\draw(1.73205,20.5)--(1.73205,21.5);
\draw[fill](25.9808,3.5)circle(0.2);
\draw(25.9808,3.5)--(26.8468,4);
\draw(25.9808,3.5)--(25.9808,4.5);
\draw[fill](1.73205,16.5)circle(0.2);
\draw(1.73205,16.5)--(1.73205,17.5);
\draw[fill](16.4545,12)circle(0.2);
\draw(16.4545,12)--(17.3205,11.5);
\draw[fill](1.73205,21.5)circle(0.2);
\draw(1.73205,21.5)--(2.59808,21);
\draw[fill](5.19615,19.5)circle(0.2);
\draw(5.19615,19.5)--(6.06218,19);
\draw(5.19615,19.5)--(5.19615,20.5);
\draw[fill](3.4641,20.5)circle(0.2);
\draw(3.4641,20.5)--(4.33013,20);
\draw[fill](6.9282,18.5)circle(0.2);
\draw(6.9282,18.5)--(7.79423,18);
\draw[fill](5.19615,18.5)circle(0.2);
\draw(5.19615,18.5)--(6.06218,18);
\draw(5.19615,18.5)--(6.06218,19);
\draw(5.19615,18.5)--(5.19615,19.5);
\draw[fill](0,16.5)circle(0.2);
\draw(0,16.5)--(0.866025,17);
\draw[fill](2.59808,20)circle(0.2);
\draw(2.59808,20)--(3.4641,19.5);
\draw(2.59808,20)--(3.4641,20.5);
\draw(2.59808,20)--(2.59808,21);
\draw[fill](7.79423,18)circle(0.2);
\draw(7.79423,18)--(8.66025,17.5);
\draw[fill](0.866025,17)circle(0.2);
\draw(0.866025,17)--(1.73205,16.5);
\draw(0.866025,17)--(1.73205,17.5);
\draw[fill](2.59808,18)circle(0.2);
\draw(2.59808,18)--(3.4641,18.5);
\draw(2.59808,18)--(2.59808,19);
\draw[fill](5.19615,20.5)circle(0.2);
\draw[fill](10.3923,14.5)circle(0.2);
\draw(10.3923,14.5)--(11.2583,14);
\draw(10.3923,14.5)--(11.2583,15);
\draw(10.3923,14.5)--(10.3923,15.5);
\draw[fill](3.4641,18.5)circle(0.2);
\draw(3.4641,18.5)--(4.33013,18);
\draw(3.4641,18.5)--(4.33013,19);
\draw(3.4641,18.5)--(3.4641,19.5);
\draw[fill](1.73205,17.5)circle(0.2);
\draw(1.73205,17.5)--(2.59808,18);
\draw(1.73205,17.5)--(1.73205,18.5);
\draw[fill](20.7846,7.5)circle(0.2);
\draw(20.7846,7.5)--(21.6506,7);
\draw(20.7846,7.5)--(20.7846,8.5);
\draw[fill](2.59808,19)circle(0.2);
\draw(2.59808,19)--(3.4641,18.5);
\draw(2.59808,19)--(3.4641,19.5);
\draw(2.59808,19)--(2.59808,20);
\draw[fill](1.73205,18.5)circle(0.2);
\draw(1.73205,18.5)--(2.59808,18);
\draw(1.73205,18.5)--(2.59808,19);
\draw(1.73205,18.5)--(1.73205,19.5);
\draw[fill](9.52628,16)circle(0.2);
\draw(9.52628,16)--(10.3923,15.5);
\draw[fill](12.1244,14.5)circle(0.2);
\draw[fill](19.0526,8.5)circle(0.2);
\draw(19.0526,8.5)--(19.9186,8);
\draw[fill](3.4641,19.5)circle(0.2);
\draw(3.4641,19.5)--(4.33013,19);
\draw(3.4641,19.5)--(4.33013,20);
\draw(3.4641,19.5)--(3.4641,20.5);
\draw[fill](14.7224,13)circle(0.2);
\draw(14.7224,13)--(15.5885,12.5);
\draw[fill](8.66025,17.5)circle(0.2);
\draw[fill](0.866025,22)circle(0.2);
\draw(0.866025,22)--(1.73205,21.5);
\draw[fill](4.33013,18)circle(0.2);
\draw(4.33013,18)--(5.19615,18.5);
\draw(4.33013,18)--(4.33013,19);
\draw[fill](4.33013,19)circle(0.2);
\draw(4.33013,19)--(5.19615,18.5);
\draw(4.33013,19)--(5.19615,19.5);
\draw(4.33013,19)--(4.33013,20);
\draw[fill](1.73205,19.5)circle(0.2);
\draw(1.73205,19.5)--(2.59808,19);
\draw(1.73205,19.5)--(2.59808,20);
\draw(1.73205,19.5)--(1.73205,20.5);
\draw[fill](8.66025,16.5)circle(0.2);
\draw(8.66025,16.5)--(9.52628,16);
\draw(8.66025,16.5)--(8.66025,17.5);
\draw[fill](15.5885,12.5)circle(0.2);
\draw(15.5885,12.5)--(16.4545,12);
\draw[fill](19.9186,7)circle(0.2);
\draw(19.9186,7)--(20.7846,7.5);
\draw(19.9186,7)--(19.9186,8);
\draw[fill](18.1865,9)circle(0.2);
\draw(18.1865,9)--(19.0526,8.5);
\draw[fill](6.06218,18)circle(0.2);
\draw(6.06218,18)--(6.9282,18.5);
\draw(6.06218,18)--(6.06218,19);
\draw[fill](24.2487,4.5)circle(0.2);
\draw(24.2487,4.5)--(25.1147,4);
\draw(24.2487,4.5)--(24.2487,5.5);
\draw[fill](10.3923,15.5)circle(0.2);
\draw(10.3923,15.5)--(11.2583,15);
\draw[fill](10.3923,13.5)circle(0.2);
\draw(10.3923,13.5)--(11.2583,14);
\draw(10.3923,13.5)--(10.3923,14.5);
\draw[fill](36.3731,3.5)circle(0.2);
\draw(36.3731,3.5)--(37.2391,3);
\draw[fill](17.3205,11.5)circle(0.2);
\draw[fill](20.7846,8.5)circle(0.2);
\draw(20.7846,8.5)--(21.6506,9);
\draw[fill](6.06218,19)circle(0.2);
\draw(6.06218,19)--(6.9282,18.5);
\draw[fill](12.9904,13)circle(0.2);
\draw(12.9904,13)--(13.8564,13.5);
\draw[fill](30.3109,3)circle(0.2);
\draw(30.3109,3)--(30.3109,4);
\draw[fill](22.5167,6.5)circle(0.2);
\draw(22.5167,6.5)--(23.3827,6);
\draw[fill](11.2583,15)circle(0.2);
\draw(11.2583,15)--(12.1244,14.5);
\draw[fill](13.8564,13.5)circle(0.2);
\draw(13.8564,13.5)--(14.7224,13);
\draw[fill](36.3731,0.5)circle(0.2);
\draw(36.3731,0.5)--(37.2391,0);
\draw(36.3731,0.5)--(37.2391,1);
\draw(36.3731,0.5)--(36.3731,1.5);
\draw[fill](27.7128,3.5)circle(0.2);
\draw(27.7128,3.5)--(28.5788,3);
\draw(27.7128,3.5)--(28.5788,4);
\draw[fill](24.2487,5.5)circle(0.2);
\draw[fill](33.775,4)circle(0.2);
\draw(33.775,4)--(34.641,3.5);
\draw(33.775,4)--(34.641,4.5);
\draw[fill](37.2391,2)circle(0.2);
\draw(37.2391,2)--(38.1051,1.5);
\draw(37.2391,2)--(38.1051,2.5);
\draw(37.2391,2)--(37.2391,3);
\draw[fill](23.3827,5)circle(0.2);
\draw(23.3827,5)--(24.2487,4.5);
\draw(23.3827,5)--(24.2487,5.5);
\draw(23.3827,5)--(23.3827,6);
\draw[fill](36.3731,2.5)circle(0.2);
\draw(36.3731,2.5)--(37.2391,2);
\draw(36.3731,2.5)--(37.2391,3);
\draw(36.3731,2.5)--(36.3731,3.5);
\draw[fill](38.9711,0)circle(0.2);
\draw(38.9711,0)--(39.8372,0.5);
\draw(38.9711,0)--(38.9711,1);
\draw[fill](38.1051,0.5)circle(0.2);
\draw(38.1051,0.5)--(38.9711,0);
\draw(38.1051,0.5)--(38.9711,1);
\draw(38.1051,0.5)--(38.1051,1.5);
\draw[fill](30.3109,4)circle(0.2);
\draw(30.3109,4)--(31.1769,4.5);
\draw[fill](28.5788,3)circle(0.2);
\draw(28.5788,3)--(29.4449,2.5);
\draw(28.5788,3)--(29.4449,3.5);
\draw(28.5788,3)--(28.5788,4);
\draw[fill](34.641,3.5)circle(0.2);
\draw(34.641,3.5)--(35.507,3);
\draw(34.641,3.5)--(34.641,4.5);
\draw[fill](19.9186,8)circle(0.2);
\draw(19.9186,8)--(20.7846,7.5);
\draw(19.9186,8)--(20.7846,8.5);
\draw[fill](11.2583,14)circle(0.2);
\draw(11.2583,14)--(12.1244,13.5);
\draw(11.2583,14)--(12.1244,14.5);
\draw(11.2583,14)--(11.2583,15);
\draw[fill](21.6506,7)circle(0.2);
\draw(21.6506,7)--(22.5167,6.5);
\draw[fill](14.7224,12)circle(0.2);
\draw(14.7224,12)--(15.5885,12.5);
\draw(14.7224,12)--(14.7224,13);
\draw[fill](27.7128,2.5)circle(0.2);
\draw(27.7128,2.5)--(28.5788,3);
\draw(27.7128,2.5)--(27.7128,3.5);
\draw[fill](38.9711,2)circle(0.2);
\draw(38.9711,2)--(39.8372,2.5);
\draw(38.9711,2)--(38.9711,3);
\draw[fill](21.6506,9)circle(0.2);
\draw[fill](34.641,2.5)circle(0.2);
\draw(34.641,2.5)--(35.507,2);
\draw(34.641,2.5)--(35.507,3);
\draw(34.641,2.5)--(34.641,3.5);
\draw[fill](25.9808,4.5)circle(0.2);
\draw(25.9808,4.5)--(26.8468,4);
\draw[fill](34.641,4.5)circle(0.2);
\draw[fill](38.1051,1.5)circle(0.2);
\draw(38.1051,1.5)--(38.9711,1);
\draw(38.1051,1.5)--(38.9711,2);
\draw(38.1051,1.5)--(38.1051,2.5);
\draw[fill](37.2391,3)circle(0.2);
\draw(37.2391,3)--(38.1051,2.5);
\draw[fill](29.4449,2.5)circle(0.2);
\draw(29.4449,2.5)--(30.3109,3);
\draw(29.4449,2.5)--(29.4449,3.5);
\draw[fill](32.0429,5)circle(0.2);
\draw(32.0429,5)--(32.909,4.5);
\draw[fill](32.909,4.5)circle(0.2);
\draw(32.909,4.5)--(33.775,4);
\draw[fill](31.1769,4.5)circle(0.2);
\draw(31.1769,4.5)--(32.0429,4);
\draw(31.1769,4.5)--(32.0429,5);
\draw[fill](17.3205,9.5)circle(0.2);
\draw(17.3205,9.5)--(18.1865,9);
\draw(17.3205,9.5)--(17.3205,10.5);
\draw[fill](26.8468,4)circle(0.2);
\draw(26.8468,4)--(27.7128,3.5);
\draw[fill](38.1051,2.5)circle(0.2);
\draw(38.1051,2.5)--(38.9711,2);
\draw(38.1051,2.5)--(38.9711,3);
\draw[fill](38.9711,3)circle(0.2);
\draw(38.9711,3)--(39.8372,2.5);
\draw[fill](39.8372,2.5)circle(0.2);
\draw[fill](22.5167,5.5)circle(0.2);
\draw(22.5167,5.5)--(23.3827,5);
\draw(22.5167,5.5)--(23.3827,6);
\draw(22.5167,5.5)--(22.5167,6.5);
\draw[fill](25.1147,4)circle(0.2);
\draw(25.1147,4)--(25.9808,3.5);
\draw(25.1147,4)--(25.9808,4.5);
\draw[fill](28.5788,4)circle(0.2);
\draw(28.5788,4)--(29.4449,3.5);
\draw[fill](35.507,2)circle(0.2);
\draw(35.507,2)--(36.3731,1.5);
\draw(35.507,2)--(36.3731,2.5);
\draw(35.507,2)--(35.507,3);
\draw[fill](32.0429,4)circle(0.2);
\draw(32.0429,4)--(32.909,3.5);
\draw(32.0429,4)--(32.909,4.5);
\draw(32.0429,4)--(32.0429,5);
\draw[fill](29.4449,3.5)circle(0.2);
\draw(29.4449,3.5)--(30.3109,3);
\draw(29.4449,3.5)--(30.3109,4);
\draw[fill](35.507,3)circle(0.2);
\draw(35.507,3)--(36.3731,2.5);
\draw(35.507,3)--(36.3731,3.5);
\draw[fill](36.3731,1.5)circle(0.2);
\draw(36.3731,1.5)--(37.2391,1);
\draw(36.3731,1.5)--(37.2391,2);
\draw(36.3731,1.5)--(36.3731,2.5);
\draw[fill](38.9711,1)circle(0.2);
\draw(38.9711,1)--(39.8372,0.5);
\draw(38.9711,1)--(38.9711,2);
\draw[fill](23.3827,6)circle(0.2);
\draw(23.3827,6)--(24.2487,5.5);
\draw[fill](37.2391,1)circle(0.2);
\draw(37.2391,1)--(38.1051,0.5);
\draw(37.2391,1)--(38.1051,1.5);
\draw(37.2391,1)--(37.2391,2);
\draw[fill](32.909,3.5)circle(0.2);
\draw(32.909,3.5)--(33.775,4);
\draw(32.909,3.5)--(32.909,4.5);
\draw[fill](39.8372,0.5)circle(0.2);
\draw[fill](37.2391,0)circle(0.2);
\draw(37.2391,0)--(38.1051,0.5);
\draw(37.2391,0)--(37.2391,1);
\end{tikzpicture}
\input{line100_bias4_2000000.txt}
\vspace{-5mm}

\ \ \ \ \ \ \ \ \ \ \ \ (a) \hspace{7cm} (b)

\vspace{7mm}
\input{line100_bias4_3000000.txt} \ \ \ \ \ \ \ \ \ \ \ \ \ \
\input{line100_bias4_4000000.txt} \ \ \ \ \ \ \ \ \ \ \ \ \ \
\input{line100_bias4_5000000.txt}

(c) \hspace{3.8cm} (d) \hspace{3.8cm} (e)
\caption{$100$ particles in a line with edges drawn, after (a) $1$ million, (b) $2$ million, (c) $3$ million, (d) $4$ million, and (e) $5$ million iterations of $\M$ with bias $\lambda = 4$.}
\label{fig:line100_bias4}
\end{figure}

\subsection{The Local Algorithm \texorpdfstring{$\A$}{A}} \label{subsec:localalga}

In order for each particle to individually run $\M$, a Markov chain with centralized control, we show how $\M$ can be translated into a fully distributed, local, asynchronous algorithm $\A$ that satisfies the constraints of the amoebot model (Section~\ref{subsec:model}).
There are two parts of this translation: $(i)$ selecting particles uniformly at random as in Step~\ref{alg:M:begin} of $\M$ must be translated to asynchronous activations of individual particles, and $(ii)$ moving particles in a combined expansion and contraction as in Steps~\ref{alg:M:expand}--\ref{alg:M:contractback} of $\M$ must be decoupled into two separate activations, since the amoebot model allows at most one movement per activation.
All other steps of $\M$ can be directly implemented by an individual particle with constant-size memory using only information from its local neighborhood.

Choosing a particle at random in Step~\ref{alg:M:begin} of $\M$ enables us to explicitly calculate the stationary distribution of $\M$ so that we can provide rigorous guarantees about its structure.
However, under the usual models of asynchronous systems from distributed computing, one cannot assume that the next particle to be activated is equally likely to be any particle.
To mimic this uniformly random activation sequence in a local way, we assume each particle has its own Poisson clock with mean $1$ and activates after a delay $t$ drawn with probability $e^{-t}$.
After completing its activation (executing Algorithm $\A$), a new delay is drawn to its next activation, and so on.
The exponential distribution guarantees that, regardless of which particle has just activated, all particles are equally likely to be the next to activate (see, e.g.,~\cite{Feller1968}).
Moreover, particles proceed without requiring knowledge of any other particles' clocks.
Similar Poisson clocks are commonly used to describe physical systems that perform concurrent updates in continuous time.

We could even better approximate asynchronous activation sequences by allowing each particle to have its own constant mean for its Poisson clock, allowing for some particles to activate more often than others in expectation.
In this setting, the probability that a given particle $P$ is the next of the $n$ particles to activate is not $1/n$, but rather some probability $a_P$ that depends on all particles' Poisson means.\footnote{Probability $a_P$ would only play a role in the analysis of $\M$ and $\A$, not in their execution. Particle $P$ does not need to know or calculate $a_P$.}
This does not change the stationary distribution of $\M$ (i.e., Lemma~\ref{lem:stat-edge} still holds with a nearly identical proof that replaces $1/n$ with $a_P$), and our main results (Theorem~\ref{thm:compress_alpha} and Corollary~\ref{cor:compress_lambda}) still follow.
Because the same results hold regardless of the rates of particles' Poisson clocks, we assume clocks with mean $1$ for simplicity.
Moreover, though Poisson activation sequences are necessary for our rigorous results, we do not expect the system's behavior would be substantially different for non-Poisson activation sequences.

\begin{algorithm}
\caption*{{\bf Algorithm $\A$}: Local, Distributed, Asynchronous Algorithm for Compression run by Particle $P$}
\begin{algorithmic}[1]
\Statex {\bf If $P$ is contracted:}
\State Let $\ell$ denote $P$'s current location. \label{alg:A:contractbegin}
\State Choose a neighboring location $\ell'$ uniformly at random from the six possible choices. \label{alg:start}
\If {$\ell'$ is unoccupied and $P$ has no expanded neighbors} \label{alg:A:expandcond}
    \State $P$ expands to simultaneously occupy $\ell$ and $\ell'$.
    \If {there are no expanded particles adjacent to $\ell$ or $\ell'$}
        \State $P$ sets $flag=\textsc{True}$ in its local memory. \label{alg:A:flagtrue}
    \Else {} $P$ sets $flag=\textsc{False}$. \label{alg:A:flagfalse}
    \EndIf
\EndIf
\Statex {\ \ }
\Statex {\bf If $P$ is expanded:}
\State Choose $q \in (0,1)$ uniformly at random. \label{alg:A:expandbegin}
\State Let $N^*(\cdot) \subseteq N(\cdot)$ be the set of neighboring particles excluding any heads of expanded particles.
\State Let $e = |N^*(\ell)|$ be the number of neighbors $P$ had when it was contracted at $\ell$, and let $e' = |N^*(\ell')|$ be the number of neighbors $P$ would have if it contracts to $\ell'$.
\If {(1) $e \neq 5$, (2) locations $\ell$ and $\ell'$ satisfy Property~\ref{prop:1} or Property~\ref{prop:2} with respect to $N^*(\cdot)$, (3) $q < \lambda^{e'-e}$, and (4) $flag = \textsc{True}$} \label{alg:A:conds}
    \State $P$ contracts to $\ell'$.
\Else {} $P$ contracts back to $\ell$. \label{alg:A:expandend}
\EndIf
\end{algorithmic}
\end{algorithm}

We now turn to decoupling the combined expansion and contraction movement in a single state transition of $\M$ into two (not necessarily consecutive) activations of a given particle running algorithm $\A$.
We must carefully handle the way in which a particle's neighborhood may change between its two activations, ensuring that at most one particle per neighborhood moves at a time, mimicking the sequential nature of $\M$.
Each particle $P$ continuously runs Algorithm $\A$, executing Steps~\ref{alg:A:contractbegin}--\ref{alg:A:flagfalse} if contracted, and Steps~\ref{alg:A:expandbegin}--\ref{alg:A:expandend} if expanded.
Conditions (1)--(3) in Step~\ref{alg:A:conds} of $\A$ are analogous to those in Step~\ref{alg:M:conds} of $\M$, but treat expanded particles as if they are still contracted at their tail location, rather than considering all occupied neighboring locations.
We use the additional Condition (4) to ensure $P$ is the only particle in its neighborhood moving to a new position since it last expanded, as we now explain in more detail.
For the purposes of this analysis, recall from Section~\ref{subsec:model} that although Algorithm $\A$ is executed concurrently by all particles, we can view the system's progress as an equivalent sequence of particles' atomic actions.

Suppose a particle $P$ eventually moves from location $\ell$ to location $\ell'$ by expanding to occupy both positions at some time $t$ and contracting to $\ell'$ at some time $t' > t$ according to an execution of $\A$.
Since $P$ eventually completes its movement to $\ell'$, there must have been no expanded particles adjacent to $\ell$ or $\ell'$ at time $t$ (by Step~\ref{alg:A:flagtrue} and Condition (4) of Step~\ref{alg:A:conds} in $\A$).
Any other particle $Q$ that expands into the neighborhood of $P$ in the time interval $(t,t')$ will see that $P$ is expanded and set its flag to \textsc{False} in Step~\ref{alg:A:flagfalse} of $\A$.
Recall from Section~\ref{subsec:model} that a particle can differentiate between a neighbor's head and tail.
Since any such neighbor $Q$ with a \textsc{False} flag must contract back to its original position during its next activation (by Condition (4) of Step~\ref{alg:A:conds} and Step~\ref{alg:A:expandend} of $\A$), particle $P$ can safely ignore any expanded heads in its neighborhood, making decisions in Steps~\ref{alg:A:expandbegin}--\ref{alg:A:conds} of $\A$ as if $Q$ had never moved.
Thus, the neighborhood of $P$ remains effectively undisturbed in the interval $(t,t')$, allowing $\A$ to faithfully emulate $\M$.

Any objective that can be accomplished by $\M$ can be accomplished by $\A$ and vice versa.
Consider an activation sequence of particles executing $\A$ that transforms the initial configuration $\sigma_0$ to a configuration $\sigma'$ that potentially contains both expanded and contracted particles.
Obtain configuration $\sigma$ from $\sigma'$ by preserving the locations of all contracted particles and considering every expanded particle to be contracted at its tail.
Then there exists a sequence of transitions in $\M$ that reaches $\sigma$.
The perimeter $p(\sigma')$ ignores heads of expanded particles (Section~\ref{subsec:terminology}), so $p(\sigma) = p(\sigma')$.
Conversely, every sequence of transitions in $\M$ that reaches a configuration $\sigma$ directly corresponds to a sequence of atomic actions (expansions followed immediately by contractions) of particles executing $\A$ also leading to $\sigma' = \sigma$, where again $p(\sigma) = p(\sigma')$.
Thus, proving $\alpha$-compression for $\sigma$ also implies $\alpha$-compression for $\sigma'$, and vice-versa.
Hence, we can use $\M$ and respective Markov chain tools and techniques in order to analyze the correctness of algorithm $\A$.
Because we show $\alpha$-compression for $\M$ for all $\alpha > 1$ (Theorem~\ref{thm:compress_alpha}), this also then implies $\alpha$-compression for $\A$ for all $\alpha > 1$.
In subsequent sections, we focus on analyzing $\M$.

We have shown our Markov chain $\M$ can be translated into a fully distributed, local, asynchronous algorithm $\A$ with the same behavior, but such implementations are not always possible in general.
Any Markov chain for particle systems that relies on non-local particle moves or has transition probabilities that depend on non-local information cannot be executed by a distributed, local algorithm.
Moreover, many algorithms under the amoebot model are not stochastic and thus cannot be meaningfully described as Markov chains (see Sections 3--4 of~\cite{Daymude2019}).

\subsection{Obliviousness and Robustness of \texorpdfstring{$\M$}{M} and \texorpdfstring{$\A$}{A}} \label{subsec:oblivious+robust}

Our algorithm for compression has two key advantages over previous algorithms for self-organizing particle systems: inherent \emph{obliviousness} and \emph{robustness}.
An algorithm is \emph{oblivious} if it is stateless; i.e., a particle remembers no information from past activations and decides what to do based only on its observations of its current environment.
In practical settings, such algorithms are desirable because they do not require persistent memory and are often self-stabilizing and fault-tolerant (see, e.g., obliviousness in mobile robots~\cite{Flocchini2019}); theoretically, they are of great interest because they are computationally weak at an individual level but can still collectively accomplish sophisticated goals.
Algorithm $\A$ for compression is the first nearly oblivious algorithm for self-organizing particle systems, as each particle only needs to store its $flag$ variable as a single bit of information between its expansion and contraction activations.
Previous works under the amoebot model (see, e.g., Sections 3--4 of~\cite{Daymude2019}), however, rely heavily on persistent particle memory for decision making and communication.

Our algorithm is also the first for self-organizing particle systems to meaningfully consider fault-tolerance.\footnote{After our compression algorithm was first introduced as~\cite{Cannon2016}, fault-tolerance for self-organizing particle systems was also considered for shape formation problems in~\cite{DiLuna2018}.}
A distributed algorithm's \emph{fault-tolerance} has to do with its ability to achieve its goals despite possible \emph{crash failures} or \emph{Byzantine failures}.
In a crash failure, an agent abruptly ceases functioning and may never be resuscitated.
These failures are particularly problematic for systems with a single point of failure, as there is no guarantee the critical agent will remain non-faulty nor that its memory and role could be assumed by another agent if it crashes.
In a Byzantine failure, some fraction of the agents are malicious and execute arbitrary behavior in an effort to stop the non-faulty portion of the system from achieving its task.

Before we introduced our compression algorithm, work on self-organizing particle systems had not addressed either type of possible fault, and many of the proposed algorithms were susceptible to complete failure if even a single particle crashed.
If one or more particles were to crash in our algorithm for compression, they would cease moving and act as fixed points around which the remaining particles would simply continue to compress.
For the more adversarial setting of Byzantine failures, since our algorithm is (nearly) oblivious and communication is limited to particles checking the flags of their neighbors, the malicious particles are unable to ``lie'' or otherwise corrupt healthy particles' behaviors.
We speculate that the malicious particles could affect the overall compression of the system by expanding away from where the system is aggregating and refusing to contract, essentially acting as fixed points.
However, if the fraction of malicious particles is small, the non-faulty particles will still be able to compress around the malicious particles, as in crash failures.
\bluecomment{JD: I've been thinking about this a bit more. I do think that $\A$ could be significantly impeded by even a few malicious particles expanding, setting their flags to true, and then never contracting (imagine on a line if the two particles at either end were malicious; such expansions would stop any system progress towards compression).}
\redcomment{SC: That's a good point. Maybe we want to say something like at distance $\geq 2$ from any Byzantine failure, compression occurs as normal? If two byzantine particles manage to freeze the particles in a stretched out configuration, like a line, the rest of the particles still compress as best possible given these constraints (kind of like crash failures)? So progress is impeded, but not really any worse than if there were just adversarial crashes. A particle freezing itself vs. freezing its neighborhood isn't really a big difference, I don't think}
\bluecomment{JD: Edited lightly to include some of this discussion.}

\subsection{Invariants for Markov Chain \texorpdfstring{$\M$}{M}} \label{subsec:invariants}

Now that we have described and discussed algorithm $\A$ and shown that it is a distributed implementation of Markov chain $\M$, we will perform the rest of our analysis directly on $\M$.
We begin by showing that $\M$ maintains certain invariants.

\begin{lem} \label{lem:invariant_connected}
	If the particle system is initially connected, during the execution of Markov chain $\M$ it remains connected.
\end{lem}
\begin{proof}
	Consider one iteration of $\M$ where a particle $P$ moves from location $\ell$ to location~$\ell'$.
	Let~$\sigma$ be the configuration before this move, and $\sigma'$ the configuration after. We show if~$\sigma$ is connected, then so is $\sigma'$.

	A move of particle $P$ from $\ell$ to $\ell'$ occurs only if $\ell$ and $\ell'$ are adjacent and satisfy Property~\ref{prop:1} or Property~\ref{prop:2}.
	First, suppose they satisfy Property~\ref{prop:1}.
	If $\sigma$ is connected, then for every particle $Q$ there exists some path $\mathcal{P} = (P = P_1, P_2, \ldots, P_k = Q)$ from $P$ to $Q$ in $\sigma$.
	By Property~\ref{prop:1}, since $P_2 \in N(\ell)$, there exists a path from $P_2$ to a particle $S \in \mathbb{S}$ that is entirely contained in $N(\ell)$.
	After $P$ moves to location $\ell'$, it remains connected to particle $Q$ by a (not necessarily simple) walk that first travels to $S$, then travels through $N(\ell)$ to $P_2$, and finally follows $\mathcal{P}$ to $Q$.
	This implies $P$ is connected to all particles from location $\ell'$, so $\sigma'$ is connected via paths through $P$.

	Next, assume locations $\ell$ and $\ell'$ satisfy Property~\ref{prop:2}.
	Let $Q$ and $\overline{Q} \neq P$ be particles; we show that if $\sigma$ is connected, then $Q$ and $\overline{Q}$ must be connected by a path not containing $P$.
	If $\sigma$ is connected, then $Q$ and $\overline{Q}$ are connected by some path $\mathcal{P} = (Q = Q_1, Q_2, \ldots, Q_k = \overline{Q})$.
	If $P$ is not in this path we are done, so suppose this path contains $P$, that is, $Q_i = P$ for some $i \in \{2, \ldots, k-1\}$.
	Both $Q_{i-1}$ and $Q_{i+1}$ are neighbors of $\ell$, and by Property~\ref{prop:2} all neighbors of $\ell$ are connected by a path in $N(\ell)$.
	Thus $\mathcal{P}$ can be augmented to form a (not necessarily simple) walk $\mathcal{W}$ by replacing $P$ with a path from $Q_{i-1}$ to $Q_{i+1}$ in $N(\ell)$.
	As $P \not\in \mathcal{W}$, this walk connects $Q$ and $\overline{Q}$ in $\sigma'$ without going through $P$, as desired.
	Because any two particles $Q, \overline{Q} \neq P$ are connected by a path not containing $P$, they remain connected after $P$ moves from $\ell$ to $\ell'$.
	Additionally, because $\ell'$ has at least one neighbor by Property~\ref{prop:2}, $P$ at location $\ell'$ is connected to at least one particle, and via that particle to all other particles in $\sigma'$.
	Thus $\sigma'$ is connected.
\end{proof}

We prove in the next subsection that $\M$ will eventually reach a configuration with no holes (Lemma~\ref{lem:eventually_no_holes}). After that point, the following lemma will apply. While it is true more broadly that $\M$ will never create new holes, we prove only what we will need, that new holes are never created in a hole-free configuration.

\begin{lem} \label{lem:invariant_holefree}
If Markov chain $\M$ reaches a connected configuration with no holes, then all subsequent configurations reached during the execution of $\mathcal{M}$ will not have holes.
\end{lem}
\begin{proof}
	Consider one iteration of $\M$ where a particle $P$ moves from location $\ell$ to location~$\ell'$.
	Let~$\sigma$ be the configuration before this move, and $\sigma'$ the configuration after.
	We show if $\sigma$ is hole-free, then so is $\sigma'$.

	By a {\it cycle} in a configuration $\sigma$ we will mean a cycle in $\Gtri$ that surrounds at least one unoccupied location and whose vertices are occupied by particles of $\sigma$.
	Note a configuration has a hole if and only if it has a cycle.
	Throughout this proof, we will argue about the existence of cycles rather than the existence of holes.

	We first show that if $\sigma'$ has a cycle then that cycle must contain $P$.
	Suppose, for the sake of contradiction, this is not the case and $\sigma'$ has a cycle $\mathcal{C}$ with $P \not\in \mathcal{C}$.
	If $P$ is removed from location $\ell'$, then cycle $\mathcal{C}$ still exists in $\sigma' - P$.
	If $P$ is then placed at $\ell$, yielding $\sigma$, then $\mathcal{C}$ still exists unless it had enclosed exactly one unoccupied location, $\ell$.
	However, this is not possible as any cycle in $\sigma' - P$ encircling $\ell$ would also necessarily encircle neighboring unoccupied location~$\ell'$.
	This implies cycle $\mathcal{C}$ exists in cycle-free configuration~$\sigma$, a contradiction.
	We conclude any cycle in $\sigma'$ must contain~$P$.

	Because particle $P$ moved from location $\ell$ to location $\ell'$ in a valid step of Markov chain $\M$, it must be true (by the conditions checked in Step~\ref{alg:M:conds} of $\M$) that $\ell$ has fewer than five neighbors and locations $\ell$ and $\ell'$ satisfy Property~\ref{prop:1} or Property~\ref{prop:2}.

	First, suppose they satisfy Property~\ref{prop:2}.
	While $P$ might momentarily create a cycle when it expands to occupy both locations $\ell$ and $\ell'$, it will then contract to location $\ell'$.
	Suppose $P$ is part of some cycle $\mathcal{C} = (P = P_1, P_2, \ldots, P_{k-1}, P_k = P)$ in $\sigma'$.
	By Property~\ref{prop:2}, $P_2$ and $P_{k-1}$ are connected by a path in $N(\ell')$ that doesn't contain $P$.
	Replacing path $(P_{k-1}, P, P_2)$ in cycle $\mathcal{C}$ by this path in $N(\ell')$ yields a (not necessarily simple) cycle $\mathcal{C}'$ in $\sigma'$ not containing $P$, a contradiction.

	Next, suppose $\ell$ and $\ell'$ satisfy Property 1.
	Because particle $P$ moved from $\ell$ to $\ell'$ in a valid step of $\M$, location $\ell$ must have at most four neighbors in $\sigma$.
	This means that in $\sigma'$, location $\ell$ has at most five neighbors --- its original neighbors plus $P$ at location $\ell'$ --- and thus is adjacent to at least one unoccupied location.
	Suppose there exists some cycle ${\mathcal{C} = (P = P_1, P_2, \ldots, P_{k-1}, P_k = P)}$ in $\sigma'$.
	This cycle encircles at least one unoccupied location $\ell'' \neq \ell$: since $\ell$ is adjacent to another unoccupied location in $\sigma'$, it cannot be the case that $\ell$ is the only unoccupied location inside $\mathcal{C}$.
	If there exists a path between $P_2$ and $P_{k-1}$ in $N(\ell')$, the argument from the previous case applies and we are done.
	Otherwise, without loss of generality, it must be that $|\mathbb{S}| = 2$ and there exist paths in $N(\ell \cup \ell')$ from $P_{k-1}$ to $S_1 \in \mathbb{S}$ and from $P_2$ to $S_2 \in \mathbb{S}$, with $S_1 \neq S_2$.
	There then exists a (not necessarily simple) cycle $\mathcal{C}^*$ in $\sigma$ obtained from $\mathcal{C}$ by replacing path $(P_{k-1}, P, P_2)$, where $P$ is in location $\ell'$, with path $(P_{k-1}, \ldots, S_1, P, S_2, \ldots, P_2)$, where $P$ is in location $\ell$.
	$\mathcal{C}^*$ is a valid cycle in $\sigma$ because it encircles unoccupied location $\ell'' \neq \ell$.
	This is a contradiction because $\sigma$ has no cycles.
	We conclude by contradiction that, in all cases, $\sigma'$ must have no cycles, and thus must have no holes.
\end{proof}

\subsection{Eventual Ergodicity of Markov chain \texorpdfstring{$\M$}{M}} \label{subsec:ergodicity}

The state space $\Omega$ of our Markov chain $\M$ is the set of all connected configurations of $n$ contracted particles, and Lemma~\ref{lem:invariant_connected} ensures that we always stay within this state space.
The initial configuration $\sigma_0$ of $\M$ may or may not have holes.
By Lemma~\ref{lem:invariant_holefree}, once a hole-free configuration is reached, $\M$ remains in the part of the state space consisting of all hole-free connected configurations, which we call $\Ohf$.
In this section, we prove that from any starting state $\M$ always reaches $\Ohf$.
Furthermore, we prove that $\M$ is irreducible on $\Ohf$, that is, for any two configurations in $\Ohf$ there exists some sequence of moves between them that has positive probability.
Stated another way,
what we show is that
all states in $\Ohf$ are {\it recurrent}, meaning once $\M$ reaches a state $\sigma \in \Ohf$ it returns to $\sigma$ with probability 1, while all states in $\Omega \setminus \Ohf$ are {\it transient}, meaning they are not recurrent.
As $\M$ is also aperiodic, we can conclude it is eventually ergodic on $\Ohf$, a necessary precondition for all of the Markov chain analysis to follow.


We note the details of these proofs have been substantially simplified and clarified from the originally published conference version of these results~\cite{Cannon2016}, where the proof of ergodicity required over 10 pages of detailed analysis.
Fig.~\ref{fig:jumps_only} illustrates one difficulty.
It depicts a hole-free particle configuration for which there exist no valid moves satisfying Property~\ref{prop:1}; the only valid moves satisfy Property~\ref{prop:2}.
Thus if moves satisfying Property~\ref{prop:2} are not included, neither $\Omega$ nor $\Ohf$ is connected.

\begin{figure}
  \centering
  \includegraphics[scale = 0.5, trim =68 454 159 99, clip]{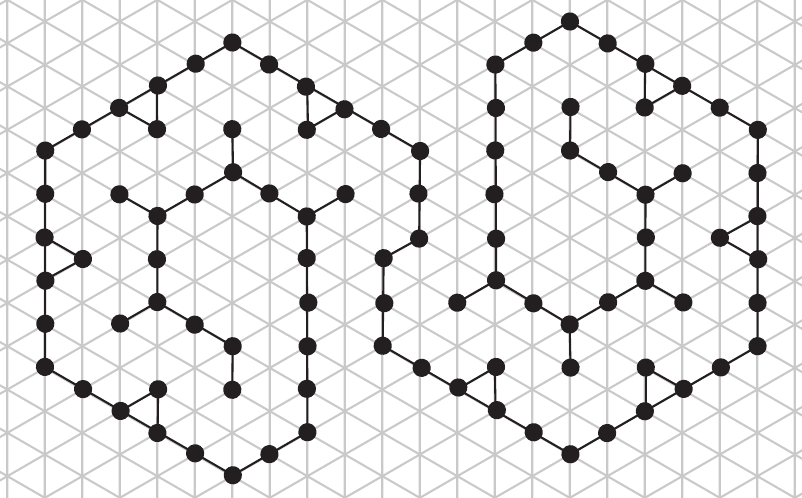}
  \caption{A particle configuration for which all valid moves of Markov chain $\M$ satisfy Property~2; no particle has a valid move satisfying Property 1. This demonstrates the subtlety of the Markov chain rules we have defined.}
  \label{fig:jumps_only}
\end{figure}

At a high level, we prove that for any configuration $\sigma$ there exists a sequence of valid particle moves transforming $\sigma$ into a straight line.
Since a straight line is hole-free, this shows that from any initial configuration in $\Omega$, there exists a sequence of moves with non-zero probability reaching $\Ohf$, as desired.
We then prove any moves of $\M$ among states of $\Ohf$ are reversible, which
implies that for any $\tau \in \Ohf$ there exists a sequence of valid particle moves transforming a straight line into $\tau$.
Altogether, this shows for any $\sigma,\tau \in \Ohf$ there exists a sequence of valid moves (within $\Ohf$) transforming any $\sigma$ into any $\tau$, as required for ergodicity.

We will let $\lin_1$ be the vertical lattice line containing the leftmost particle(s) in $\sigma$.
We label the subsequent vertical lattice lines as $\lin_2, \lin_3, \lin_4$, and so on.
The process for moving the particles into one straight line is a sweep line algorithm, an approach often used in computational geometry~\cite{Shamos1976,Fortune1987}.
We first consider the particles in leftmost vertical line $\lin_1$, then the particles in $\lin_2$, and so on; when considering line $\lin_i$, we maintain the following invariants:

\vspace{2mm}
{\bf Invariants:}
\begin{enumerate}[nolistsep,noitemsep]
  \item All particles left of $\lin_i$ form lines stretching down and left. \label{inv:1} 
  \item Each such line stretches down and left from a particle in $\lin_i$ has an empty location directly below it. \label{inv:2}
\end{enumerate}
\vspace{2mm}

\noindent Fig.~\ref{fig:invariants_i} gives an example of a particle configuration and a line $\lin_i$ satisfying these invariants.
We describe how to, starting in a configuration in which the invariants are satisfied for $\lin_i$, find a sequence of valid particle moves after which $\lin_{i+1}$ satisfies the invariants.
For the configuration in Fig.~\ref{fig:invariants_i}, the configuration obtained after first ensuring $\lin_{i+1}$ satisfies Invariant~\ref{inv:1} is shown in Fig.~\ref{fig:invariants_both}, and the configuration after ensuring $\lin_{i+1}$ also satisfies Invariant~\ref{inv:2} is shown in Fig.~\ref{fig:invariants_i+1}.

\begin{figure}
\centering
\subfloat[]{
    \includegraphics[scale = 0.55, page = 1]{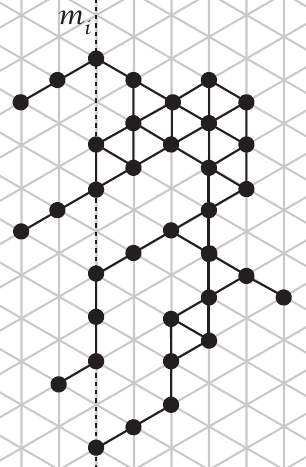}
    \label{fig:invariants_i}
} \hfil
\subfloat[]{
    \includegraphics[scale = 0.55, page = 2]{invariants.pdf}
    \label{fig:invariants_both}
} \hfil
\subfloat[]{
    \includegraphics[scale = 0.55, page = 3]{invariants.pdf}
    \label{fig:invariants_i+1}
}
\caption{(a) An example of a particle configuration and a line $\lin_i$ that satisfies both invariants. (b) After a sequence of moves described in Lemma~\ref{lem:fewerparticles}, $\lin_{i+1}$ satisfies Invariant~\ref{inv:1}. (c) After a sequence of moves described in Lemma~\ref{lem:inv_linestart}, $\lin_{i+1}$ also satisfies Invariant~\ref{inv:2}.}
\label{fig:invariants}
\end{figure}

Throughout this subsection, a {\it component of line $\lin_i$} will refer to a maximal collection of particles in $\lin_i$ that are connected via paths in $\lin_i$.
For example, in Fig.~\ref{fig:invariants_i}, $\lin_i$ has four components (from top to bottom: of one, two, three, and one particles, respectively).
We begin with a lemma about particle movements that will play a key role.

\begin{lem} \label{lem:eliminate1}
Suppose particle $P$ has exactly two neighbors, $Q_1$ below it and $Q_2$ above-right of it, and let $\ell$ be the unoccupied location below-right of $P$.
There exists a sequence of valid moves, occurring strictly below and right of $P$, after which either it is valid for $P$ to move to $\ell$ or some other particle has already moved to $\ell$.
\end{lem}
\begin{proof}
We induct on the number of particles strictly below and right of $P$.
If there are no such particles, then it is valid (satisfying Property~\ref{prop:1}) for $P$ to move from its current location $\ell_0$ to $\ell$.
This is because $N(\ell_0) \cap N(\ell) = \{Q_1, Q_2\}$, and either these are the only two particles in $N(\ell_0 \cup \ell)$ (Fig.~\ref{fig:eliminate_base_1}) or there is exactly one other particle in $N(\ell_0 \cup \ell)$ and it is adjacent to $Q_2$ (Fig.~\ref{fig:eliminate_base_2}).
Thus the conclusions of the lemma are already satisfied with an empty set of moves.

\begin{figure}
\centering
\subfloat[]{
    \includegraphics[scale = 0.8, page = 5]{eliminate.pdf}
    \label{fig:eliminate_base_1}
} \hfil
\subfloat[]{
    \includegraphics[scale = 0.8, page = 6]{eliminate.pdf}
    \label{fig:eliminate_base_2}
} \\
\subfloat[]{
    \includegraphics[scale = 0.8, page = 1]{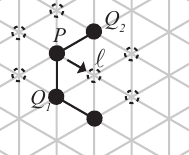}
    \label{fig:eliminate_start}
} \hfil
\subfloat[]{
    \includegraphics[scale = 0.8, page = 2]{eliminate.pdf}
    \label{fig:eliminate_P'}
} \hfil
\subfloat[]{
    \includegraphics[scale = 0.8, page = 3]{eliminate.pdf}
    \label{fig:eliminate_NP'disconn}
} \hfil
\subfloat[]{
    \includegraphics[scale = 0.8, page = 4]{eliminate.pdf}
    \label{fig:eliminate_NP'conn}
}
\caption{Particle positions from the base case (top row) and inductive step (bottom row) of the proof of Lemma \ref{lem:eliminate1}. Particles are represented by black circles, and unoccupied locations are represented by dashed circles. Neighboring particles have a black line drawn between them.}
\label{fig:eliminate}
\end{figure}

Suppose there are $k > 0$ particles strictly below and right of $P$, and for all $0 \leq k' < k$ the lemma holds.
If it is already valid for $P$ to move to $\ell$, we are done; an example is given in Fig.~\ref{fig:eliminate_start}.
Otherwise, since $P$ has fewer than five neighbors, it must be that neither Property~\ref{prop:1} nor Property~\ref{prop:2} is satisfied.
Note $\mathbb{S} = N(P) \cap N(\ell)$ contains two particles, $Q_1$ and $Q_2$.
Because Property~\ref{prop:1} doesn't hold, and $N(P)$ doesn't contain any particles other than those of $\mathbb{S}$, it must be that there is a particle $P'$ in $N(\ell)$ that is not connected to a particle in $\mathbb{S}$ by a path within $N(\ell)$.
Then $P'$ must occupy the location below-right of $\ell$, and the locations adjacent to both $\ell$ and $P'$ must be unoccupied; see Fig.~\ref{fig:eliminate_P'}.
We now consider $N(P')$, which is of size at least one and at most three.

First, we suppose $N(Q')$ is not connected; see Fig.~\ref{fig:eliminate_NP'disconn}.
In this case, $P'$ must have exactly two neighbors, one below $P'$ and the other above-right of $P'$, while location $\ell'$ below-right of $P'$ is unoccupied.
There are fewer than $k$ particles below and right of $P'$ because this is a proper subset of the $k$ particles below and right of $P$.
By the induction hypothesis, we conclude there is a sequence of moves occurring entirely below and right of $P'$ after which either it is valid for $P'$ to move to $\ell'$ or another particle has moved to $\ell'$.
In the first case, we let $P'$ move to $\ell'$ and afterwards it is valid (satisfying Property~\ref{prop:1}) for $P$ to move to $\ell$, because $N(\ell)$ now contains only $Q_1$ and $Q_2$.
In the second case, a particle has moved to $\ell'$ but $N(P')$ otherwise remains unchanged, causing $N(P')$ to now be connected, the case we consider next.

Suppose $N(P')$, which is of size at least one and at most three, is connected; see Fig.~\ref{fig:eliminate_NP'conn}.
Note the current location of $P'$ and location $\ell$ satisfy Property~\ref{prop:2}, so particle $P'$ can move to $\ell$.
As $P'$ and $\ell$ are below and right of $P$, this move satisfies the conclusions of the lemma.
\end{proof}

If $\lin_i$ satisfies the invariants, we want to give a sequence of moves after which $\lin_{i+1}$ also satisfies the invariants.
The following lemma will be used towards that goal.

\begin{lem} \label{lem:fewerparticles}
If $\lin_i$ satisfies both invariants and has a component of size at least two, there exists a sequence of valid moves that decreases the number of particles in $\lin_i$, after which $\lin_i$ still satisfies the invariants.
\end{lem}
\begin{proof}
	Consider any component of $\lin_i$ of size at least two, and let $P$ be the topmost particle in this component.
$P$ has a particle below it, no particle above it, and  by Invariants~\ref{inv:1} and~\ref{inv:2} has no particle above-left or below-left of it.
The two locations right of $P$ may or may not be occupied.
We consider two cases: when $N(P)$ is connected, and when it is not.

\begin{figure}
\centering
\subfloat[]{
    \includegraphics[scale = 0.8, page = 1]{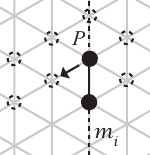}
    \label{fig:NPconn_1}
} \hfil
\subfloat[]{
    \includegraphics[scale = 0.8, page = 2]{Pnbhd.pdf}
    \label{fig:NPconn_2}
} \hfil
\subfloat[]{
    \includegraphics[scale = 0.8, page = 3]{Pnbhd.pdf}
    \label{fig:NPconn_3}
} \hfil
\subfloat[]{
    \includegraphics[scale = 0.8, page = 1]{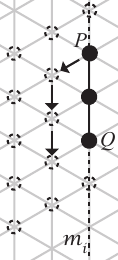}
    \label{fig:movePtoline_1}
} \hfil
\subfloat[]{
    \includegraphics[scale = 0.8, page = 3]{movePtoline.pdf}
    \label{fig:movePtoline_2}
}
\caption{If $P$ is the topmost particle in a component of $\lin_i$ of size at least 2 and its neighborhood is connected, then (a)--(c) are the three possibilities for $N(P)$. In all three of these cases, moving $P$ down-left satisfies Property~\ref{prop:1}. (d) and (e) show the two cases for subsequently moving $P$ to a new position such that the invariants still hold for $\lin_i$.}
\label{fig:NPconn}
\end{figure}

When $N(P)$ is disconnected, we invoke Lemma~\ref{lem:eliminate1}.
It must be that $P$ has two neighbors that satisfy the conditions of the lemma, and so there exists a sequence of valid moves after which either location $\ell$ below-right of $P$ is occupied by another particle or it is valid for $P$ to move to $\ell$.
All moves in this sequence occur right of $P$, and thus don't affect the invariants for $\lin_i$.
If it is now valid for $P$ to move to $\ell$, we make this move and the number of particles in $\lin_i$ has decreased, as desired.
If another particle has moved to $\ell$, then $N(P)$ is now connected, the next case we consider.

When $N(P)$ is connected, it must look as in Fig.~\ref{fig:NPconn_1}, \ref{fig:NPconn_2}, or~\ref{fig:NPconn_3}.
In all cases, particle $P$ moving down-left is a valid move that decreases the number of particles in~$\lin_i$.
However, Invariant~\ref{inv:1} no longer holds for $\lin_i$ after this move, so we continue to move particle~$P$ down until it is adjacent to the bottom particle $Q$ in this component of particles in $\lin_i$.
If there is not already a line stretching down and left from $Q$, then $P$ moves down once more to start such a line (Fig.~\ref{fig:movePtoline_1}), which is valid because of the invariants for $\lin_i$.
If this line stretching down and left from $Q$ already exists, we note the locations at distances one and two above this line must all be unoccupied.
This follows from Invariants~\ref{inv:1} and~\ref{inv:2} for $\lin_i$: all particles left of $\lin_i$ must extend down and left from the bottom particle of some component in $\lin_i$, and the first such particle above $Q$ is at least two units above $P$'s original location and thus at least three units above $Q$.
Thus, it is valid (satisfying Property~\ref{prop:1}) to move $P$ along this line and add it to the end (Fig.~\ref{fig:movePtoline_2}).
In all cases, the number of particles in~$\lin_i$ decreases while the invariants for $\lin_i$ remain satisfied, as desired.
\end{proof}

Lemma~\ref{lem:fewerparticles} can be applied iteratively until all components of $\lin_i$ are of size one, and all particles left of $\lin_i$ form lines stretching down-left from these components of size one.
Thus, all particles left of $\lin_{i+1}$ form lines stretching down-left, satisfying Invariant~\ref{inv:1} for $\lin_{i+1}$.
We now consider how to also satisfy Invariant~\ref{inv:2} for $\lin_{i+1}$.

\begin{lem} \label{lem:inv_linestart}
If $\lin_i$ satisfies both invariants and $\lin_{i+1}$ satisfies Invariant~\ref{inv:1}, then there exists a sequence of valid moves after which $\lin_{i+1}$ satisfies both invariants.
\end{lem}
\begin{proof}
Because our particle configuration is connected, each line left of $\lin_{i+1}$ is connected to some particle in $\lin_{i+1}$.
However, the line may not stretch down and left from this particle or this particle may not have an empty location below it, as is required by Invariant~\ref{inv:2}.
Consider any component of $\lin_{i+1}$ which is adjacent to at least one line left of $\lin_{i+1}$ stretching down-left.
To satisfy Invariant~\ref{inv:2}, we merge all such lines into one, stretching down-left from the bottom particle $Q$ in this component.
First, we move the lowest line so that it is stretching down-left from $Q$.
An entire line can be moved down one unit by first moving the rightmost particle in this line (the particle in line $\lin_{i}$) down once, which is necessarily a valid move, and then by subsequently moving the remaining particles down once from right to left (for an example of this downward movement of a line, see Fig.~\ref{fig:line_union_a}).
This can be repeated until this lowest line is in the desired position, stretching down and left from $Q$.

\begin{figure}
\centering
\subfloat[]{
    \includegraphics[scale = 0.7, page = 1]{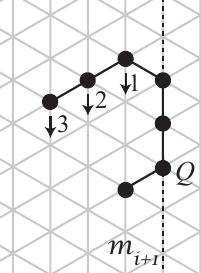}
    \label{fig:line_union_a}
} \hfil
\subfloat[]{
    \includegraphics[scale = 0.7, page = 2]{line_union.pdf}
    \label{fig:line_union_b}
} \hfil
\subfloat[]{
    \includegraphics[scale = 0.7, page = 3]{line_union.pdf}
    \label{fig:line_union_c}
} \hfil
\subfloat[]{
    \includegraphics[scale = 0.7, page = 4]{line_union.pdf}
    \label{fig:line_union_d}
} \hfil
\subfloat[]{
    \includegraphics[scale = 0.7, page = 5]{line_union.pdf}
    \label{fig:line_union_e}
}
\caption{The process of merging two lines stretching down-left from the same component of $\lin_{i+1}$ in order to satisfy Invariant~\ref{inv:2}. In (a) and (b), the moves must occur in the order listed.}
\label{fig:line_union}
\end{figure}

Iteratively consider the next lowest line.
As before, we move this line down one unit at a time by moving the particles each down once from right to left until the line is flush with the bottommost line (Fig.~\ref{fig:line_union_a}--\ref{fig:line_union_c}).
The particles in this line can then easily be added to the bottommost line one at a time, from left to right, as in Fig.~\ref{fig:line_union_c}--\ref{fig:line_union_e}.
We repeat this line merging process until all particles stretching down-left from this component of $\lin_i$ have been reorganized into one line stretching down-left from $Q$.
After repeating this process for all components in $\lin_{i+1}$, Invariant~\ref{inv:2} is satisfied for $\lin_{i+1}$.
Invariant~\ref{inv:1} still holds for $\lin_{i+1}$ as all particles are still in lines, so $\lin_{i+1}$ now satisfies both invariants, as claimed.
\end{proof}

We now combine the previous two lemmas to get the main inductive step for our sweep-line procedure.

\begin{lem} \label{lem:inductivestep}
If $\lin_i$ satisfies both invariants, then there exists a sequence of valid particle moves after which $\lin_{i+1}$ also satisfies both invariants.
\end{lem}
\begin{proof}
	Suppose $\lin_i$ satisfies both invariants.
If there are connected components of two or more particles contained in $\lin_i$, we can iteratively apply Lemma~\ref{lem:fewerparticles} to reduce the number of particles in $\lin_i$ without affecting the invariants.
After this, all components of $\lin_i$ consist of one particle.
Now all particles left of $\lin_{i+1}$ are in lines (possible consisting of just one particle) stretching down-left, satisfying Invariant~\ref{inv:1}.
Next, we can apply Lemma~\ref{lem:inv_linestart} to ensure that $\lin_{i+1}$ also satisfies Invariant~\ref{inv:2}, merging any lines stretching down-left from the same component of $\lin_{i+1}$.
Thus, there exists a sequence of valid moves after which $\lin_{i+1}$ satisfies both invariants, as claimed.
\end{proof}

\begin{lem} \label{lem:line}
There exists a valid sequence of moves transforming any configuration into a line.
\end{lem}
\begin{proof}
	Initially, $\lin_1$ for $\sigma$ trivially satisfies the invariants because there are no particles left of $\lin_1$.
Repeatedly using Lemma~\ref{lem:inductivestep}, we obtain a sequence of moves after which the invariants hold for some line $\lin_k$ which has no particles to its right.

All particles in $\lin_k$ must be in a single component.
If this was not the case, then the configuration would not be connected: particles left of $\lin_k$ only form lines that are insufficient to connect multiple components of $\lin_k$, and there are no particles right of $\lin_k$.
We know that the particle configuration must be connected because initial configuration $\sigma_0$ was connected and we have only made valid particle moves (Lemma~\ref{lem:invariant_connected}), so this is a contradiction, and $\lin_k$ must have a single component.

We repeatedly apply Lemma~\ref{lem:fewerparticles} until there is only one particle left in $\lin_k$ and line $\lin_k$ still satisfies the invariants.
At this point the particles form a single line stretching down-left from the single particle in $\lin_k$, and we have given a sequence of valid moves transforming an arbitrary configuration into a line.
\end{proof}

In particular, this shows that for any connected configuration there exists a valid sequence of moves transforming it into a configuration with no holes.

\begin{lem} \label{lem:eventually_no_holes}
Eventually $\M$ reaches a configuration with no holes, after which no holes are ever introduced again.
\end{lem}
\begin{proof}
	Let $\sigma_0 \in \Omega$ be the initial (connected) particle configuration given as input to Markov chain $\M$.
By Lemma~\ref{lem:line} for $\sigma_0$, there is positive probability that $\M$ will reach $\Ohf \subset \Omega$, the set of hole-free connected particle configurations.
Lemma~\ref{lem:line} holds for any configuration, so this is also true of each subsequent state $\sigma_i$.
Since $\Omega$ is finite, $\M$ must eventually reach $\Ohf$, as desired.
Finally, by Lemma~\ref{lem:invariant_holefree}, once $\Omega^*$ is reached, the particle system will remain hole-free for the rest of $\M$'s execution.
\end{proof}

Note that the previous lemma is equivalent to saying that any configuration with a hole  is a transient state of Markov chain $\M$.
We present one more lemma before proving $\M$ is irreducible on $\Ohf$ once it reaches $\Ohf$.
Let $M$ be the transition matrix of $\M$, that is,  $M(\sigma, \tau)$ is the probability of moving from state~$\sigma$ to state $\tau$ in one step of $\M$.

\begin{lem} \label{lem:rev}
	For any two configurations $\sigma,\tau \in \Ohf$, if $M(\sigma, \tau) > 0$ then $M(\tau,\sigma) > 0$; that is, once $\M$ reaches $\Ohf$, all of its transitions are reversible on $\Ohf$.
\end{lem}
\begin{proof}
	Let $\sigma,\tau \in \Ohf$ be any two configurations such that $M(\sigma, \tau) > 0$.
	Then $\sigma$ and $\tau$ differ by one particle $P$ that is at location $\ell$ in $\sigma$ and at adjacent location $\ell'$ in $\tau$.

	In $\tau$, particle $P$ at location $\ell'$ has at most four neighbors.
	It cannot have six neighbors because location $\ell$, which was previously occupied by $P$ in $\sigma$, is now unoccupied.
	It cannot have five neighbors because otherwise $\ell'$ would have been a hole in $\sigma$ when $P$ was at~$\ell$, a contradiction to our assumption that $\sigma \in \Ohf$.
	Because $M(\sigma, \tau) > 0$, Property~\ref{prop:1} or Property~\ref{prop:2} must hold for $\ell$ and $\ell'$.
	Both properties are symmetric with regard to the role played by $\ell$ and $\ell'$.
	Thus, if Markov chain $\M$, in state $\tau$, selects particle $P$, neighboring location $\ell$, and sufficiently small probability $q$ in Step~\ref{alg:M:q}, then because Conditions (1)--(3) of Step~\ref{alg:M:conds} are satisfied, particle $P$ moves to location $\ell$.
	This proves $M(\tau, \sigma) > 0$.
\end{proof}

\begin{lem} \label{lem:irred}
	Once Markov chain $\M$ reaches $\Ohf$, it is irreducible on $\Ohf$, the state space of all connected configurations without holes.
\end{lem}
\begin{proof}
	Let $\sigma$ and $\tau$ be any two connected configurations of $n$ particles with no holes. By Lemma~\ref{lem:line}, there exists a sequence of valid moves transforming $\sigma$ into a line.
	By Lemmas~\ref{lem:line} and~\ref{lem:rev}, there exists a sequence of valid moves transforming this line into $\tau$.
\end{proof}

\begin{cor} \label{lem:ergod}
	Once $\M$ reaches $\Ohf$, it is ergodic on $\Ohf$.
\end{cor}
\begin{proof}
	By Lemma~\ref{lem:irred}, $\M$ is irreducible on $\Ohf$.
	As long as $n > 1$ then every particle has at least one neighbor, so $\M$ is aperiodic because at each iteration there is a probability of at least $1/6$ that a particle proposes moving into an occupied neighboring location so no move is made.
	Thus, once $\M$ reaches $\Ohf$, it is ergodic on $\Ohf$.
\end{proof}

We note that $\M$ is not irreducible on $\Omega$, and thus not ergodic on $\Omega$, because it is not possible to get from a hole-free configuration to a configuration with a hole.
Ergodicity is necessary to apply tools from Markov chain analysis, as we do in the next subsection, which is why we focus on the behavior of $\M$ after it reaches $\Ohf$.

\subsection{The stationary distribution \texorpdfstring{$\pi$}{} of \texorpdfstring{$\M$}{M}} \label{subsec:stationary}

In this section we determine the stationary distribution of $\M$.

\begin{lem} \label{lem:stat-mass}
	If $\pi$ is a stationary distribution of $\M$, then for any $\sigma \in \Omega \setminus \Ohf$, $\pi(\sigma) = 0$.
\end{lem}
\begin{proof}
	For any configuration $\sigma \in \Omega \setminus \Ohf$, there is a positive probability of moving into~$\Ohf$ in some later time step (Lemma~\ref{lem:eventually_no_holes}).
	For any configuration $\tau \in \Ohf$, there is zero probability of reaching a configuration with holes (Lemma~\ref{lem:invariant_holefree}).
	If a stationary distribution~$\pi$ were to put any positive probability mass on states in $\Omega \setminus \Ohf$, over time the total probability mass within $\Omega \setminus \Ohf$ would decrease as it leaks into $\Ohf$ with no possibility of returning.
	Thus such a distribution could not be stationary, a contradiction.
	We conclude that any stationary distribution $\pi$ has $\pi(\sigma) = 0$ for all $\sigma \in \Omega \setminus \Ohf$, as claimed.
\end{proof}

\begin{lem} \label{lem:stat-edge}
	$\M$ has a unique stationary distribution $\pi$ given by
	\[\pi(\sigma) = \left\{ \begin{array}{cl} \frac{\lambda ^{e(\sigma)}}{Z} & \sigma \in \Ohf \\ 0 & \sigma \in \Omega \setminus \Ohf \end{array} \right.\]
	where $Z = \sum_{\sigma\in \Ohf} \lambda^{e(\sigma)}$ is the normalizing constant, also called the partition function.
\end{lem}
\begin{proof}
	Lemma~\ref{lem:stat-mass} guarantees that any stationary distribution of $\M$ has $\pi(\sigma) = 0$ for configurations $\sigma \not\in \Ohf$.
	Once $\M$ reaches $\Ohf$ (which it is guaranteed to by Lemma~\ref{lem:eventually_no_holes}), it is ergodic on $\Ohf$ (Lemma~\ref{lem:ergod}).
	We conclude, because $\Ohf$ is finite, that $\M$ on $\Ohf$ has a unique stationary distribution, and thus $\M$ on $\Omega$ also has a unique stationary distribution.

	We confirm that $\pi$ as stated above is this unique stationary distribution by detailed balance.
	Let $\sigma$ and $\tau$ be configurations in $\Ohf$ with $\sigma \neq \tau$ such that $M(\sigma, \tau) > 0$.
	By Lemma~\ref{lem:rev}, also $M(\tau, \sigma) > 0$.
	Suppose particle $P$ moves from location $\ell$ in $\sigma$ to neighboring location $\ell'$ in $\tau$.
	Let $e$ be the number of edges formed by $P$ has when it is in location~$\ell$, and let $e'$ be that number when $P$ is in location $\ell'$.
	This implies $e(\sigma) - e(\tau) = e - e'$. If $\lambda^{e'} \leq \lambda^{e}$, then we see that
	\[M(\sigma, \tau) = \frac{1}{n} \cdot \frac{1}{6} \cdot \lambda^{e'-e} ~\text{ and }~ M(\tau, \sigma) =  \frac{1}{n} \cdot \frac{1}{6} \cdot 1.\]
	In this case we can verify that $\sigma$ and $\tau$ satisfy the detailed balance condition:
	\[\pi(\sigma) M(\sigma, \tau) = \frac{\lambda^{e(\sigma)}}{Z} \cdot \frac{\lambda^{e'-e}}{6n} = \frac{\lambda^{e(\tau)}}{Z \cdot 6n} = \pi(\tau) M(\tau, \sigma).\]
	If $\lambda^{e'} > \lambda^{e}$, we can similarly calculate these probabilities to verify detailed balance:
	\[M(\sigma, \tau) = \frac{1}{n} \cdot \frac{1}{6} \cdot 1 ~\text{ and }~ M(\tau, \sigma) = \frac{1}{n} \cdot \frac{1}{6} \cdot \lambda^{e-e'},\]
	\[\pi(\sigma) M(\sigma, \tau)= \frac{\lambda^{e(\sigma)}}{Z \cdot 6n} = \frac{\lambda^{e(\tau)}}{Z} \frac{\lambda^{e-e'}}{6n} = \pi(\tau) M(\tau, \sigma).\]
	Since the detailed balance condition is satisfied for all $\sigma, \tau \in \Ohf$, it only remains to verify that $\pi$ is in fact a probability distribution:
	\[\sum_{\sigma \in \Omega} \pi(\sigma) = \sum_{\sigma \in \Ohf} \frac{\lambda^{e(\sigma)}}{Z} + \sum_{\sigma \in \Omega \setminus \Ohf} 0 = \frac{\sum_{\sigma\in \Ohf} \lambda^{e(\sigma)}}{\sum_{\sigma\in \Ohf} \lambda^{e(\sigma)}} = 1.\]
	We conclude $\pi$ is the unique stationary distribution of $\M$.
\end{proof}

While it is natural to assume maximizing the number of edges in a particle configuration results in more compression, here we formalize this.
We prove $\pi$ can also be expressed in terms of perimeter.
This implies $\M$ converges to a distribution weighted by the perimeter of configurations, a global characteristic, even though the probability of any particle move is determined only by local information.

\begin{cor} \label{cor:stat-perim}
	The stationary distribution $\pi$ of $\M$ is also given by
	\[\pi(\sigma) = \left\{ \begin{array}{cl} \frac{\lambda ^{-p(\sigma)}}{Z} & \sigma \in \Ohf \\ 0 & \sigma \in \Omega \setminus \Ohf \end{array} \right.\]
	where $Z = \sum_{\sigma \in \Ohf} \lambda^{-p(\sigma)}$ is the normalizing constant, also called the partition function.
\end{cor}
\begin{proof}
	This expression is  equal to $\M$'s unique stationary distribution when $\sigma \not\in \Ohf$, so it only remains to verify the case $\sigma \in \Ohf$.
	We use Lemma~\ref{lem:edge=perim} and Lemma~\ref{lem:stat-edge}:
	\[\pi(\sigma) = \frac{\lambda^{e(\sigma)}}{\sum_{\sigma \in \Ohf} \lambda^{e(\sigma)}}
	= \frac{\lambda^{3n - p(\sigma) - 3}}{\sum_{\sigma\in \Ohf} \lambda^{3n - p(\sigma) - 3}}
	= \frac{\lambda^{3n-3}}{\lambda^{3n-3}} \cdot \frac{\lambda^{-p(\sigma)}}{\sum_{\sigma\in \Ohf} \lambda^{-p(\sigma)}}
	= \frac{\lambda^{-p(\sigma)}}{\sum_{\sigma\in \Ohf} \lambda^{-p(\sigma)}}.\]
\end{proof}

The conference version of this paper~\cite{Cannon2016} also expressed the stationary distribution in terms of the number of triangles in a configuration. Recall a {\it triangle} is a face of $\Gtri$ that has all three of its vertices occupied by particles and $t(\sigma)$ is the number of triangles in configuration $\sigma$.
We include the following corollary for completeness, but will not use it in subsequent sections.

\begin{cor} \label{cor:stat-tri}
	The stationary distribution $\pi$ of $\M$ is also given by
	\[\pi(\sigma) = \left\{ \begin{array}{cl} \frac{\lambda ^{t(\sigma)}}{Z} & \sigma \in \Ohf \\ 0 & \sigma \in \Omega \setminus \Ohf \end{array} \right.\]
	where $Z = \sum_{\sigma \in \Ohf} \lambda^{t(\sigma)}$ is the  normalizing constant, also called the partition function.
\end{cor}
\begin{proof}
	This follows from Lemma \ref{lem:tri=perim} and Corollary \ref{cor:stat-perim}:
	\[\pi(\sigma) = \frac{\lambda^{-p(\sigma)}}{\sum_{\sigma \in \Ohf} \lambda^{-p(\sigma)}}
	= \frac{\lambda^{-(2n - t(\sigma) - 2)}}{\sum_{\sigma \in \Ohf} \lambda^{-(2n - t(\sigma) - 2)}}
	= \frac{\lambda^{-2n+2}}{\lambda^{-2n+2}} \cdot \frac{\lambda^{t(\sigma)}}{\sum_{\sigma \in \Ohf} \lambda^{t(\sigma)}}
	= \frac{\lambda^{t(\sigma)}}{\sum_{\sigma \in \Ohf} \lambda^{t(\sigma)}}.\]
\end{proof}

\subsection{Convergence Time of Markov Chain \texorpdfstring{$\M$}{M}} \label{subsec:convergence}

We prove in Section~\ref{sec:results} that when $\lambda > \cbound$, if Markov chain $\M$ has converged to its stationary distribution, then with all but exponentially small probability the particle system will be compressed.
However, we do not give explicit bounds on the time required for this to occur, and we believe proving rigorous bounds will be challenging.

\redcomment{SC: Dana, please review this discussion.}
A common measure of convergence time of a Markov chain is its {\it mixing time}, the number of iterations until the distribution is within total variation distance $\varepsilon$ of the stationary distribution, starting from the worst initial configuration.
Getting a polynomial bound on the mixing time of our Markov chain $\M$ 
is likely to be challenging because of its similarity to physical systems such as the Ising and Potts models, common models of ferromagnetism from statistical physics.
For example, local-update dynamics for the two-dimensional Ising model with constant boundary conditions are believed to have polynomial mixing time, though proving so remains a difficult open problem despite breakthrough works showing subexponential~\cite{Martinelli2010} and subsequently quasipolynomial~\cite{Lubetzky2013} mixing time upper bounds.
Our $\M$ also uses local update steps and, like the constant boundary Ising model, has two possible states for each vertex of a lattice (occupied vs. unoccupied) and outside a region of interest all states are the same (unoccupied).
The shrinkage over time of the boundary of the particle configuration under $\M$ is similar to the shrinkage over time in the Ising model of `droplets' of one state surrounded by the other state (see, e.g.,~\cite{Caputo2011} for work investigating such `droplets' in two and three dimensions).
This shrinking of droplets is believed --- but not proved --- to be the salient feature determining the mixing time for the Ising model with constant boundary conditions.
We see similar difficulties in analyzing the mixing time of our Markov chain $\M$, and thus believe obtaining rigorous upper bounds on its mixing time will be challenging.



However, mixing time may not be the correct measure of our algorithm's convergence.
While we prove in later sections that compression occurs after $\M$ has reached its stationary distribution, compression could occur much earlier.
Thus, even if it takes exponential time for $\M$ to converge to its stationary distribution, which is certainly plausible, it may be true that the particles achieve compression after only a polynomial number of steps.
When starting from a line of $n$ particles, our simulations indicate that doubling the number of particles consistently results in about a ten-fold increase in iterations until compression is achieved.
Based on this, we conjecture the number of iterations of $\M$ until compression occurs is $\Omega(n^3)$ and $\bigO{n^4}$, the equivalent of $\Omega(n^2)$ and $\bigO{n^3}$ asynchronous rounds of $\A$.
Furthermore, we do not expect the presence of holes in the initial configuration to significantly delay compression, even though this may increase the mixing time.

\section{Achieving Compression} \label{sec:results}

We proved in Section~\ref{subsec:stationary} that Markov chain $\M$ converges to a unique stationary distribution, and we know that distribution exactly (Corollary~\ref{cor:stat-perim}).
In this section, we show that when parameter $\lambda$ is large enough, this stationary distribution exhibits compression with high probability.
While compression could actually occur even earlier, before $\M$ is close to stationarity, our proofs rely on analyzing the stationary distribution of $\M$.

Recall for any $\alpha > 1$ we say a configuration $\sigma$ with $n$ particles is $\alpha$-compressed if its perimeter $p(\sigma) < \alpha \cdot \pmin$, where $\pmin$ is the minimum possible perimeter of a configuration with $n$ particles.
We prove that, for any $\alpha > 1$ and provided $\lambda$ and $n$ are large enough, a configuration chosen at random according to the stationary distribution of $\M$ is $\alpha$-compressed with all but a probability that is exponentially small (in $n$).
Values of $\alpha$ closer to $1$ simply require larger $\lambda$ values.
Conversely, we then prove (as a corollary) that for any $\lambda > \cbound$, there is a constant $\alpha$ such that with high probability $\alpha$-compression occurs at stationarity.

\subsection{Preliminaries: Perimeter and Self-Avoiding Walks} \label{subsec:perimsaws}

We begin with some necessary results bounding the number of connected, hole-free particle system configurations with a certain perimeter.
Let $S_{\alpha}$ be the set of all connected, hole-free configurations with perimeter at least $\alpha \cdot \pmin$, for some constant $\alpha > 1$.
We only consider hole-free configurations because we are concerned with behavior of $\M$ at stationarity and the stationary distribution $\pi$ of $\M$ only gives positive probability to hole-free configurations in $\Ohf$ (Corollary~\ref{cor:stat-perim}).
We want an upper bound on $\pi(S_\alpha) = \sum_{\sigma \in S_\alpha} \pi(\sigma)$, the probability of being in a configuration with large perimeter at stationarity, in order to argue that this probability is exponentially small.

Let $c_k$ denote the number of connected, hole-free configurations with perimeter $k$.
Recall that $\pmax = 2n-2$ is the maximum possible perimeter for a configuration of $n$ particles; using the expression for $\pi$ given in Corollary~\ref{cor:stat-perim}, we can write $\pi(S_\alpha)$ as:
\[\pi(S_\alpha) = \sum_{\sigma \in S_\alpha} \pi(\sigma)
= \sum_{\sigma \in S_\alpha} \frac{\lambda^{-p(\sigma)}}{Z}
= \frac{\sum_{k = \lceil \alpha \cdot \pmin \rceil}^{\pmax} c_k \lambda^{-k}}{Z}.\]
Recall that Corollary~\ref{cor:stat-perim} defined the partition function as $Z = \sum_{\sigma \in \Ohf} \lambda^{-p(\sigma)}$, the summed weight of all connected, hole-free configurations.
In order to give an upper bound on $\pi(S_\alpha)$, we establish a lower bound on $Z$ and an upper bound on $c_k$.
It suffices to use the trivial bound $Z \geq \lambda^{-\pmin}$ for the former; to bound the latter, we turn to lattice duality and self-avoiding walks (for a more thorough treatment of self-avoiding walks, see, e.g.,~\cite{Bauerschmidt2012}).

\begin{defn} \label{defn:saw}
A {\it self-avoiding walk} in a graph is a walk that never visits the same vertex twice.
\end{defn}

Self-avoiding walks are most commonly studied for graphs that are planar lattices, and we will focus on self-avoiding walks in the hexagonal lattice, also called the honeycomb lattice (Fig.~\ref{fig:SAW-ex-hex-lattice}).
Examples of self-avoiding walks and non-self-avoiding walks in this lattice are shown in Fig.~\ref{fig:SAW-ex-good} and~\ref{fig:SAW-ex-bad}, respectively.
The hexagonal lattice is of interest because it is dual to the triangular lattice $\Gtri$ that particles occupy in our model.
That is, by creating a new vertex in every face of the triangular lattice and connecting two of these new vertices if their corresponding triangular faces have a common edge, we obtain the hexagonal lattice; see Fig.~\ref{fig:dual}.

\begin{figure}
\centering
\subfloat[]{
    \includegraphics[scale = 0.45, page = 2]{SAW-ex.pdf}
    \label{fig:SAW-ex-hex-lattice}
} \hfil
\subfloat[]{
    \includegraphics[scale = 0.45, page = 1]{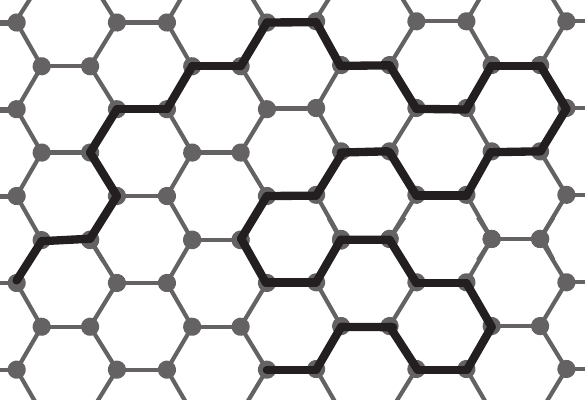}
    \label{fig:SAW-ex-good}
} \hfil
\subfloat[]{
    \includegraphics[scale = 0.45, page = 3]{SAW-ex.pdf}
    \label{fig:SAW-ex-bad}
}
\caption{(a) The hexagonal lattice. (b) A self-avoiding walk in the hexagonal lattice. (c) A walk that is not self-avoiding.}
\label{fig:SAW-ex}
\end{figure}

The number of self-avoiding walks of a certain length starting from a fixed vertex has been extensively studied for many planar lattices.
This number is believed to grow exponentially with the length of the walk, and the base of this exponent is known as the \emph{connective constant} of the lattice.
More concretely, if $N_l$ is the number of self-avoiding walks of length $l$ in some planar lattice $L$, then the connective constant of lattice $L$ is defined as $\mu_L = \lim_{l \to \infty} (N_l)^{1/l}$.
For example, the connective constant of the square lattice is $2.625622 \leq \mu_{sq} \leq 2.679193$, but an exact value has not been rigorously proved~\cite{Jensen2004,Ponitz2000}.
The only lattice for which the connective constant is exactly known is our lattice of interest, the hexagonal lattice.

\begin{figure}
\centering
\subfloat[]{
    \includegraphics[scale = 0.6]{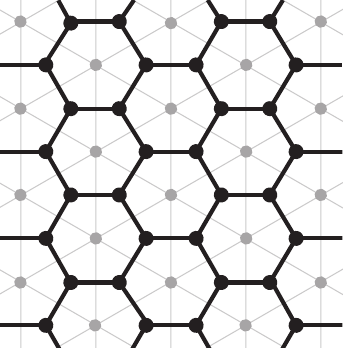}
    \label{fig:dual}
} \hfil
\subfloat[]{
    \includegraphics[scale = 0.4]{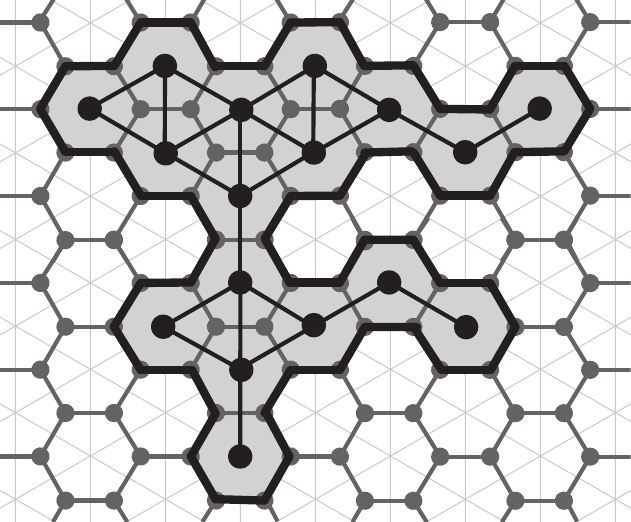}
    \label{fig:dual_ex}
}
\caption{(a) The duality between the triangular lattice and the hexagonal lattice. (b) An example of a particle configuration $\sigma$, its corresponding polygon in the hexagonal lattice (shaded), and the boundary of this region which is a self-avoiding polygon in the hexagonal lattice (bold).}
\end{figure}

\begin{thm}[\cite{DuminilCopin2012}] \label{thm:muhex}
The connective constant of the hexagonal lattice is $\mu_{hex} = \sqrt{2+\sqrt{2}}$.
\end{thm}

\noindent This theorem implies that the number of self-avoiding walks of length $l$ in the hexagonal lattice is $f(l) \cdot \mu_{hex}^l$, for some subexponential function $f$.

To bound the number of connected, hole-free particle system configurations with some fixed perimeter, we turn from self-avoiding walks to the closely related notion of {\it self-avoiding polygons}, where a self-avoiding polygon is a self-avoiding walk that starts and ends at the same vertex (Fig.~\ref{fig:dual_ex}).
The number of self-avoiding walks of length $l$ is an upper bound on the number of self-avoiding polygons of perimeter $l$.

\begin{lem} \label{lem:mu}
The number of connected, hole-free particle system configurations with $n$ particles and perimeter $k$ is at most $f(k) \cdot (2+\sqrt{2})^k$ for some subexponential function $f$.
\end{lem}
\begin{proof}
Consider the dual to the triangular lattice $\Gtri$, the hexagonal lattice $\Ghex$ (Fig.~\ref{fig:dual}).
For any connected, hole-free particle system configuration $\sigma$ with $n$ particles, consider the union $A_\sigma$ of all the faces of $\Ghex$ corresponding to vertices of $\Gtri$ that are occupied in $\sigma$.
Whenever two particles are adjacent in $\Gtri$, their corresponding faces in $\Ghex$ share an edge.
This union $A_\sigma$ is a simply connected polygon because $\sigma$ is connected and has no holes; moreover, the boundary of $A_\sigma$ forms a self-avoiding polygon in $\Ghex$ (bold in Fig.~\ref{fig:dual_ex}).

We first show that the number of connected, hole-free particle system configurations in $\Gtri$ with perimeter $k$ is at most the number of self-avoiding polygons in $\Ghex$ with perimeter $2k+6$.
Let $\sigma$ be a connected, hole-free configuration with perimeter $k$.
A particle $P$ is on the (unique external) boundary of $\sigma$ if and only if its corresponding hexagon $H_P$ in $\Ghex$ shares an edge with the perimeter of $A_\sigma$.
That is, if a particle $P$ appears once on the boundary of $\sigma$ with exterior angle $\theta_P \in \{120^\circ, 180^\circ, 240^\circ, 300^\circ, 360^\circ\}$, then $H_P$ has $(\theta_P/60^\circ) - 1$ of its edges contained in the perimeter of $A_\sigma$.
More generally, if a particle $P$ appears $m_P \geq 1$ times on the boundary of $\sigma$ with exterior angles summing to $\theta_P$, then $H_P$ has $(\theta_P/60^\circ) - m_P$ of its edges contained in the perimeter of $A_\sigma$.
Thus, we conclude the number of edges on the perimeter of $A_\sigma$ is:
\[p(A_\sigma) = \sum_{P \in p(\sigma)} \left(\frac{\theta_P}{60^\circ} - m_P\right)
= \frac{1}{60^\circ} \left(\sum_{P \in p(\sigma)} \theta_P\right) - k
= \frac{1}{60^\circ} \left(180^\circ k + 360^\circ\right) - k
= 2k + 6.\]

As noted before, the number of self-avoiding polygons in a lattice with perimeter $l$ is at most the number of self-avoiding walks in that lattice with length $l$.
So the number of self-avoiding polygons in $\Ghex$ with perimeter $2k+6$ is at most the number of self-avoiding walks in $\Ghex$ of length $2k+6$.
By Theorem~\ref{thm:muhex}, the number of self-avoiding walks in $\Ghex$ of length $2k+6$ is $f_1(2k+6) \cdot \mu_{hex}^{2k+6}$, for some subexponential function $f_1$.
Let $f(k) = f_1(2k+6) \cdot \mu_{hex}^6$; note that $f$ is also subexponential.
Then, all together, the number of connected, hole-free particle system configurations with perimeter $k$ is at most $f(k) \cdot \mu_{hex}^{2k} = f(k) \cdot (2 + \sqrt{2})^k$, as desired.
\end{proof}

We can restate Lemma~\ref{lem:mu} in a slightly different way to make it easier to use in our later proofs.

\begin{lem} \label{lem:nu}
For any $\nu > \cbound$, there is an integer $n_1(\nu)$ such that for all $n \geq n_1(\nu)$, the number of connected, hole-free particle system configurations with $n$ particles and perimeter $k$ is at most $\nu^k$.
\end{lem}
\begin{proof}
From Lemma~\ref{lem:mu} we know that the number of connected, hole-free configurations with $n$ particles and perimeter $k$ is at most at most $f(k) \cdot (2+\sqrt{2})^k$, for some subexponential function $f$.
Because $\nu > \cbound$ and $f$ is subexponential, it follows that
\[\lim_{k \to \infty} \frac{f(k) \cdot (2+\sqrt{2})^k}{\nu^k} = 0.\]
Let $k_1(\nu)$ be such that for all $k > k_1(\nu)$, $f(k) \cdot (2+\sqrt{2})^k \leq \nu^k$; $k_1(\nu)$ must exist because the above limit is less than one. Let $n_1(\nu) = k_1(\nu)^2$.
For any $n \geq n_1(\nu)$, all connected configurations with $n$ particles have perimeter at least $k_1(\nu)$ by Lemma~\ref{lem:pmin}.
We conclude that for any $n \geq n_1(\nu)$ and for any $k$ between $\pmin(n)$ and $\pmax(n)$, $f(k) \cdot (2+\sqrt{2})^k \leq \nu^k$ and thus the number of connected, hole-free configurations with perimeter $k$ is at most $\nu^k$, as claimed.
If $k$ is not within these bounds, there are $0$ configurations with $n$ particles and perimeter $k$, so the lemma is trivially true.
\end{proof}

We note that, in general, the closer $\nu$ is to $\cbound$ the larger $n_1(\nu)$ has to be.

\subsection{Proof of Compression} \label{subsec:compproof}

To simplify notation, we define the {\it weight} of a configuration $\sigma$ to be $w(\sigma) = \pi(\sigma) \cdot Z = \lambda^{-p(\sigma)}$.
For a set $S \subseteq \Omega$, we define $w(S) = \sum_{\sigma \in S}w(\sigma)$ as the sum of the weights of all configurations in $S$.
We now prove our main result.

\begin{thm} \label{thm:compress_alpha}
For any $\alpha > 1$, let $\lambda^* = (\cbound)^{\frac{\alpha}{\alpha-1}}$. There exists $n^* \geq 0$ and $\zeta < 1$ such that for all $\lambda > \lambda^*$ and $n > n^*$, the probability that a random sample $\sigma$ drawn according to the stationary distribution $\pi$ of $\M$ is not $\alpha$-compressed is exponentially small:
\[\p_{\sigma \sim \pi}\left(p(\sigma) \geq \alpha \cdot \pmin\right) < \zeta^{\sqrt{n}}.\]
\end{thm}
\begin{proof}
Recall that $S_\alpha$ is the set of connected, hole-free configurations with perimeter at least $\alpha \cdot \pmin$.
We wish to show that $\pi(S_\alpha)$ is smaller than some function that is exponentially small in $n$.

We first consider the partition function $Z$ of $\pi$; recall $Z = \sum_{\sigma \in \Ohf} \lambda^{-p(\sigma)}$.
If $\sigma_{min}$ is a configuration of $n$ particles achieving the minimum possible perimeter $\pmin$, then $w(\sigma_{min}) = \lambda^{-p_{min}}$ is a lower bound on $Z$. It follows that
\[\pi(S_\alpha) = \frac{w(S_\alpha)}{Z}
< \frac{w(S_\alpha)}{w(\sigma_{min})}.\]

The remainder of this proof will be spent finding an upper bound on $w(S_\alpha) / w(\sigma_{min})$ that is exponentially small in $n$.
To begin, we stratify $S_\alpha$ into sets of configurations that have the same perimeter; recall every configuration has an integer perimeter because of lattice constraints.
Let $A_k$ be the set of all configurations with perimeter $k$; then $S_\alpha = \bigcup_{k = \lceil \alpha \cdot \pmin \rceil}^{\pmax} A_k$.
Noting that $p_{max} = 2n-2$, we can then write
\[\frac{w(S_\alpha)}{w(\sigma_{min})} = \frac{\sum_{k = \lceil \alpha \cdot \pmin \rceil}^{2n-2} w(A_k)}{\lambda^{-\pmin}}.\]

Since all configurations in $A_k$ have the same perimeter $k$, they also have the same weight $\lambda^{-k}$; thus, $w(A_k) = |A_k|\lambda^{-k}$.
Choose $\nu$ such that $\lambda^* < \nu^{\frac{\alpha}{\alpha - 1}} < \lambda$, implying $\cbound < \nu < \lambda^{\frac{\alpha - 1}{\alpha}}$; since $\lambda > \lambda^*$, such a $\nu$ must exist.
By Lemma~\ref{lem:nu}, provided $n$ is large enough (i.e., larger than the value $n_1(\nu)$), we have $|A_k| \leq \nu^k$.
So, we have
\[\frac{w(S_\alpha)}{w(\sigma_{min})}
= \frac{\sum_{k = \lceil \alpha \cdot \pmin \rceil}^{2n-2} |A_k| \lambda^{-k}}{\lambda^{-\pmin}}
\leq \frac{\sum_{k = \lceil \alpha \cdot \pmin \rceil}^{2n-2} \nu^k \lambda^{-k}}{\lambda^{-\pmin}}
= \sum_{k = \lceil \alpha \cdot \pmin \rceil}^{2n-2} \nu^k \lambda^{-k}\lambda^{\pmin}.\]

Because $k \geq \alpha \cdot \pmin$, it follows that $k / \alpha \geq \pmin$.
As $\lambda > \lambda^* > 2+\sqrt{2} > 1$, we see that
\[\frac{w(S_\alpha)}{w(\sigma_{min})}
\leq \sum_{k = \lceil \alpha \cdot \pmin \rceil}^{2n-2} \nu^k \lambda^{-k}\lambda^{k / \alpha}
= \sum_{k = \lceil \alpha \cdot \pmin \rceil}^{2n-2} \left(\frac{\nu}{\lambda^{\frac{\alpha-1}{\alpha}}}\right)^k.\]

We chose $\nu$ such that $\nu < \lambda^{\frac{\alpha-1}{\alpha}}$, so $(\nu / \lambda^{\frac{\alpha-1}{\alpha}}) < 1$.
Because $k \geq \alpha \cdot \pmin > \alpha \sqrt{n}$  (by Lemma~\ref{lem:pmin}), we see that
\[\frac{w(S_\alpha)}{w(\sigma_{min})}
\leq \sum_{k = \lceil \alpha \cdot \pmin \rceil}^{2n-2} \left(\frac{\nu}{\lambda^{\frac{\alpha-1}{\alpha}}}\right)^{\alpha \sqrt{n}}
\leq (2n-2) \left(\frac{\nu}{\lambda^{\frac{\alpha-1}{\alpha}}}\right)^{\alpha \sqrt{n}}.\]

Since $(\nu / \lambda^{\frac{\alpha-1}{\alpha}}) < 1$, we can find a constant $\zeta < 1$ such for all sufficiently large $n$, say $n \geq n_2$, it holds that
\[\frac{w(S_\alpha)}{w(\sigma_{min})} \leq (2n-2)  \left(\frac{\nu}{\lambda^{\frac{\alpha-1}{\alpha}}}\right)^{\alpha \sqrt{n}}
< \zeta^{\sqrt{n}}.\]

Setting $n^* = \max(n_1(\nu), n_2)$ completes the proof of the theorem.
\end{proof}

Though Theorem~\ref{thm:compress_alpha} is proved only in the case where the number of particles is sufficiently large, we expect and observe it to hold for much smaller $n$.
We note that the closer the value $\nu^{\frac{\alpha}{\alpha - 1}}$ used in the above proof is to $\lambda^*$ the larger $n_1(\nu)$ is, and the closer $\nu$ is to $\lambda$, the larger $n_2$ is.
In particular, when $\lambda$ is close to $\lambda^*$ then $n^*$ must be large, while for $\lambda \gg \lambda^*$, smaller values of $n^*$ suffice for the proof.
Computing an exact value for $n^*$ is difficult because the value $n_1(\nu)$ from Lemma~\ref{lem:nu} is not known explicitly; see Section 4 of~\cite{DuminilCopin2012} and references therein.

While the above result shows that $\M$ (and, by extension, the local algorithm $\A$) accomplishes $\alpha$-compression for any $\alpha > 1$, smaller values of $\alpha$ require larger values of $\lambda$.
In practice, when $\lambda$ is large $\M$ takes a very long time to reach any compressed configuration.
Because of this, what happens when $\lambda$ is small is also of interest.
We now show that provided $\lambda > \cbound$, there is some constant $\alpha$ such that $\alpha$-compression occurs.
Of course, there is again a tradeoff: the smaller $\lambda$ is, the larger $\alpha$ is.

\begin{cor} \label{cor:compress_lambda}
For any $\lambda > \cbound$, for any constant $\alpha > \log_{\cbound}\lambda/(\log_{\cbound}\lambda - 1)$ there exists $n^* \geq 0$ and $\zeta < 1$ such that for all $n \geq n^*$, a random sample $\sigma$ drawn according to the stationary distribution $\pi$ of $\M$ satisfies
\[\p_{\sigma \sim \pi}\left(p(\sigma) \geq \alpha \cdot \pmin\right) < \zeta^{\sqrt{n}}.\]
\end{cor}
\begin{proof}
If $\alpha > \frac{\log_{\cbound}\lambda}{\log_{\cbound}\lambda - 1}$, then solving for $\lambda$ gives $\lambda > \nu^{\frac{\alpha}{\alpha-1}}$.
Theorem~\ref{thm:compress_alpha} then gives the desired result.
\end{proof}

\section{Using \texorpdfstring{$\M$}{M} for Expansion} \label{sec:expansion}

Now that we have proved Markov chain $\M$ (and local algorithm $\A$) yields compression whenever $\lambda > \cbound$, it is natural to ask about the behavior of $\M$ when $\lambda \leq \cbound$.
As $\lambda > 1$ corresponds to particles favoring having more neighbors, one might conjecture that compression occurs whenever $\lambda > 1$.
We show, counterintuitively, that this is not the case: for all $\lambda < \ebound$, compression does not occur.
For example, consider the simulation depicted in Fig.~\ref{fig:bias2_notCompressed}; even after $20$ million iterations of $\M$ with bias $\lambda = 2$, the system has not compressed.
This stands in stark contrast to the simulations depicted in Fig.~\ref{fig:line100_bias4}, which achieved good compression after only $5$ million iterations of $\M$ using $\lambda = 4$.

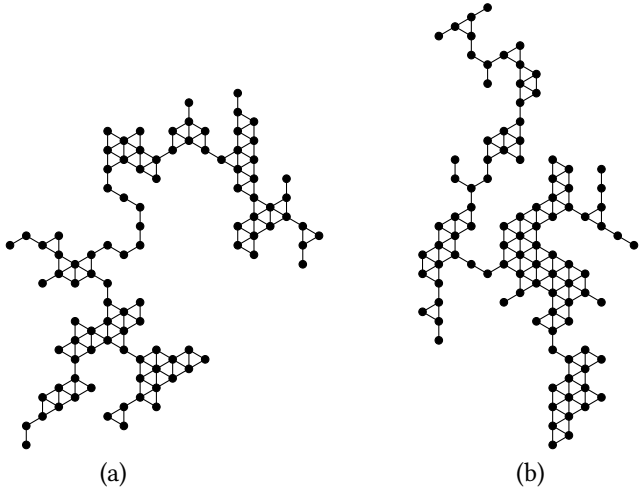
\begin{figure}[ht]
\centering
\begin{tikzpicture}[scale=0.25]
\draw[fill](10.3923,15.5)circle(0.2);
\draw(10.3923,15.5)--(11.2583,15);
\draw(10.3923,15.5)--(10.3923,16.5);
\draw[fill](12.1244,18.5)circle(0.2);
\draw[fill](5.19615,1.5)circle(0.2);
\draw(5.19615,1.5)--(6.06218,1);
\draw(5.19615,1.5)--(6.06218,2);
\draw[fill](6.06218,6)circle(0.2);
\draw(6.06218,6)--(6.9282,6.5);
\draw(6.06218,6)--(6.06218,7);
\draw[fill](12.9904,14)circle(0.2);
\draw(12.9904,14)--(12.9904,15);
\draw[fill](6.9282,12.5)circle(0.2);
\draw[fill](4.33013,10)circle(0.2);
\draw(4.33013,10)--(5.19615,10.5);
\draw[fill](3.4641,2.5)circle(0.2);
\draw(3.4641,2.5)--(4.33013,3);
\draw(3.4641,2.5)--(3.4641,3.5);
\draw[fill](5.19615,14.5)circle(0.2);
\draw(5.19615,14.5)--(6.06218,15);
\draw(5.19615,14.5)--(5.19615,15.5);
\draw[fill](2.59808,2)circle(0.2);
\draw(2.59808,2)--(3.4641,2.5);
\draw(2.59808,2)--(2.59808,3);
\draw[fill](12.9904,12)circle(0.2);
\draw(12.9904,12)--(13.8564,11.5);
\draw(12.9904,12)--(13.8564,12.5);
\draw(12.9904,12)--(12.9904,13);
\draw[fill](2.59808,9)circle(0.2);
\draw(2.59808,9)--(3.4641,8.5);
\draw(2.59808,9)--(3.4641,9.5);
\draw(2.59808,9)--(2.59808,10);
\draw[fill](12.1244,10.5)circle(0.2);
\draw(12.1244,10.5)--(12.9904,10);
\draw(12.1244,10.5)--(12.9904,11);
\draw(12.1244,10.5)--(12.1244,11.5);
\draw[fill](0,10.5)circle(0.2);
\draw(0,10.5)--(0.866025,11);
\draw[fill](2.59808,11)circle(0.2);
\draw[fill](7.79423,15)circle(0.2);
\draw(7.79423,15)--(8.66025,15.5);
\draw[fill](5.19615,5.5)circle(0.2);
\draw(5.19615,5.5)--(6.06218,5);
\draw(5.19615,5.5)--(6.06218,6);
\draw(5.19615,5.5)--(5.19615,6.5);
\draw[fill](13.8564,12.5)circle(0.2);
\draw(13.8564,12.5)--(14.7224,12);
\draw(13.8564,12.5)--(14.7224,13);
\draw[fill](7.79423,4)circle(0.2);
\draw(7.79423,4)--(8.66025,3.5);
\draw(7.79423,4)--(8.66025,4.5);
\draw(7.79423,4)--(7.79423,5);
\draw[fill](12.9904,15)circle(0.2);
\draw(12.9904,15)--(12.9904,16);
\draw[fill](6.06218,2)circle(0.2);
\draw(6.06218,2)--(6.9282,2.5);
\draw[fill](10.3923,4.5)circle(0.2);
\draw[fill](7.79423,14)circle(0.2);
\draw(7.79423,14)--(7.79423,15);
\draw[fill](6.9282,3.5)circle(0.2);
\draw(6.9282,3.5)--(7.79423,3);
\draw(6.9282,3.5)--(7.79423,4);
\draw(6.9282,3.5)--(6.9282,4.5);
\draw[fill](12.1244,16.5)circle(0.2);
\draw(12.1244,16.5)--(12.9904,16);
\draw(12.1244,16.5)--(12.9904,17);
\draw(12.1244,16.5)--(12.1244,17.5);
\draw[fill](8.66025,3.5)circle(0.2);
\draw(8.66025,3.5)--(9.52628,4);
\draw(8.66025,3.5)--(8.66025,4.5);
\draw[fill](12.1244,15.5)circle(0.2);
\draw(12.1244,15.5)--(12.9904,15);
\draw(12.1244,15.5)--(12.9904,16);
\draw(12.1244,15.5)--(12.1244,16.5);
\draw[fill](14.7224,12)circle(0.2);
\draw(14.7224,12)--(15.5885,11.5);
\draw(14.7224,12)--(14.7224,13);
\draw[fill](12.9904,10)circle(0.2);
\draw(12.9904,10)--(12.9904,11);
\draw[fill](6.9282,14.5)circle(0.2);
\draw(6.9282,14.5)--(7.79423,14);
\draw(6.9282,14.5)--(7.79423,15);
\draw(6.9282,14.5)--(6.9282,15.5);
\draw[fill](6.06218,15)circle(0.2);
\draw(6.06218,15)--(6.9282,14.5);
\draw(6.06218,15)--(6.9282,15.5);
\draw(6.06218,15)--(6.06218,16);
\draw[fill](6.06218,16)circle(0.2);
\draw(6.06218,16)--(6.9282,15.5);
\draw(6.06218,16)--(6.9282,16.5);
\draw[fill](6.06218,7)circle(0.2);
\draw(6.06218,7)--(6.9282,6.5);
\draw(6.06218,7)--(6.9282,7.5);
\draw[fill](6.06218,10)circle(0.2);
\draw(6.06218,10)--(6.9282,10.5);
\draw[fill](3.4641,6.5)circle(0.2);
\draw(3.4641,6.5)--(4.33013,6);
\draw[fill](7.79423,3)circle(0.2);
\draw(7.79423,3)--(8.66025,3.5);
\draw(7.79423,3)--(7.79423,4);
\draw[fill](3.4641,5.5)circle(0.2);
\draw(3.4641,5.5)--(4.33013,5);
\draw(3.4641,5.5)--(4.33013,6);
\draw(3.4641,5.5)--(3.4641,6.5);
\draw[fill](2.59808,5)circle(0.2);
\draw(2.59808,5)--(3.4641,4.5);
\draw(2.59808,5)--(3.4641,5.5);
\draw[fill](5.19615,15.5)circle(0.2);
\draw(5.19615,15.5)--(6.06218,15);
\draw(5.19615,15.5)--(6.06218,16);
\draw(5.19615,15.5)--(5.19615,16.5);
\draw[fill](5.19615,6.5)circle(0.2);
\draw(5.19615,6.5)--(6.06218,6);
\draw(5.19615,6.5)--(6.06218,7);
\draw(5.19615,6.5)--(5.19615,7.5);
\draw[fill](4.33013,5)circle(0.2);
\draw(4.33013,5)--(5.19615,5.5);
\draw(4.33013,5)--(4.33013,6);
\draw[fill](12.9904,17)circle(0.2);
\draw[fill](6.06218,1)circle(0.2);
\draw(6.06218,1)--(6.06218,2);
\draw[fill](6.9282,7.5)circle(0.2);
\draw[fill](12.1244,17.5)circle(0.2);
\draw(12.1244,17.5)--(12.9904,17);
\draw(12.1244,17.5)--(12.1244,18.5);
\draw[fill](6.9282,11.5)circle(0.2);
\draw(6.9282,11.5)--(6.9282,12.5);
\draw[fill](15.5885,11.5)circle(0.2);
\draw(15.5885,11.5)--(16.4545,11);
\draw[fill](6.06218,13)circle(0.2);
\draw(6.06218,13)--(6.9282,12.5);
\draw[fill](10.3923,16.5)circle(0.2);
\draw[fill](2.59808,3)circle(0.2);
\draw(2.59808,3)--(3.4641,2.5);
\draw(2.59808,3)--(3.4641,3.5);
\draw[fill](7.79423,2)circle(0.2);
\draw(7.79423,2)--(7.79423,3);
\draw[fill](3.4641,8.5)circle(0.2);
\draw(3.4641,8.5)--(4.33013,9);
\draw(3.4641,8.5)--(3.4641,9.5);
\draw[fill](6.06218,5)circle(0.2);
\draw(6.06218,5)--(6.9282,4.5);
\draw(6.06218,5)--(6.06218,6);
\draw[fill](0.866025,11)circle(0.2);
\draw(0.866025,11)--(1.73205,10.5);
\draw[fill](9.52628,18)circle(0.2);
\draw[fill](6.9282,10.5)circle(0.2);
\draw(6.9282,10.5)--(6.9282,11.5);
\draw[fill](11.2583,15)circle(0.2);
\draw(11.2583,15)--(12.1244,14.5);
\draw(11.2583,15)--(12.1244,15.5);
\draw[fill](8.66025,15.5)circle(0.2);
\draw(8.66025,15.5)--(9.52628,16);
\draw(8.66025,15.5)--(8.66025,16.5);
\draw[fill](6.9282,2.5)circle(0.2);
\draw(6.9282,2.5)--(7.79423,2);
\draw(6.9282,2.5)--(7.79423,3);
\draw(6.9282,2.5)--(6.9282,3.5);
\draw[fill](8.66025,4.5)circle(0.2);
\draw(8.66025,4.5)--(9.52628,4);
\draw(8.66025,4.5)--(9.52628,5);
\draw[fill](1.73205,8.5)circle(0.2);
\draw(1.73205,8.5)--(2.59808,9);
\draw[fill](3.4641,3.5)circle(0.2);
\draw(3.4641,3.5)--(4.33013,3);
\draw(3.4641,3.5)--(3.4641,4.5);
\draw[fill](3.4641,4.5)circle(0.2);
\draw(3.4641,4.5)--(4.33013,5);
\draw(3.4641,4.5)--(3.4641,5.5);
\draw[fill](4.33013,6)circle(0.2);
\draw(4.33013,6)--(5.19615,5.5);
\draw(4.33013,6)--(5.19615,6.5);
\draw[fill](3.4641,9.5)circle(0.2);
\draw(3.4641,9.5)--(4.33013,9);
\draw(3.4641,9.5)--(4.33013,10);
\draw[fill](14.7224,13)circle(0.2);
\draw(14.7224,13)--(14.7224,14);
\draw[fill](2.59808,10)circle(0.2);
\draw(2.59808,10)--(3.4641,9.5);
\draw(2.59808,10)--(2.59808,11);
\draw[fill](13.8564,11.5)circle(0.2);
\draw(13.8564,11.5)--(14.7224,12);
\draw(13.8564,11.5)--(13.8564,12.5);
\draw[fill](5.19615,7.5)circle(0.2);
\draw(5.19615,7.5)--(6.06218,7);
\draw(5.19615,7.5)--(5.19615,8.5);
\draw[fill](9.52628,4)circle(0.2);
\draw(9.52628,4)--(10.3923,4.5);
\draw(9.52628,4)--(9.52628,5);
\draw[fill](5.19615,10.5)circle(0.2);
\draw(5.19615,10.5)--(6.06218,10);
\draw[fill](6.9282,6.5)circle(0.2);
\draw(6.9282,6.5)--(6.9282,7.5);
\draw[fill](15.5885,10.5)circle(0.2);
\draw(15.5885,10.5)--(16.4545,11);
\draw(15.5885,10.5)--(15.5885,11.5);
\draw[fill](12.1244,13.5)circle(0.2);
\draw(12.1244,13.5)--(12.9904,13);
\draw(12.1244,13.5)--(12.9904,14);
\draw(12.1244,13.5)--(12.1244,14.5);
\draw[fill](9.52628,16)circle(0.2);
\draw(9.52628,16)--(10.3923,15.5);
\draw(9.52628,16)--(10.3923,16.5);
\draw(9.52628,16)--(9.52628,17);
\draw[fill](5.19615,8.5)circle(0.2);
\draw[fill](1.73205,10.5)circle(0.2);
\draw(1.73205,10.5)--(2.59808,10);
\draw(1.73205,10.5)--(2.59808,11);
\draw[fill](0.866025,0)circle(0.2);
\draw(0.866025,0)--(0.866025,1);
\draw[fill](5.19615,16.5)circle(0.2);
\draw(5.19615,16.5)--(6.06218,16);
\draw[fill](1.73205,1.5)circle(0.2);
\draw(1.73205,1.5)--(2.59808,2);
\draw(1.73205,1.5)--(1.73205,2.5);
\draw[fill](12.9904,13)circle(0.2);
\draw(12.9904,13)--(13.8564,12.5);
\draw(12.9904,13)--(12.9904,14);
\draw[fill](9.52628,17)circle(0.2);
\draw(9.52628,17)--(10.3923,16.5);
\draw(9.52628,17)--(9.52628,18);
\draw[fill](12.1244,14.5)circle(0.2);
\draw(12.1244,14.5)--(12.9904,14);
\draw(12.1244,14.5)--(12.9904,15);
\draw(12.1244,14.5)--(12.1244,15.5);
\draw[fill](12.9904,16)circle(0.2);
\draw(12.9904,16)--(12.9904,17);
\draw[fill](6.9282,4.5)circle(0.2);
\draw(6.9282,4.5)--(7.79423,4);
\draw(6.9282,4.5)--(7.79423,5);
\draw[fill](5.19615,13.5)circle(0.2);
\draw(5.19615,13.5)--(6.06218,13);
\draw(5.19615,13.5)--(5.19615,14.5);
\draw[fill](1.73205,2.5)circle(0.2);
\draw(1.73205,2.5)--(2.59808,2);
\draw(1.73205,2.5)--(2.59808,3);
\draw[fill](0.866025,1)circle(0.2);
\draw(0.866025,1)--(1.73205,1.5);
\draw[fill](12.1244,11.5)circle(0.2);
\draw(12.1244,11.5)--(12.9904,11);
\draw(12.1244,11.5)--(12.9904,12);
\draw[fill](6.9282,16.5)circle(0.2);
\draw[fill](4.33013,3)circle(0.2);
\draw[fill](7.79423,5)circle(0.2);
\draw(7.79423,5)--(8.66025,4.5);
\draw[fill](16.4545,11)circle(0.2);
\draw[fill](6.9282,15.5)circle(0.2);
\draw(6.9282,15.5)--(7.79423,15);
\draw(6.9282,15.5)--(6.9282,16.5);
\draw[fill](12.9904,11)circle(0.2);
\draw(12.9904,11)--(13.8564,11.5);
\draw(12.9904,11)--(12.9904,12);
\draw[fill](14.7224,14)circle(0.2);
\draw[fill](4.33013,9)circle(0.2);
\draw(4.33013,9)--(5.19615,8.5);
\draw(4.33013,9)--(4.33013,10);
\draw[fill](8.66025,16.5)circle(0.2);
\draw(8.66025,16.5)--(9.52628,16);
\draw(8.66025,16.5)--(9.52628,17);
\draw[fill](9.52628,5)circle(0.2);
\draw(9.52628,5)--(10.3923,4.5);
\draw[fill](15.5885,9.5)circle(0.2);
\draw(15.5885,9.5)--(15.5885,10.5);
\end{tikzpicture} \hspace{1cm}
\begin{tikzpicture}[scale=0.25]
\draw[fill](6.06218,12.5)circle(0.2);
\draw(6.06218,12.5)--(6.9282,12);
\draw(6.06218,12.5)--(6.9282,13);
\draw[fill](0.866025,9.5)circle(0.2);
\draw(0.866025,9.5)--(1.73205,10);
\draw(0.866025,9.5)--(0.866025,10.5);
\draw[fill](8.66025,9)circle(0.2);
\draw[fill](0.866025,6.5)circle(0.2);
\draw(0.866025,6.5)--(0.866025,7.5);
\draw[fill](1.73205,22)circle(0.2);
\draw(1.73205,22)--(2.59808,21.5);
\draw(1.73205,22)--(2.59808,22.5);
\draw[fill](7.79423,0.5)circle(0.2);
\draw(7.79423,0.5)--(7.79423,1.5);
\draw[fill](0.866025,7.5)circle(0.2);
\draw(0.866025,7.5)--(0.866025,8.5);
\draw[fill](2.59808,20.5)circle(0.2);
\draw(2.59808,20.5)--(3.4641,20);
\draw(2.59808,20.5)--(2.59808,21.5);
\draw[fill](4.33013,7.5)circle(0.2);
\draw(4.33013,7.5)--(5.19615,8);
\draw[fill](2.59808,21.5)circle(0.2);
\draw(2.59808,21.5)--(2.59808,22.5);
\draw[fill](9.52628,12.5)circle(0.2);
\draw(9.52628,12.5)--(9.52628,13.5);
\draw[fill](2.59808,11.5)circle(0.2);
\draw(2.59808,11.5)--(2.59808,12.5);
\draw[fill](5.19615,20)circle(0.2);
\draw(5.19615,20)--(6.06218,19.5);
\draw(5.19615,20)--(5.19615,21);
\draw[fill](6.9282,9)circle(0.2);
\draw(6.9282,9)--(7.79423,8.5);
\draw(6.9282,9)--(7.79423,9.5);
\draw(6.9282,9)--(6.9282,10);
\draw[fill](4.33013,16.5)circle(0.2);
\draw(4.33013,16.5)--(5.19615,16);
\draw(4.33013,16.5)--(5.19615,17);
\draw[fill](5.19615,12)circle(0.2);
\draw(5.19615,12)--(6.06218,11.5);
\draw(5.19615,12)--(6.06218,12.5);
\draw[fill](10.3923,11)circle(0.2);
\draw(10.3923,11)--(11.2583,10.5);
\draw[fill](0,10)circle(0.2);
\draw(0,10)--(0.866025,9.5);
\draw(0,10)--(0.866025,10.5);
\draw[fill](2.59808,12.5)circle(0.2);
\draw(2.59808,12.5)--(2.59808,13.5);
\draw[fill](4.33013,11.5)circle(0.2);
\draw(4.33013,11.5)--(5.19615,11);
\draw(4.33013,11.5)--(5.19615,12);
\draw[fill](9.52628,11.5)circle(0.2);
\draw(9.52628,11.5)--(10.3923,11);
\draw(9.52628,11.5)--(9.52628,12.5);
\draw[fill](7.79423,1.5)circle(0.2);
\draw(7.79423,1.5)--(8.66025,2);
\draw(7.79423,1.5)--(7.79423,2.5);
\draw[fill](9.52628,7.5)circle(0.2);
\draw[fill](6.9282,12)circle(0.2);
\draw(6.9282,12)--(7.79423,12.5);
\draw(6.9282,12)--(6.9282,13);
\draw[fill](6.9282,7)circle(0.2);
\draw(6.9282,7)--(7.79423,6.5);
\draw(6.9282,7)--(7.79423,7.5);
\draw(6.9282,7)--(6.9282,8);
\draw[fill](5.19615,8)circle(0.2);
\draw(5.19615,8)--(6.06218,8.5);
\draw(5.19615,8)--(5.19615,9);
\draw[fill](3.4641,14)circle(0.2);
\draw(3.4641,14)--(3.4641,15);
\draw[fill](8.66025,2)circle(0.2);
\draw(8.66025,2)--(9.52628,2.5);
\draw(8.66025,2)--(8.66025,3);
\draw[fill](6.9282,2)circle(0.2);
\draw(6.9282,2)--(7.79423,1.5);
\draw(6.9282,2)--(7.79423,2.5);
\draw(6.9282,2)--(6.9282,3);
\draw[fill](6.06218,10.5)circle(0.2);
\draw(6.06218,10.5)--(6.9282,10);
\draw(6.06218,10.5)--(6.06218,11.5);
\draw[fill](3.4641,20)circle(0.2);
\draw(3.4641,20)--(4.33013,20.5);
\draw[fill](8.66025,3)circle(0.2);
\draw(8.66025,3)--(9.52628,2.5);
\draw(8.66025,3)--(8.66025,4);
\draw[fill](3.4641,23)circle(0.2);
\draw[fill](9.52628,13.5)circle(0.2);
\draw(9.52628,13.5)--(9.52628,14.5);
\draw[fill](5.19615,16)circle(0.2);
\draw(5.19615,16)--(5.19615,17);
\draw[fill](3.4641,9)circle(0.2);
\draw(3.4641,9)--(4.33013,9.5);
\draw[fill](7.79423,14.5)circle(0.2);
\draw[fill](0,9)circle(0.2);
\draw(0,9)--(0.866025,8.5);
\draw(0,9)--(0.866025,9.5);
\draw(0,9)--(0,10);
\draw[fill](4.33013,10.5)circle(0.2);
\draw(4.33013,10.5)--(5.19615,10);
\draw(4.33013,10.5)--(5.19615,11);
\draw(4.33013,10.5)--(4.33013,11.5);
\draw[fill](0.866025,5.5)circle(0.2);
\draw(0.866025,5.5)--(0.866025,6.5);
\draw[fill](6.9282,8)circle(0.2);
\draw(6.9282,8)--(7.79423,7.5);
\draw(6.9282,8)--(7.79423,8.5);
\draw(6.9282,8)--(6.9282,9);
\draw[fill](3.4641,19)circle(0.2);
\draw(3.4641,19)--(3.4641,20);
\draw[fill](0,7)circle(0.2);
\draw(0,7)--(0.866025,6.5);
\draw(0,7)--(0.866025,7.5);
\draw[fill](6.9282,0)circle(0.2);
\draw(6.9282,0)--(7.79423,0.5);
\draw(6.9282,0)--(6.9282,1);
\draw[fill](7.79423,6.5)circle(0.2);
\draw(7.79423,6.5)--(7.79423,7.5);
\draw[fill](9.52628,4.5)circle(0.2);
\draw[fill](2.59808,9.5)circle(0.2);
\draw(2.59808,9.5)--(3.4641,9);
\draw[fill](8.66025,8)circle(0.2);
\draw(8.66025,8)--(9.52628,7.5);
\draw(8.66025,8)--(8.66025,9);
\draw[fill](6.9282,15)circle(0.2);
\draw(6.9282,15)--(7.79423,14.5);
\draw[fill](7.79423,12.5)circle(0.2);
\draw(7.79423,12.5)--(8.66025,12);
\draw(7.79423,12.5)--(7.79423,13.5);
\draw[fill](7.79423,13.5)circle(0.2);
\draw(7.79423,13.5)--(7.79423,14.5);
\draw[fill](7.79423,4.5)circle(0.2);
\draw(7.79423,4.5)--(8.66025,4);
\draw(7.79423,4.5)--(8.66025,5);
\draw[fill](2.59808,22.5)circle(0.2);
\draw(2.59808,22.5)--(3.4641,23);
\draw[fill](8.66025,5)circle(0.2);
\draw(8.66025,5)--(9.52628,4.5);
\draw[fill](5.19615,9)circle(0.2);
\draw(5.19615,9)--(6.06218,8.5);
\draw(5.19615,9)--(6.06218,9.5);
\draw(5.19615,9)--(5.19615,10);
\draw[fill](3.4641,16)circle(0.2);
\draw(3.4641,16)--(4.33013,15.5);
\draw(3.4641,16)--(4.33013,16.5);
\draw[fill](1.73205,11)circle(0.2);
\draw(1.73205,11)--(2.59808,11.5);
\draw(1.73205,11)--(1.73205,12);
\draw[fill](1.73205,10)circle(0.2);
\draw(1.73205,10)--(2.59808,9.5);
\draw(1.73205,10)--(1.73205,11);
\draw[fill](6.06218,9.5)circle(0.2);
\draw(6.06218,9.5)--(6.9282,9);
\draw(6.06218,9.5)--(6.9282,10);
\draw(6.06218,9.5)--(6.06218,10.5);
\draw[fill](1.73205,12)circle(0.2);
\draw(1.73205,12)--(2.59808,11.5);
\draw(1.73205,12)--(2.59808,12.5);
\draw[fill](4.33013,9.5)circle(0.2);
\draw(4.33013,9.5)--(5.19615,9);
\draw(4.33013,9.5)--(5.19615,10);
\draw(4.33013,9.5)--(4.33013,10.5);
\draw[fill](11.2583,10.5)circle(0.2);
\draw[fill](6.06218,11.5)circle(0.2);
\draw(6.06218,11.5)--(6.9282,12);
\draw(6.06218,11.5)--(6.06218,12.5);
\draw[fill](5.19615,19)circle(0.2);
\draw(5.19615,19)--(6.06218,18.5);
\draw(5.19615,19)--(6.06218,19.5);
\draw(5.19615,19)--(5.19615,20);
\draw[fill](1.73205,15)circle(0.2);
\draw[fill](9.52628,14.5)circle(0.2);
\draw[fill](5.19615,10)circle(0.2);
\draw(5.19615,10)--(6.06218,9.5);
\draw(5.19615,10)--(6.06218,10.5);
\draw(5.19615,10)--(5.19615,11);
\draw[fill](6.06218,18.5)circle(0.2);
\draw(6.06218,18.5)--(6.06218,19.5);
\draw[fill](4.33013,20.5)circle(0.2);
\draw(4.33013,20.5)--(5.19615,20);
\draw(4.33013,20.5)--(5.19615,21);
\draw[fill](7.79423,9.5)circle(0.2);
\draw(7.79423,9.5)--(8.66025,9);
\draw[fill](6.9282,5)circle(0.2);
\draw(6.9282,5)--(7.79423,4.5);
\draw(6.9282,5)--(6.9282,6);
\draw[fill](5.19615,17)circle(0.2);
\draw(5.19615,17)--(5.19615,18);
\draw[fill](0.866025,8.5)circle(0.2);
\draw(0.866025,8.5)--(0.866025,9.5);
\draw[fill](5.19615,21)circle(0.2);
\draw[fill](7.79423,8.5)circle(0.2);
\draw(7.79423,8.5)--(8.66025,8);
\draw(7.79423,8.5)--(8.66025,9);
\draw(7.79423,8.5)--(7.79423,9.5);
\draw[fill](1.73205,14)circle(0.2);
\draw(1.73205,14)--(2.59808,13.5);
\draw(1.73205,14)--(1.73205,15);
\draw[fill](0.866025,11.5)circle(0.2);
\draw(0.866025,11.5)--(1.73205,11);
\draw(0.866025,11.5)--(1.73205,12);
\draw[fill](6.9282,14)circle(0.2);
\draw(6.9282,14)--(7.79423,13.5);
\draw(6.9282,14)--(7.79423,14.5);
\draw(6.9282,14)--(6.9282,15);
\draw[fill](2.59808,13.5)circle(0.2);
\draw(2.59808,13.5)--(3.4641,14);
\draw[fill](7.79423,3.5)circle(0.2);
\draw(7.79423,3.5)--(8.66025,3);
\draw(7.79423,3.5)--(8.66025,4);
\draw(7.79423,3.5)--(7.79423,4.5);
\draw[fill](5.19615,18)circle(0.2);
\draw(5.19615,18)--(6.06218,18.5);
\draw(5.19615,18)--(5.19615,19);
\draw[fill](6.06218,8.5)circle(0.2);
\draw(6.06218,8.5)--(6.9282,8);
\draw(6.06218,8.5)--(6.9282,9);
\draw(6.06218,8.5)--(6.06218,9.5);
\draw[fill](8.66025,12)circle(0.2);
\draw(8.66025,12)--(9.52628,11.5);
\draw(8.66025,12)--(9.52628,12.5);
\draw[fill](5.19615,15)circle(0.2);
\draw(5.19615,15)--(5.19615,16);
\draw[fill](6.9282,13)circle(0.2);
\draw(6.9282,13)--(7.79423,12.5);
\draw(6.9282,13)--(7.79423,13.5);
\draw(6.9282,13)--(6.9282,14);
\draw[fill](0.866025,21.5)circle(0.2);
\draw(0.866025,21.5)--(1.73205,22);
\draw[fill](6.06218,6.5)circle(0.2);
\draw(6.06218,6.5)--(6.9282,6);
\draw(6.06218,6.5)--(6.9282,7);
\draw[fill](6.9282,1)circle(0.2);
\draw(6.9282,1)--(7.79423,0.5);
\draw(6.9282,1)--(7.79423,1.5);
\draw(6.9282,1)--(6.9282,2);
\draw[fill](0.866025,10.5)circle(0.2);
\draw(0.866025,10.5)--(1.73205,10);
\draw(0.866025,10.5)--(1.73205,11);
\draw(0.866025,10.5)--(0.866025,11.5);
\draw[fill](7.79423,7.5)circle(0.2);
\draw(7.79423,7.5)--(8.66025,8);
\draw(7.79423,7.5)--(7.79423,8.5);
\draw[fill](5.19615,11)circle(0.2);
\draw(5.19615,11)--(6.06218,10.5);
\draw(5.19615,11)--(6.06218,11.5);
\draw(5.19615,11)--(5.19615,12);
\draw[fill](7.79423,2.5)circle(0.2);
\draw(7.79423,2.5)--(8.66025,2);
\draw(7.79423,2.5)--(8.66025,3);
\draw(7.79423,2.5)--(7.79423,3.5);
\draw[fill](6.9282,10)circle(0.2);
\draw(6.9282,10)--(7.79423,9.5);
\draw[fill](6.9282,6)circle(0.2);
\draw(6.9282,6)--(7.79423,6.5);
\draw(6.9282,6)--(6.9282,7);
\draw[fill](6.06218,19.5)circle(0.2);
\draw[fill](8.66025,4)circle(0.2);
\draw(8.66025,4)--(9.52628,4.5);
\draw(8.66025,4)--(8.66025,5);
\draw[fill](4.33013,15.5)circle(0.2);
\draw(4.33013,15.5)--(5.19615,15);
\draw(4.33013,15.5)--(5.19615,16);
\draw(4.33013,15.5)--(4.33013,16.5);
\draw[fill](9.52628,2.5)circle(0.2);
\draw[fill](6.9282,3)circle(0.2);
\draw(6.9282,3)--(7.79423,2.5);
\draw(6.9282,3)--(7.79423,3.5);
\draw[fill](3.4641,15)circle(0.2);
\draw(3.4641,15)--(4.33013,15.5);
\draw(3.4641,15)--(3.4641,16);
\end{tikzpicture} \\
(a) \hspace{5cm} (b)

\caption{$100$ particles in a line with occupied edges drawn, after (a) $10$ million and (b) $20$ million iterations of $\M$ with bias $\lambda = 2$.}
\label{fig:bias2_notCompressed}
\end{figure}

Analogous to our definition of $\alpha$-compression, we say a configuration $\sigma$ is $\beta$-{\em expanded} for some $0 < \beta < 1$ if $p(\sigma) > \beta \cdot \pmax$.
For a configuration of $n$ particles, $\pmax = 2n-2 = \Theta(n)$ and $\pmin = \Theta(\sqrt{n})$, so $\beta$-expansion and $\alpha$-compression for any constants $\beta$ and $\alpha$ are mutually exclusive for sufficiently large $n$.
We prove in this section that, for all $0 < \lambda < \ebound$ and provided $n$ is large enough, there is a constant $\beta$ such that a configuration chosen at random according to the stationary distribution of $\M$ is $\beta$-expanded with all but exponentially small probability.
As mentioned above, this is notable because it implies that $\lambda > 1$ (i.e., favoring more neighbors) is not sufficient to guarantee compression as one might first guess.

We begin with some preliminaries about counting the number of particle system configurations with a certain perimeter, which will give us a lower bound on the partition function $Z$.
We then use this bound to show that for all $\lambda < \sqrt{2}$, expansion occurs.
By revisiting and improving this lower bound on $Z$, we can improve this result and show expansion occurs for all $\lambda < 2.17$.

\subsection{A Non-trivial Lower Bound on the Partition Function} \label{subsec:nontrivZbound}

Let $S^\beta$ be the set of all connected, hole-free particle system configurations with perimeter at most $\beta \cdot \pmax$ for some constant $0 < \beta < 1$.
Analogous to the approach for proving compression, we want to show $\pi(S^\beta) = \sum_{\sigma \in S^\beta} \lambda^{-p(\sigma)} / Z$, the stationary probability of being in a configuration with small perimeter, is exponentially small.
Recall that Corollary~\ref{cor:stat-perim} defined the partition function as $Z = \sum_{\sigma \in \Ohf} \lambda^{-p(\sigma)}$, the summed weight of all connected, hole-free configurations.
The critical component of this result is an improved lower bound on $Z$; the trivial bound of $Z \geq \lambda^{-\pmin}$ used for compression does not suffice for expansion.
We give our first non-trivial lower bound on $Z$ in Lemma~\ref{lem:lambda<1}, and this result is valid for all $\lambda > 0$.
Later, we will obtain an improved lower bound on $Z$ that is valid for all $\lambda \geq 1$.

Obtaining these lower bounds on $Z$ for expansion requires a lower bound on the number of configurations with $n$ particles and a given perimeter; this is the opposite of what we did for compression, where we upper bounded this quantity.
To begin, we recall $\pmax = 2n-2$ and note:
\[Z = \sum_{\sigma \in \Ohf} \lambda^{-p(\sigma)} \geq \sum_{\scriptsize \begin{array}{c} \sigma \in \Ohf : \\ p(\sigma) = \pmax \end{array}} \lambda^{-p(\sigma)}
= c_{2n-2} \lambda^{-(2n-2)},\]
where $c_{2n-2}$ is the number of configurations with $n$ particles and perimeter exactly $2n-2$.
Note if a configuration $\sigma$ with $n$ particles has perimeter $2n-2$, then by Lemmas~\ref{lem:edge=perim} and~\ref{lem:tri=perim} it must be that $\sigma$ has exactly $n-1$ edges and no triangles; that is, $\sigma$ is an induced tree in $\Gtri$.
We present a method for enumerating a subset of these trees, giving a lower bound on $c_{2n-2}$.

\begin{lem} \label{lem:lambda<1}
For any $\lambda > 0$, $Z \geq (\sqrt{2}/\lambda)^{\pmax}$.
\end{lem}
\begin{proof}
We enumerate $n$-vertex paths in $\Gtri$ where every step is either down-right or up-right; this is a subset of the trees contributing to $c_{2n-2}$.
Starting from the first particle, there are $2^{n-1}$ ways to place rest of the particles to form such a path, where each one is either up-right or down-right from the previous one.
This means there are at least $2^{n-1}$ such paths, giving
\[c_{2n-2} \geq 2^{n-1} = \sqrt{2}^{2n-2} = \sqrt{2}^{\pmax}.\]
From this, it follows that
\[Z = \sum_{\sigma \in \Ohf} \lambda^{-p(\sigma)} \geq \sum_{\scriptsize \begin{array}{c} \sigma \in \Ohf : \\ p(\sigma) = \pmax \end{array}} \lambda^{-\pmax} \geq \sqrt{2}^{\pmax} \lambda^{-\pmax} = \left(\frac{\sqrt{2}}{\lambda}\right)^{\pmax}.\]
\end{proof}

As the next result will show, Lemma~\ref{lem:lambda<1} directly implies the particle system does not compress, even in the limit, for any $\lambda < \sqrt{2}$.
This bound could be improved significantly with a better lower bound for $c_{2n-2}$, but this will be eclipsed by the lower bound for $Z$ when $\lambda \geq 1$ given in Section~\ref{subsec:improvedZbound}.

\subsection{Proof of Expansion} \label{subsec:expproof}

We now show, using Lemma~\ref{lem:lambda<1}, that for any value of $\beta \in (0,1)$ it is possible to achieve $\beta$-expansion by simply running $\M$ with input parameter $\lambda$ sufficiently small.
The closer $\beta$ is to $1$, the closer $\lambda$ must be to $0$ in order to achieve $\beta$-expansion.

\begin{thm} \label{thm:expand_beta}
For any $0 < \beta < 1$, let $\lambda^* = \min\left(\sqrt{2}, \sqrt{2}^{\frac{1}{1-\beta}} (\cbound)^{\frac{-\beta}{1-\beta}}\right)$. There exists $n^* \geq 0$ and $\zeta < 1$ such that for all $\lambda < \lambda^*$ and $n \geq n^*$, the probability that a random sample $\sigma$ drawn according to the stationary distribution $\pi$ of $\M$ is not $\beta$-expanded is exponentially small:
\[\p_{\sigma \sim \pi}\left(p(\sigma) \leq \beta \cdot \pmax\right) < \zeta^{\sqrt{n}}.\]
\end{thm}
\begin{proof}
Because $\lambda < \lambda^*$, we know $\lambda < \sqrt{2}^{\frac{1}{1-\beta}} (\cbound)^{\frac{-\beta}{1-\beta}}$. Rearranging terms in this expression, we obtain $\lambda^{(\beta-1)/\beta} 2^{1/2\beta} > 2+\sqrt{2}$.
Let $\nu$ be any value satisfying $\lambda^{(\beta-1)/\beta} 2^{1/2\beta} > \nu > 2+\sqrt{2}$.
We will later use the fact that for any such choice of $\nu$,
\[\nu \lambda^{\frac{1-\beta}{\beta}} 2^{-\frac{1}{2\beta}} < 1.\]

Let $S^\beta$ be the set of configurations of perimeter at most $\beta \cdot \pmax$.
We wish to show that $\pi(S^\beta)$ is smaller than some function that is exponentially small in $n$.
Using Lemma~\ref{lem:lambda<1} to upper bound the partition function $Z$ of the stationary distribution $\pi$, we have
\[\pi\big(S^\beta\big) = \frac{w\big(S^\beta\big)}{Z}
\leq \frac{w\big(S^\beta\big)}{\big(\sqrt{2} / \lambda\big)^{\pmax}}
= w\big(S^\beta\big)\left(\frac{\lambda}{\sqrt{2}}\right)^{\pmax}.\]

The remainder of this proof will be spent finding an upper bound on the right hand side of the above equation that is exponentially small in $n$.
To begin, we stratify $S^\beta$ into sets of configurations that have the same perimeter.
Let $B_k$ be the set of all configurations with perimeter $k$; then $S^\beta = \bigcup_{k = \pmin}^{\lfloor\beta \cdot \pmax \rfloor} B_k$.
We can then write
\[w\big(S^\beta\big)\left(\frac{\lambda}{\sqrt{2}}\right)^{\pmax} = \sum_{k = \pmin}^{\lfloor \beta \cdot \pmax \rfloor} w(B_k) \left(\frac{\lambda}{\sqrt{2}}\right)^{\pmax}.\]

The weight of each element in the set $B_k$ is the same, $\lambda^{-k}$.
By Lemma~\ref{lem:nu} and our careful choice of $\nu$, above, the number of configurations in set $B_k$ is at most $\nu^k$ provided $k$ is sufficiently large.
So,
\[w\big(S^\beta\big)\left(\frac{\lambda}{\sqrt{2}}\right)^{\pmax}
\leq \sum_{k = \pmin}^{\lfloor \beta \cdot \pmax \rfloor} \nu^k\lambda^{-k} \left(\frac{\lambda}{\sqrt{2}}\right)^{\pmax}.\]

Because $k \leq \beta \cdot \pmax$, we have $\pmax \geq k/\beta$.
As $\lambda < \lambda^* \leq\sqrt{2}$, we have
\[w\big(S^\beta\big)\left(\frac{\lambda}{\sqrt{2}}\right)^{\pmax}
\leq \sum_{k = \pmin}^{\lfloor \beta \cdot \pmax \rfloor} \nu^k\lambda^{-k} \left(\frac{\lambda}{\sqrt{2}}\right)^{\frac{k}{\beta}}
= \sum_{k = \pmin}^{\lfloor \beta \cdot \pmax \rfloor}
\left( \frac{ \nu \lambda^{\frac{1-\beta}{\beta}}}{\sqrt{2}^{\frac{1}{\beta}}} \right)^k.\]

Because of our choice of $\nu$, the rightmost term in parentheses is less than one.
By applying the inequalities $\pmax = 2n-2 \geq k \geq \pmin > \sqrt{n}$ (by Lemma~\ref{lem:pmin}), we see that
\[w\big(S^\beta\big)\left(\frac{\lambda}{\sqrt{2}}\right)^{\pmax}
\leq \sum_{k = \pmin}^{\lfloor \beta \cdot \pmax \rfloor}
\left( \frac{ \nu \lambda^{\frac{1-\beta}{\beta}} }{\sqrt{2}^{\frac{1}{\beta}}} \right)^{\sqrt{n}}
\leq (2n-2) \left( \frac{ \nu \lambda^{\frac{1-\beta}{\beta}} }{\sqrt{2}^{\frac{1}{\beta}}} \right)^{\sqrt{n}}.\]

Again using the fact that the rightmost term in parentheses is less than one, we can find a constant $\zeta < 1$ and an $n^*$ such that for all $n \geq n^*$,
\[\p_{\sigma \sim \pi}\left(p(\sigma) \leq \beta \cdot \pmax\right)
= \pi\big(S^\beta\big)
\leq w\big(S^\beta\big)\left(\frac{\lambda}{\sqrt{2}}\right)^{\pmax}
\leq (2n-2) \left( \frac{ \nu \lambda^{\frac{1-\beta}{\beta}} }{\sqrt{2}^{\frac{1}{\beta}}} \right)^{\sqrt{n}}
< \zeta^{\sqrt{n}}.\]
\end{proof}

While the above result shows that $\M$ accomplishes $\beta$-expansion for any $\beta < 1$, larger values of $\beta$ require smaller values of $\lambda$.
However, larger values of $\lambda$ are still of interest as we wish to characterize how the behavior of $\M$ and $\A$ depends on $\lambda$.
We now show that provided $\lambda < \sqrt{2}$, there is some constant $\beta$ such that $\beta$-expansion occurs.
Of course, there is again a tradeoff: the larger $\lambda$ is, the smaller $\beta$ is.

\begin{cor} \label{cor:expand_lambda}
For all $0 < \lambda < \sqrt{2}$, for any constant $\beta < \frac{ \log \sqrt{2} - \log \lambda }{\log(2+\sqrt{2}) - \log \lambda}$, there exists $n^* \geq 0$ and $\zeta < 1$ such that for all $n \geq n^*$, a random sample $\sigma$ drawn according to the stationary distribution $\pi$ of $\M$ satisfies
\[\p_{\sigma \sim \pi}\left(p(\sigma) \leq \beta \cdot \pmax\right) < \zeta^{\sqrt{n}}.\]
\end{cor}
\begin{proof}
Theorem~\ref{thm:expand_beta} applies whenever $\lambda < \sqrt{2}^{\frac{1}{1-\beta}}(\cbound)^{\frac{-\beta}{1-\beta}}$.
Solving for $\beta$, we see the theorem applies whenever $\beta < \frac{ \log \sqrt{2} - \log \lambda }{\log(2+\sqrt{2}) - \log \lambda}$, as desired.
\end{proof}

This proves the counterintuitive result that $\lambda > 1$ is not sufficient to guarantee compression.
While $\lambda > 1$ guarantees that configurations with smaller perimeter have higher weight at stationarity, our work in this section shows that there are so many configurations with large perimeter that, for $\lambda < \sqrt{2}$, these large perimeter configurations dominate the stationary distribution.
Raising $\lambda$ above $\cbound$, we observe an energy/entropy tradeoff.
In this regime, the high energy (small perimeter) configurations dominate the state space as opposed to those with high entropy (large perimeter), yielding compression.

\subsection{An Improved Lower Bound on the Partition Function} \label{subsec:improvedZbound}

We can improve the bound of $\lambda < \sqrt{2}$ appearing in Corollary~\ref{cor:expand_lambda} by finding a better lower bound on the partition function $Z$.
When we know $\lambda \geq 1$, the improved bounds in Lemma~\ref{lem:lambda>1} can be used to show $\beta$-expansion occurs for an even greater range of values for $\lambda$, $\lambda < 2.17$.
Again, larger values of $\lambda$ necessitate smaller, but still constant, values of $\beta$.

To get an improved bound on $Z$, the key observation is that when $\lambda \geq 1$, any value $k < 2n-2$ satisfies $\lambda^{-k} \geq \lambda^{-(2n-2)}$.
Thus, as $\pmax = 2n-2$, it follows that
\[Z = \sum_{\sigma \in \Ohf} \lambda^{-p(\sigma)}
\geq \sum_{\sigma \in \Ohf} \lambda^{-(2n-2)} = |\Ohf| \cdot \lambda^{-(2n-2)},\]
where the sums are over all connected, hole-free particle system configurations with $n$ particles.
Thus, it suffices to find a lower bound on the total number of connected, hole-free configurations with $n$ particles and any perimeter, instead of only counting the number of configurations with maximum perimeter as we did in Section~\ref{subsec:nontrivZbound}.
Leveraging this observation will yield a better lower bound on $Z$ than in the previous case where $\lambda$ was unrestricted.

\begin{figure}
\centering
\subfloat[]{
    \includegraphics[scale = 0.8, page = 1]{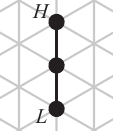}
	\label{fig:3particles_1}
} \hfil
\subfloat[]{
    \includegraphics[scale = 0.8, page = 2]{3particles.pdf}
	\label{fig:3particles_2}
} \hfil
\subfloat[]{
    \includegraphics[scale = 0.8, page = 3]{3particles.pdf}
	\label{fig:3particles_3}
} \hfil
\subfloat[]{
    \includegraphics[scale = 0.8, page = 4]{3particles.pdf}
	\label{fig:3particles_4}
} \hfil
\subfloat[]{
    \includegraphics[scale = 0.8, page = 5]{3particles.pdf}
	\label{fig:3particles_5}
} \hfil
\subfloat[]{
    \includegraphics[scale = 0.8, page = 6]{3particles.pdf}
	\label{fig:3particles_6}
} \\
\subfloat[]{
    \includegraphics[scale = 0.8, page = 7]{3particles.pdf}
	\label{fig:3particles_7}
} \hfil
\subfloat[]{
    \includegraphics[scale = 0.8, page = 8]{3particles.pdf}
	\label{fig:3particles_8}
} \hfil
\subfloat[]{
    \includegraphics[scale = 0.8, page = 9]{3particles.pdf}
	\label{fig:3particles_9}
} \hfil
\subfloat[]{
    \includegraphics[scale = 0.8, page = 10]{3particles.pdf}
	\label{fig:3particles_10}
} \hfil
\subfloat[]{
    \includegraphics[scale = 0.8, page = 11]{3particles.pdf}
	\label{fig:3particles_11}
}
\caption{All $11$ connected hole-free configurations with three particles. In each, the highest leftmost particle is labeled $H$, and the lowest leftmost particle is labeled $L$; when there is only one leftmost particle $H = L$.}
\label{fig:3particles}
\end{figure}

\begin{lem} \label{lem:1.8}
For $\lambda \geq 1$, $Z \geq 0.12 \cdot (1.67/\lambda)^{\pmax}$.
\end{lem}
\begin{proof}
We first give a lower bound on the number of connected, hole-free configurations with $n$ particles by iteratively enumerating a subset of them.
Note there are $11$ connected, hole-free configurations with exactly $3$ particles; all $11$ are shown in Fig.~\ref{fig:3particles}.

Given some hole-free configuration $\sigma$ with $1 + 3j$ particles, $j \geq 0$, we show how to enumerate $22$ distinct hole-free configurations of $4 + 3j$ particles, each consisting of three particles added to the right of $\sigma$.
Let $P$ be the highest rightmost particle of $\sigma$ and let $Q$ be the lowest rightmost particle of $\sigma$; possibly $P = Q$.
Choose any of the $11$ hole-free configurations with $3$ particles, and let $L$ be its lowest leftmost particle and $H$ be its highest leftmost particle as in Fig.~\ref{fig:3particles}; possibly $H = L$.
Attach this configuration to $\sigma$ either by placing $H$ below and right of $Q$ or by placing $L$ above and right of $P$; see Fig.~\ref{fig:lowerbound} for two such examples.
Note even if $Q = P$ and $H = L$, this still results in two distinct attachments.
In the first case, all locations directly below $Q$ and all locations directly above $H$ are unoccupied; this ensures the only adjacency between $\sigma$ and the newly added three particles is between $Q$ and $H$, meaning no holes have been created.
Similarly in the second case, all locations above $P$ or below $L$ are unoccupied, again ensuring no holes form.

\begin{figure}
\centering
\subfloat[]{
    \includegraphics[scale = 0.8, page = 1]{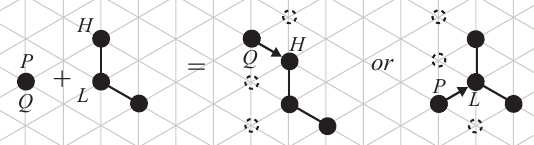}
	\label{fig:lowerbound1}
} \\
\subfloat[]{
    \includegraphics[scale = 0.8, page = 2]{tree_iterations.pdf}
	\label{fig:lowerbound2}
}
\caption{The iterative process of Lemma~\ref{lem:1.8}. (a) One of the 11 connected hole-free configurations on 3 vertices, and the two ways it can attach to the single particle with which the iterative process begins. (b) Another of the 11 connected hole-free configurations on 3 vertices, and the two ways it can attach to a configuration with four particles.}
\label{fig:lowerbound}
\end{figure}

Using this process and beginning with a single particle (as in Fig.~\ref{fig:lowerbound1}), we can enumerate $22^j$ distinct configurations with $1+3j$ particles for all $j \geq 0$.
This does not enumerate all configurations on $1+3j$ particles: for example, there are $42$ configurations on $4$ particles and this process only enumerates $22$ of them.
However, this process iterates nicely and produces reasonable lower bounds as the number of particles gets large.

To get a lower bound on the number of configurations of $n$ particles when $n \not\equiv 1~(\text{mod } 3)$, we can simply enumerate all configurations on $1 + 3j \leq n$ particles for $j = \left\lfloor \frac{n-1}{3}\right\rfloor$, and add one or two particles to each in some deterministic way.
We conclude that for any $n$, the number of connected, hole-free configurations of $n$ particles is at least
\[22^{\left\lfloor \frac{n-1}{3}\right\rfloor} \geq 22^{\frac{n-1}{3}} \cdot 22^{-2/3}
= 22^{-2/3} (22^{1/6})^{2n-2}
> 0.12 \cdot 1.67^{2n-2}.\]

Using this bound, it follows that
\[Z = \sum_{\sigma \in \Ohf} \lambda^{-p(\sigma)}
\geq \sum_{\sigma \in \Ohf} \lambda^{-(2n-2)}
= |\Ohf| \cdot \lambda^{-(2n-2)}
> 0.12 \cdot 1.67^{2n-2} \cdot \lambda^{-(2n-2)}
= 0.12 \cdot\left(\frac{1.67}{\lambda}\right)^{\pmax}.\]
\end{proof}

This bound can be improved even further by considering configurations of $50$ particles instead of configurations of three particles.
A result by Jensen~\cite{Jensen2009} will be essential.
In that paper, the author presents a parallel algorithm efficient enough to count the number of benzenoid hydrocarbons containing $h$ hexagonal cells up to $h = 50$.
A benzenoid hydrocarbon containing $h$ hexagonal cells is exactly equivalent to a connected, hole-free particle system configuration with $h$ particles, implying the following.

\begin{lem}[\cite{Jensen2009}]
The number of connected, hole-free particle system configurations with $50$ particles is
\[N_{50} = \numprint{2430068453031180290203185942420933}.\]
\end{lem}

\begin{lem} \label{lem:lambda>1}
For $\lambda \geq 1$, $Z \geq 0.13 \cdot (2.17/\lambda)^{\pmax}$.
\end{lem}
\begin{proof}
We use the same approach as in Lemma~\ref{lem:1.8}, noting that $2.17 \sim (2N_{50})^{1/100}$.
To get a lower bound on the number of configurations with $n$ particles, first write $n$ as $n = 1 + 50i + j$, where $i,j \in \mathbb{Z}_{\geq 0}$ and $j < 50$; subject to these requirements, $i$ and $j$ are unique.
Iteratively construct one particle configuration $\sigma$ with $n$ particles by beginning with a single particle and repeatedly attaching one of the $N_{50}$ configurations with $50$ particles to the right as in the proof of Lemma~\ref{lem:1.8}: place its highest leftmost particle $H$ below and right of the existing configuration's lowest rightmost particle $Q$, or place its lowest leftmost particle $L$ above and right of the existing configuration's highest rightmost particle $P$.
This process, applied $i$ times, yields a connected, hole-free configuration of $1 + 50i = n - j$ particles.
There are then $2N_j$ ways, following the same procedure, to attach the remaining $j$ particles to form a configuration of $n$ particles.
In this way, we can enumerate $(2N_{50})^i \cdot 2 N_j$ unique connected, hole-free configurations of $n$ particles.
It follows that the number of connected, hole-free configurations of $n$ particles is at least
\[(2N_{50})^i \cdot 2N_j = (2N_{50})^{\frac{n-1-j}{50}} \cdot 2N_j
= (2N_{50})^{\frac{n-1}{50}} \cdot (2N_{50})^{-\frac{j}{50}} \cdot 2N_j.\]

Calculations show that for all $0 \leq j < 50$, $(2N_{50})^{-j/50} \cdot 2N_j \geq 0.13$.
It follows that the number of connected, hole-free configurations of $n$ particles is at least
\[0.13 \cdot \left(\left(2N_{50}\right)^{1/100}\right)^{2n-2} = 0.13 \cdot \left( \left(2N_{50}\right)^{1/100} \right)^{\pmax}.\]

Noting that $(2N_{50})^{1/100} > 2.17$, it follows that
\[Z \geq \sum_{\sigma \in \Ohf} \lambda^{-(2n-2)}
= |\Ohf| \cdot \lambda^{-(2n-2)}
\geq 0.13 \cdot (2.17)^{\pmax} \lambda^{-\pmax}
= 0.13 \cdot (2.17/\lambda)^{\pmax}.\]
\end{proof}

As we will see next, this will directly imply that the particle system will not exhibit compression for any $\lambda < 2.17$.
We expect this bound will improve given accurate counts of the number of connected, hole-free configurations for even larger numbers of particles.
Computationally this seems infeasible, and a careful analysis of the work done in \cite{Jensen2009} suggests the best bound achievable by this method would be expansion for all $\lambda < 2.27$, only a mild improvement and still far from the known lower bound for compression, $\lambda > \cbound$.

\subsection{Proofs of Expansion for a larger range of \texorpdfstring{$\lambda$}{Lambda}}

We now show, using Lemma~\ref{lem:lambda>1}, that it is possible to achieve $\beta$-expansion (i.e., compression does not occur) using any value of $\lambda$ up to $2.17$.

\begin{thm} \label{thm:2.17}
For all $1 \leq \lambda < \ex := (2N_{50})^{1/100} \sim 2.17$, for any $\beta < \frac{\log \ex - \log \lambda}{\log(\cbound) - \log \lambda }$ there exists $n^* \geq 0$ and $\zeta < 1$ such that for all $n \geq n^*$, a random sample $\sigma$ drawn according to the stationary distribution $\pi$ of $\M$ satisfies
\[\p_{\sigma \sim \pi}\left(p(\sigma) \leq \beta \cdot \pmax\right) < \zeta^{\sqrt{n}}.\]
\end{thm}
\begin{proof}
First, note that the condition
\[\beta < \frac{\log \ex - \log \lambda}{\log\big(\cbound\big) -\log \lambda}
= \frac{\log (\ex/\lambda)}{\log \big(\big(\cbound\big)/ \lambda\big)}\]
can be equivalently expressed as
$\cbound < \lambda^{(\beta - 1) / \beta} \ex^{1 / \beta}$.
Pick $\nu$ to be between these two values, so that $\cbound < \nu < \lambda^{(\beta - 1) / \beta} \ex^{1 / \beta}$.

Let $S^\beta$ be the set of configurations of perimeter at most $\beta \cdot \pmax$.
We wish to show that $\pi(S^\beta)$ is smaller than some function that is exponentially small in $n$.
Applying Lemma~\ref{lem:lambda>1}, which gives a lower bound on the partition function $Z$ of stationary distribution $\pi$, we see that
\[\pi\big(S^\beta\big) = \frac{w\big(S^\beta\big)}{Z}
\leq \frac{w\big(S^\beta\big)}{0.13 \left(\frac{\ex}{\lambda}\right)^{\pmax}}
\leq 8 \ w\big(S^\beta\big)\left(\frac{\lambda}{\ex}\right)^{\pmax}.\]

The remainder of this proof will be spent finding an upper bound on the right hand side of the above equation that is exponentially small in $n$.
To begin, we stratify $S^\beta$ into sets of configurations that have the same perimeter.
Let $B_k$ be the set of all configurations with perimeter $k$; then $S^\beta = \bigcup_{k = \pmin}^{\lfloor \beta \cdot \pmax \rfloor} B_k$.
We can then write
\[w\big(S^\beta\big)\left(\frac{\lambda}{\ex}\right)^{\pmax}
= \sum_{k = \pmin}^{\lfloor \beta \cdot \pmax \rfloor} w(B_k) \left(\frac{\lambda}{\ex}\right)^{\pmax}.\]

The weight of each element in the set $B_k$ is the same, $\lambda^{-k}$.
By Lemma~\ref{lem:nu}, the number of elements in set $B_k$ is at most $\nu^k$ for $k$ sufficiently large because $\nu > \cbound$.
We see that
\[\ w\big(S^\beta\big)\left(\frac{\lambda}{\ex}\right)^{\pmax}
\leq \sum_{k = \pmin}^{\lfloor \beta \cdot \pmax \rfloor} \nu^k\lambda^{-k} \left(\frac{\lambda}{\ex}\right)^{\pmax}.\]

Because $k \leq \beta \cdot \pmax$, we have $\pmax \geq k/\beta$. As $\lambda < x$, we have
\[w\big(S^\beta\big)\left(\frac{\lambda}{\ex}\right)^{\pmax}
\leq \sum_{k = \pmin}^{\lfloor \beta \cdot \pmax \rfloor} \nu^k\lambda^{-k} \left(\frac{\lambda}{\ex}\right)^{\frac{k}{\beta}}
= \sum_{k = \pmin}^{\lfloor \beta \cdot \pmax \rfloor} \
\left( \frac{ \nu \lambda^{\frac{1}{\beta}}}{\lambda \ex^{\frac{1}{\beta}}} \right)^k
= \sum_{k = \pmin}^{\lfloor \beta \cdot \pmax \rfloor} \
\left( \frac{ \nu}{\lambda^{\frac{\beta-1}{\beta}} \ex^{\frac{1}{\beta}}} \right)^k.\]

Because we picked $\nu$ so that $\nu < \lambda^{(\beta - 1) / \beta} \ex^{1 / \beta}$, the rightmost term in parentheses is less than one.
By applying the inequalities $\pmax = 2n-2 \geq k \geq \pmin > \sqrt{n}$ (by Lemma~\ref{lem:pmin}), we see that
\[w\big(S^\beta\big)\left(\frac{\lambda}{\ex}\right)^{\pmax}
\leq \sum_{k = \pmin}^{\lfloor \beta \cdot \pmax \rfloor} \
\left(\frac{\nu}{\lambda^{\frac{\beta-1}{\beta}} \ex^{\frac{1}{\beta}} } \right)^{\sqrt{n}}
\leq (2n-2)\left( \frac{ \nu}{\lambda^{\frac{\beta-1}{\beta}} \ex^{\frac{1}{\beta}}} \right)^{\sqrt{n}}.\]

Again using the fact that the rightmost term in parentheses above is less than one, we can find a constant $\zeta < 1$ and an $n^*$ such that for all $n \geq n^*$,
\[\p_{\sigma \sim \pi}\left(p(\sigma) \leq \beta \cdot \pmax\right)
= \pi\big(S^\beta\big)
\leq 8 w\big(S^\beta\big)\left(\frac{\lambda}{\ex}\right)^{\pmax}
\leq 8 (2n-2)\left( \frac{ \nu}{\lambda^{\frac{\beta-1}{\beta}} \ex^{\frac{1}{\beta}}} \right)^{\sqrt{n}}
< \zeta^{\sqrt{n}}.\]
\end{proof}

Combining Theorem~\ref{thm:2.17} with Corollary~\ref{cor:expand_lambda} gives the following result.

\begin{cor}
For all $0 < \lambda < (2N_{50})^{1/100} \sim 2.17 $, there exists a constant $0 < \beta < 1$ such that for sufficiently large $n$ with all but exponentially small probability a sample drawn according to stationary distribution $\pi$ of $\M$ is $\beta$-expanded.
\end{cor}
\begin{proof}
	If $0 < \lambda < 1$, then by Corollary~\ref{cor:expand_lambda}, for any constant $\beta < \frac{ \log \sqrt{2} - \log \lambda }{\log(2+\sqrt{2}) - \log \lambda}$, for sufficiently large $n$ with all but exponentially small probability $\beta$-expansion occurs at stationarity . Note that for $\lambda < 1$,  $1 > \frac{ \log \sqrt{2} - \log \lambda }{\log(2+\sqrt{2}) - \log \lambda} > 0$, so there always exists such a positive constant $\beta$ less than this bound for which $\beta$-expansion occurs.

	If $1 \leq \lambda < \ex := (2N_{50})^{1/100} \sim 2.17  $, then by Theorem~\ref{thm:2.17}, for any $\beta < \frac{\log \ex - \log \lambda}{\log(\cbound) - \log \lambda }$, for sufficiently large $n$ with all but exponentially small probability $\beta$-expansion occurs at stationarity. Because $1 > \frac{\log \ex - \log \lambda}{\log(\cbound) - \log \lambda } > 0$, there always exists such a positive constant $\beta$ less than this bound for which $\beta$-expansion occurs.

	We conclude that for any $0 < \lambda < (2N_{50})^{1/100}$, there exists a constant $0 < \beta < 1$ such that for sufficiently large $n$ with all but exponentially small probability a sample drawn according to stationary distribution $\pi$ of $\M$ is $\beta$-expanded.
\end{proof}

\section{Conclusion}\label{sec:concl}

We have given comprehensive results that show that our Markov chain $\M$ and distributed algorithm $\A$ provably achieve compression whenever $\lambda > \cbound \sim 3.41$ and provably achieve expansion (and thus do not achieve compression) whenever $\lambda < \ebound$, but some open problems remain.
As discussed at length in Section~\ref{subsec:convergence}, we have been unable to provide bounds on the mixing time of $\M$, that is, on the amount of time our algorithm needs to run for before it reaches stationarity.
However, reaching stationarity may not be a necessary condition for compression to occur, and empirical analysis suggests that compression occurs in polynomial time, after roughly $\bigO{n^4}$ iterations of $\M$.
Furthermore, our current proofs do not provide any information about the behavior of our algorithm for intermediate values of $\lambda$. We conjecture there is a {\it phase transition} in $\lambda$, i.e., a critical value $\lambda_c$ such that for all $\lambda > \lambda_c$ the particle system achieves compression and for all $\lambda < \lambda_c$ it achieves expansion. Phase transitions exist for similar statistical physics models such as the Ising model (see, e.g., Section 3.7 of~\cite{Friedli2018}) and Potts model (see, e.g.,~\cite{Hintermann1978}).
The proofs of Theorems~\ref{thm:compress_alpha} and~\ref{thm:2.17} indicate that if $\lambda_c$ exists, then $\ebound \leq \lambda_c \leq \cbound$.
\bluecomment{JD: Would we be interested in adding a nod to our undergrad Kevin's work here in moving the bound from 2.17 to 2.21, even just as an indication that finding this critical point may be difficult?}
\redcomment{SC: Sure, do you have a citation, even if it's just `in preparation'? Something like "Recent work has showed that expansion occurs up to 2.21 \cite{kevin}" makes sense. }

Beyond compression and expansion, the stochastic approach that we originate here has already been useful for addressing, via local algorithms, a plethora of other fundamental problems within programmable matter, including versions of optimization, clustering, and directed motion problems.
For example,~\cite{Arroyo2018} uses the stochastic approach to give a local, distributed algorithm for solving a global optimization problem known as the {\it shortcut bridging problem}.
Certain ant colonies have been observed to solve this problem during their foraging process~\cite{Reid2015}, but the mechanisms by which they do so are not well understood.
The algorithm of~\cite{Arroyo2018} gives one way simple entities with limited computational power can collectively solve a global optimization problem via local interactions.
In more recent applied work in swarm robotics, a Markov chain algorithm for a self-organizing particle system was used provide a theoretical explanation of a behavior known as {\it phototaxing} observed in a particular robot swarm. In both the robot swarm and our abstract particle system, individuals had no sense of directionality, but when they were able to vary how quickly they moved in response to light, this produced directed motion of the swarm as a whole~\cite{Savoie2018}. Other recent work applies the stochastic approach to the {\it separation} problem, where the goal is for particles of different colors to either intermingle or segregate (cluster)~\cite{Cannon2018}.


We believe we have so far only scratched the surface of the stochastic approach to programmable matter.
More broadly, the stochastic approach could be applied to accomplish any objective that can be described by a global energy function (where the desirable configurations have low energy values), provided changes in energy due to particle movements can be calculated with only local information.
For compression, we used the energy function $H(\sigma) = -e(\sigma)$, and because changes to the number of edges due to a particle move can be calculated just by looking at that particle's neighborhood, we could turn our Markov chain that favors more edges into a distributed algorithm that favors more edges.
For any such energy function, we can give a Markov chain whose stationary distribution favors low energy configurations, but this is not always enough to provably accomplish our objectives.
For example, when $\lambda = 2$, our Markov chain favors configurations with more edges, but compression does not occur.
If there are many more configurations with high energy than with low energy, a situation known as {\it high entropy},
then the probability a configuration drawn from the stationary distribution accomplishes the desired objective may not be very high.
For this reason, biases must be large enough to guarantee low energy configurations --- those that accomplish the objectives --- dominate the state space, even if there are many undesirable high energy configurations; for compression, we prove $\lambda > \cbound$ is sufficient.
We used Peierls arguments to analyze this energy/entropy trade-off for compression and expansion, and more sophisticated Peierls-type arguments were used in the proofs of shortcut bridging~\cite{Arroyo2018} and separation~\cite{Cannon2018}.
We expect similar proof approaches to work for these types of analyses in the future, though technical details could vary widely and will depend on the specifics of each problem.

\begin{acks}
S.~Cannon was supported in part by a grant from the \grantsponsor{simons}{Simons Foundation}{https://www.simonsfoundation.org/}
( \#\grantnum{simons}{361047} to Sarah Cannon) and the \grantsponsor{NSF}{National Science Foundation}{https://www.nsf.gov} under awards \grantnum{NSF}{DMS-1803325} and
\grantnum{NSF}{DGE-1148903}. Any opinions, findings, and conclusions or recommendations
expressed in this material are those of the author(s) and do not necessarily reflect the views of the National Science Foundation.

J.~J.~Daymude and A.~W.~Richa gratefully acknowledge their support from the \grantsponsor{NSF}{National Science Foundation}{https://www.nsf.gov} under awards \grantnum{NSF}{CCF-1422603}, \grantnum{NSF}{CCF-1637393}, and \grantnum{NSF}{CCF-1733680}.

D.~Randall was supported in part by the \grantsponsor{NSF}{National Science Foundation}{https://www.nsf.gov} under awards \grantnum{NSF}{CCF-1526900}, \grantnum{NSF}{CCF-1637031}, and \grantnum{NSF}{CCF-1733812}.
\end{acks}

\bibliographystyle{ACM-Reference-Format}
\bibliography{biblio}


\begin{thebibliography}{54}


\ifx \showCODEN    \undefined \def \showCODEN     #1{\unskip}     \fi
\ifx \showDOI      \undefined \def \showDOI       #1{#1}\fi
\ifx \showISBNx    \undefined \def \showISBNx     #1{\unskip}     \fi
\ifx \showISBNxiii \undefined \def \showISBNxiii  #1{\unskip}     \fi
\ifx \showISSN     \undefined \def \showISSN      #1{\unskip}     \fi
\ifx \showLCCN     \undefined \def \showLCCN      #1{\unskip}     \fi
\ifx \shownote     \undefined \def \shownote      #1{#1}          \fi
\ifx \showarticletitle \undefined \def \showarticletitle #1{#1}   \fi
\ifx \showURL      \undefined \def \showURL       {\relax}        \fi
\providecommand\bibfield[2]{#2}
\providecommand\bibinfo[2]{#2}
\providecommand\natexlab[1]{#1}
\providecommand\showeprint[2][]{arXiv:#2}

\bibitem[\protect\citeauthoryear{Adleman}{Adleman}{1994}]%
        {Adleman1994}
\bibfield{author}{\bibinfo{person}{Leonard~M. Adleman}.}
  \bibinfo{year}{1994}\natexlab{}.
\newblock \showarticletitle{Molecular computation of solutions to combinatorial
  problems}.
\newblock \bibinfo{journal}{\emph{Science}} \bibinfo{volume}{266},
  \bibinfo{number}{5187} (\bibinfo{year}{1994}), \bibinfo{pages}{1021--1024}.
\newblock


\bibitem[\protect\citeauthoryear{Andr{\'e}s~Arroyo, Cannon, Daymude, Randall,
  and Richa}{Andr{\'e}s~Arroyo et~al\mbox{.}}{2018}]%
        {Arroyo2018}
\bibfield{author}{\bibinfo{person}{Marta Andr{\'e}s~Arroyo},
  \bibinfo{person}{Sarah Cannon}, \bibinfo{person}{Joshua~J. Daymude},
  \bibinfo{person}{Dana Randall}, {and} \bibinfo{person}{Andr{\'e}a~W. Richa}.}
  \bibinfo{year}{2018}\natexlab{}.
\newblock \showarticletitle{A Stochastic Approach to Shortcut Bridging in
  Programmable Matter}.
\newblock \bibinfo{journal}{\emph{Natural Computing}} \bibinfo{volume}{17},
  \bibinfo{number}{4} (\bibinfo{year}{2018}), \bibinfo{pages}{723--741}.
\newblock


\bibitem[\protect\citeauthoryear{Bampas, Czyzowicz, G{\k{a}}sieniec, Ilcinkas,
  and Labourel}{Bampas et~al\mbox{.}}{2010}]%
        {Bampas2010}
\bibfield{author}{\bibinfo{person}{Evangelos Bampas}, \bibinfo{person}{Jurek
  Czyzowicz}, \bibinfo{person}{Leszek G{\k{a}}sieniec}, \bibinfo{person}{David
  Ilcinkas}, {and} \bibinfo{person}{Arnaud Labourel}.}
  \bibinfo{year}{2010}\natexlab{}.
\newblock \showarticletitle{Almost Optimal Asynchronous Rendezvous in Infinite
  Multidimensional Grids}. In \bibinfo{booktitle}{\emph{Distributed Computing}}
  \emph{(\bibinfo{series}{DISC '10})}. \bibinfo{publisher}{Springer Berlin
  Heidelberg}, \bibinfo{address}{Berlin, Heidelberg},
  \bibinfo{pages}{297--311}.
\newblock


\bibitem[\protect\citeauthoryear{Bauerschmidt, Duminil-Copin, Goodman, and
  Slade}{Bauerschmidt et~al\mbox{.}}{2012}]%
        {Bauerschmidt2012}
\bibfield{author}{\bibinfo{person}{Roland Bauerschmidt}, \bibinfo{person}{Hugo
  Duminil-Copin}, \bibinfo{person}{Jesse Goodman}, {and}
  \bibinfo{person}{Gordon Slade}.} \bibinfo{year}{2012}\natexlab{}.
\newblock \showarticletitle{Lectures on self-avoiding walks}.
\newblock \bibinfo{journal}{\emph{Probability and Statistical Physics in Two
  and More Dimensions}}  \bibinfo{volume}{15} (\bibinfo{year}{2012}),
  \bibinfo{pages}{395--476}.
\newblock


\bibitem[\protect\citeauthoryear{Baxter, Enting, and Tsang}{Baxter
  et~al\mbox{.}}{1980}]%
        {Baxter1980}
\bibfield{author}{\bibinfo{person}{Rodney~J. Baxter}, \bibinfo{person}{I.~G.
  Enting}, {and} \bibinfo{person}{S.~K. Tsang}.}
  \bibinfo{year}{1980}\natexlab{}.
\newblock \showarticletitle{Hard-Square Lattice Gas}.
\newblock \bibinfo{journal}{\emph{Journal of Statistical Physics}}
  \bibinfo{volume}{22} (\bibinfo{year}{1980}), \bibinfo{pages}{465--489}.
\newblock


\bibitem[\protect\citeauthoryear{Blanca, Chen, Galvin, Randall, and
  Tetali}{Blanca et~al\mbox{.}}{2018}]%
        {Blanca2018}
\bibfield{author}{\bibinfo{person}{Antonio Blanca}, \bibinfo{person}{Yuxuan
  Chen}, \bibinfo{person}{David Galvin}, \bibinfo{person}{Dana Randall}, {and}
  \bibinfo{person}{Prasad Tetali}.} \bibinfo{year}{2018}\natexlab{}.
\newblock \showarticletitle{Phase coexistence for the hard-core model on
  {$\mathbb{Z}^2$}}.
\newblock \bibinfo{journal}{\emph{Combinatorics, Probability and Computing}}
  (\bibinfo{year}{2018}), \bibinfo{pages}{1--22}.
\newblock


\bibitem[\protect\citeauthoryear{Borgs, Chayes, Kim, Frieze, Tetali, Vigoda,
  and Vu}{Borgs et~al\mbox{.}}{1999}]%
        {Borgs1999}
\bibfield{author}{\bibinfo{person}{Christian Borgs},
  \bibinfo{person}{Jennifer~T. Chayes}, \bibinfo{person}{Jeong~Han Kim},
  \bibinfo{person}{Alan Frieze}, \bibinfo{person}{Prasad Tetali},
  \bibinfo{person}{Eric Vigoda}, {and} \bibinfo{person}{Van~Ha Vu}.}
  \bibinfo{year}{1999}\natexlab{}.
\newblock \showarticletitle{Torpid mixing of some {M}onte {C}arlo {M}arkov
  chain algorithms in statistical physics}. In
  \bibinfo{booktitle}{\emph{Proceedings of the 40th Annual Symposium on
  Foundations of Computer Science}} \emph{(\bibinfo{series}{FOCS '99})}.
  \bibinfo{publisher}{IEEE Computer Society}, \bibinfo{address}{Washington, DC,
  USA}, \bibinfo{pages}{218--229}.
\newblock


\bibitem[\protect\citeauthoryear{Camazine, Visscher, Finley, and
  Vetter}{Camazine et~al\mbox{.}}{1999}]%
        {Camazine1999}
\bibfield{author}{\bibinfo{person}{Scott Camazine}, \bibinfo{person}{Peter~K.
  Visscher}, \bibinfo{person}{Jennifer Finley}, {and}
  \bibinfo{person}{Richard~S. Vetter}.} \bibinfo{year}{1999}\natexlab{}.
\newblock \showarticletitle{House-Hunting by Honey Bee Swarms: Collective
  Decisions and Individual Behaviors}.
\newblock \bibinfo{journal}{\emph{Insectes Sociaux}} \bibinfo{volume}{46},
  \bibinfo{number}{4} (\bibinfo{year}{1999}), \bibinfo{pages}{348--360}.
\newblock


\bibitem[\protect\citeauthoryear{Cannon, Daymude, Gokmen, Randall, and
  Richa}{Cannon et~al\mbox{.}}{2018}]%
        {Cannon2018}
\bibfield{author}{\bibinfo{person}{Sarah Cannon}, \bibinfo{person}{Joshua~J.
  Daymude}, \bibinfo{person}{Cem Gokmen}, \bibinfo{person}{Dana Randall}, {and}
  \bibinfo{person}{Andr{\'{e}}a~W. Richa}.} \bibinfo{year}{2018}\natexlab{}.
\newblock \bibinfo{title}{A Local Stochastic Algorithm for Separation in
  Heterogeneous Self-Organizing Particle Systems}.  (\bibinfo{year}{2018}).
\newblock
\newblock
\shownote{In preparation, preprint available online at
  \url{https://arxiv.org/abs/1805.04599}.}


\bibitem[\protect\citeauthoryear{Cannon, Daymude, Randall, and Richa}{Cannon
  et~al\mbox{.}}{2016}]%
        {Cannon2016}
\bibfield{author}{\bibinfo{person}{Sarah Cannon}, \bibinfo{person}{Joshua~J.
  Daymude}, \bibinfo{person}{Dana Randall}, {and} \bibinfo{person}{Andr\'ea~W.
  Richa}.} \bibinfo{year}{2016}\natexlab{}.
\newblock \showarticletitle{A {M}arkov chain algorithm for compression in
  self-organizing particle systems}. In \bibinfo{booktitle}{\emph{Proceedings
  of the 2016 ACM Symposium on Principles of Distributed Computing}}
  \emph{(\bibinfo{series}{PODC '16})}. \bibinfo{publisher}{ACM},
  \bibinfo{address}{New York, NY, USA}, \bibinfo{pages}{279--288}.
\newblock


\bibitem[\protect\citeauthoryear{Caputo, Martinelli, Simenhaus, and
  Toninelli}{Caputo et~al\mbox{.}}{2011}]%
        {Caputo2011}
\bibfield{author}{\bibinfo{person}{Pietro Caputo}, \bibinfo{person}{Fabio
  Martinelli}, \bibinfo{person}{Francois Simenhaus}, {and}
  \bibinfo{person}{Fabio~Lucio Toninelli}.} \bibinfo{year}{2011}\natexlab{}.
\newblock \showarticletitle{``Zero" temperature stochastic {3D} {I}sing model
  and dimer covering fluctuations: a first step towards interface mean
  curvature motion}.
\newblock \bibinfo{journal}{\emph{Communications of Pure and Applied
  Mathematics}}  \bibinfo{volume}{64} (\bibinfo{year}{2011}),
  \bibinfo{pages}{778--831}.
\newblock


\bibitem[\protect\citeauthoryear{Chavoya and Duthen}{Chavoya and
  Duthen}{2006}]%
        {Chavoya2006}
\bibfield{author}{\bibinfo{person}{Arturo Chavoya} {and} \bibinfo{person}{Yves
  Duthen}.} \bibinfo{year}{2006}\natexlab{}.
\newblock \showarticletitle{Using a Genetic Algorithm to Evolve Cellular
  Automata for 2{D}/3{D} Computational Development}. In
  \bibinfo{booktitle}{\emph{Proceedings of the 8th Annual Conference on Genetic
  and Evolutionary Computation}} \emph{(\bibinfo{series}{GECCO '06})}.
  \bibinfo{publisher}{ACM}, \bibinfo{address}{New York, NY, USA},
  \bibinfo{pages}{231--232}.
\newblock


\bibitem[\protect\citeauthoryear{Chen, Doty, Holden, Thachuk, Woods, and
  Yang}{Chen et~al\mbox{.}}{2014}]%
        {Chen2014}
\bibfield{author}{\bibinfo{person}{Ho-Lin Chen}, \bibinfo{person}{David Doty},
  \bibinfo{person}{Dhiraj Holden}, \bibinfo{person}{Chris Thachuk},
  \bibinfo{person}{Damien Woods}, {and} \bibinfo{person}{Chun-Tao Yang}.}
  \bibinfo{year}{2014}\natexlab{}.
\newblock \showarticletitle{Fast Algorithmic Self-assembly of Simple Shapes
  Using Random Agitation}. In \bibinfo{booktitle}{\emph{DNA Computing and
  Molecular Programming}} \emph{(\bibinfo{series}{DNA20})}.
  \bibinfo{publisher}{Springer International Publishing},
  \bibinfo{address}{Cham}, \bibinfo{pages}{20--36}.
\newblock


\bibitem[\protect\citeauthoryear{Chirikjian}{Chirikjian}{1994}]%
        {Chirikjian1994}
\bibfield{author}{\bibinfo{person}{Gregory~S. Chirikjian}.}
  \bibinfo{year}{1994}\natexlab{}.
\newblock \showarticletitle{Kinematics of a metamorphic robotic system}. In
  \bibinfo{booktitle}{\emph{Proceedings of the 1994 IEEE International
  Conference on Robotics and Automation}} \emph{(\bibinfo{series}{ICRA '94})},
  Vol.~\bibinfo{volume}{1}. \bibinfo{publisher}{IEEE Computer Society},
  \bibinfo{address}{Washington, DC, USA}, \bibinfo{pages}{449--455}.
\newblock


\bibitem[\protect\citeauthoryear{David, Athina, Christophe, Nynika, Steffen,
  Achim, and Skylar}{David et~al\mbox{.}}{2015}]%
        {David2015}
\bibfield{author}{\bibinfo{person}{Correa David}, \bibinfo{person}{Papadopoulou
  Athina}, \bibinfo{person}{Guberan Christophe}, \bibinfo{person}{Jhaveri
  Nynika}, \bibinfo{person}{Reichert Steffen}, \bibinfo{person}{Menges Achim},
  {and} \bibinfo{person}{Tibbits Skylar}.} \bibinfo{year}{2015}\natexlab{}.
\newblock \showarticletitle{3D-Printed Wood: Programming Hygroscopic Material
  Transformations}.
\newblock \bibinfo{journal}{\emph{3D Printing and Additive Manufacturing}}
  \bibinfo{volume}{2}, \bibinfo{number}{3} (\bibinfo{year}{2015}),
  \bibinfo{pages}{106--116}.
\newblock


\bibitem[\protect\citeauthoryear{Daymude, Gmyr, Richa, Scheideler, and
  Strothmann}{Daymude et~al\mbox{.}}{2017}]%
        {Daymude2017}
\bibfield{author}{\bibinfo{person}{Joshua~J. Daymude}, \bibinfo{person}{Robert
  Gmyr}, \bibinfo{person}{Andr\'ea~W. Richa}, \bibinfo{person}{Christian
  Scheideler}, {and} \bibinfo{person}{Thim Strothmann}.}
  \bibinfo{year}{2017}\natexlab{}.
\newblock \showarticletitle{Improved Leader Election for Self-Organizing
  Programmable Matter}. In \bibinfo{booktitle}{\emph{Algorithms for Sensor
  Systems}} \emph{(\bibinfo{series}{ALGOSENSORS '17})}.
  \bibinfo{publisher}{Springer International Publishing},
  \bibinfo{address}{Cham}, \bibinfo{pages}{127--140}.
\newblock


\bibitem[\protect\citeauthoryear{Daymude, Hinnenthal, Richa, and
  Scheideler}{Daymude et~al\mbox{.}}{2019}]%
        {Daymude2019}
\bibfield{author}{\bibinfo{person}{Joshua~J. Daymude},
  \bibinfo{person}{Kristian Hinnenthal}, \bibinfo{person}{Andr\'ea~W. Richa},
  {and} \bibinfo{person}{Christian Scheideler}.}
  \bibinfo{year}{2019}\natexlab{}.
\newblock \showarticletitle{Computing by Programmable Particles}.
\newblock In \bibinfo{booktitle}{\emph{Distributed Computing by Mobile
  Entities: Current Research in Moving and Computing}}.
  \bibinfo{publisher}{Springer}, \bibinfo{address}{Cham},
  \bibinfo{pages}{615--681}.
\newblock


\bibitem[\protect\citeauthoryear{Derakhshandeh, Dolev, Gmyr, Richa, Scheideler,
  and Strothmann}{Derakhshandeh et~al\mbox{.}}{2014}]%
        {Derakhshandeh2014}
\bibfield{author}{\bibinfo{person}{Zahra Derakhshandeh},
  \bibinfo{person}{Shlomi Dolev}, \bibinfo{person}{Robert Gmyr},
  \bibinfo{person}{Andr\'ea~W. Richa}, \bibinfo{person}{Christian Scheideler},
  {and} \bibinfo{person}{Thim Strothmann}.} \bibinfo{year}{2014}\natexlab{}.
\newblock \showarticletitle{Brief announcement: amoebot - a new model for
  programmable matter}. In \bibinfo{booktitle}{\emph{Proceedings of the 26th
  ACM Symposium on Parallelism in Algorithms and Architectures}}
  \emph{(\bibinfo{series}{SPAA '14})}. \bibinfo{publisher}{ACM},
  \bibinfo{address}{New York, NY, USA}, \bibinfo{pages}{220--222}.
\newblock


\bibitem[\protect\citeauthoryear{Derakhshandeh, Gmyr, Richa, Scheideler, and
  Strothmann}{Derakhshandeh et~al\mbox{.}}{2015}]%
        {Derakhshandeh2015-shape}
\bibfield{author}{\bibinfo{person}{Zahra Derakhshandeh},
  \bibinfo{person}{Robert Gmyr}, \bibinfo{person}{Andr{\'{e}}a~W. Richa},
  \bibinfo{person}{Christian Scheideler}, {and} \bibinfo{person}{Thim
  Strothmann}.} \bibinfo{year}{2015}\natexlab{}.
\newblock \showarticletitle{An Algorithmic Framework for Shape Formation
  Problems in Self-Organizing Particle Systems}. In
  \bibinfo{booktitle}{\emph{Proceedings of the Second Annual International
  Conference on Nanoscale Computing and Communication}}
  \emph{(\bibinfo{series}{NANOCOM '15})}. \bibinfo{publisher}{ACM},
  \bibinfo{address}{New York, NY, USA}, \bibinfo{pages}{21:1--21:2}.
\newblock


\bibitem[\protect\citeauthoryear{Derakhshandeh, Gmyr, Richa, Scheideler, and
  Strothmann}{Derakhshandeh et~al\mbox{.}}{2016}]%
        {Derakhshandeh2016}
\bibfield{author}{\bibinfo{person}{Zahra Derakhshandeh},
  \bibinfo{person}{Robert Gmyr}, \bibinfo{person}{Andr\'ea~W. Richa},
  \bibinfo{person}{Christian Scheideler}, {and} \bibinfo{person}{Thim
  Strothmann}.} \bibinfo{year}{2016}\natexlab{}.
\newblock \showarticletitle{Universal Shape Formation for Programmable Matter}.
  In \bibinfo{booktitle}{\emph{Proceedings of the 28th ACM Symposium on
  Parallelism in Algorithms and Architectures}} \emph{(\bibinfo{series}{SPAA
  '16})}. \bibinfo{publisher}{ACM}, \bibinfo{address}{New York, NY, USA},
  \bibinfo{pages}{289--299}.
\newblock


\bibitem[\protect\citeauthoryear{Deutsch and Dormann}{Deutsch and
  Dormann}{2017}]%
        {Deutsch2017}
\bibfield{author}{\bibinfo{person}{Andreas Deutsch} {and}
  \bibinfo{person}{Sabine Dormann}.} \bibinfo{year}{2017}\natexlab{}.
\newblock \bibinfo{booktitle}{\emph{Cellular Automaton Modeling of Biological
  Pattern Formation} (\bibinfo{edition}{2nd} ed.)}.
\newblock \bibinfo{publisher}{Birkh{\"{a}}user Basel}, \bibinfo{address}{Basel,
  Switzerland}.
\newblock


\bibitem[\protect\citeauthoryear{Devreotes}{Devreotes}{1989}]%
        {Devreotes1989}
\bibfield{author}{\bibinfo{person}{Peter Devreotes}.}
  \bibinfo{year}{1989}\natexlab{}.
\newblock \showarticletitle{Dictyostelium discoideum: A Model System for
  Cell-Cell Interactions in Development}.
\newblock \bibinfo{journal}{\emph{Science}} \bibinfo{volume}{245},
  \bibinfo{number}{4922} (\bibinfo{year}{1989}), \bibinfo{pages}{1054--1058}.
\newblock


\bibitem[\protect\citeauthoryear{Di~Luna, Flocchini, Prencipe, Santoro, and
  Viglietta}{Di~Luna et~al\mbox{.}}{2018}]%
        {DiLuna2018}
\bibfield{author}{\bibinfo{person}{Giuseppe~Antonio Di~Luna},
  \bibinfo{person}{Paola Flocchini}, \bibinfo{person}{Giuseppe Prencipe},
  \bibinfo{person}{Nicola Santoro}, {and} \bibinfo{person}{Giovanni
  Viglietta}.} \bibinfo{year}{2018}\natexlab{}.
\newblock \showarticletitle{Line Recovery by Programmable Particles}. In
  \bibinfo{booktitle}{\emph{Proceedings of the 19th International Conference on
  Distributed Computing and Networking}} \emph{(\bibinfo{series}{ICDCN '18})}.
  \bibinfo{publisher}{ACM}, \bibinfo{address}{New York, NY, USA},
  \bibinfo{pages}{4:1--4:10}.
\newblock


\bibitem[\protect\citeauthoryear{Diaconis and Saloff-Coste}{Diaconis and
  Saloff-Coste}{1993}]%
        {Diaconis1993}
\bibfield{author}{\bibinfo{person}{Persi Diaconis} {and}
  \bibinfo{person}{Laurent Saloff-Coste}.} \bibinfo{year}{1993}\natexlab{}.
\newblock \showarticletitle{Comparison Theorems for Reversible {M}arkov
  Chains}.
\newblock \bibinfo{journal}{\emph{The Annals of Applied Probability}}
  \bibinfo{volume}{3}, \bibinfo{number}{3} (\bibinfo{year}{1993}),
  \bibinfo{pages}{696--730}.
\newblock


\bibitem[\protect\citeauthoryear{Dobrushin}{Dobrushin}{1968}]%
        {Dobrushin1968}
\bibfield{author}{\bibinfo{person}{Roland~L. Dobrushin}.}
  \bibinfo{year}{1968}\natexlab{}.
\newblock \showarticletitle{The problem of Uniqueness of a Gibbsian Random
  Field and the Problem of Phase Transitions}.
\newblock \bibinfo{journal}{\emph{Functional Analysis and Its Applications}}
  \bibinfo{volume}{2} (\bibinfo{year}{1968}), \bibinfo{pages}{302--312}.
\newblock


\bibitem[\protect\citeauthoryear{Duminil-Copin and Smirnov}{Duminil-Copin and
  Smirnov}{2012}]%
        {DuminilCopin2012}
\bibfield{author}{\bibinfo{person}{Hugo Duminil-Copin} {and}
  \bibinfo{person}{Stanislav Smirnov}.} \bibinfo{year}{2012}\natexlab{}.
\newblock \showarticletitle{The connective constant of the honeycomb lattice
  equals $\sqrt{2+\sqrt{2}}$}.
\newblock \bibinfo{journal}{\emph{Annals of Mathematics}}
  \bibinfo{volume}{175}, \bibinfo{number}{3} (\bibinfo{year}{2012}),
  \bibinfo{pages}{1653--1665}.
\newblock


\bibitem[\protect\citeauthoryear{Feller}{Feller}{1968}]%
        {Feller1968}
\bibfield{author}{\bibinfo{person}{William Feller}.}
  \bibinfo{year}{1968}\natexlab{}.
\newblock \bibinfo{booktitle}{\emph{An Introduction to Probability Theory and
  Its Applications}}. Vol.~\bibinfo{volume}{1}.
\newblock \bibinfo{publisher}{Wiley}, \bibinfo{address}{New York}.
\newblock


\bibitem[\protect\citeauthoryear{Flocchini, Prencipe, and Santoro}{Flocchini
  et~al\mbox{.}}{2019}]%
        {Flocchini2019}
\bibfield{editor}{\bibinfo{person}{Paola Flocchini}, \bibinfo{person}{Giuseppe
  Prencipe}, {and} \bibinfo{person}{Nicola Santoro}} (Eds.).
  \bibinfo{year}{2019}\natexlab{}.
\newblock \bibinfo{booktitle}{\emph{Distributed Computing by Mobile Entities}}.
\newblock \bibinfo{publisher}{Springer International Publishing},
  \bibinfo{address}{Switzerland}.
\newblock


\bibitem[\protect\citeauthoryear{Flocchini, Prencipe, Santoro, and
  Widmayer}{Flocchini et~al\mbox{.}}{2008}]%
        {Flocchini2008}
\bibfield{author}{\bibinfo{person}{Paola Flocchini}, \bibinfo{person}{Giuseppe
  Prencipe}, \bibinfo{person}{Nicola Santoro}, {and} \bibinfo{person}{Peter
  Widmayer}.} \bibinfo{year}{2008}\natexlab{}.
\newblock \showarticletitle{Arbitrary pattern formation by asynchronous,
  anonymous, oblivious robots}.
\newblock \bibinfo{journal}{\emph{Theoretical Computer Science}}
  \bibinfo{volume}{407}, \bibinfo{number}{1} (\bibinfo{year}{2008}),
  \bibinfo{pages}{412--447}.
\newblock


\bibitem[\protect\citeauthoryear{Fortune}{Fortune}{1987}]%
        {Fortune1987}
\bibfield{author}{\bibinfo{person}{Steven Fortune}.}
  \bibinfo{year}{1987}\natexlab{}.
\newblock \showarticletitle{A sweepline algorithm for {V}oronoi diagrams}.
\newblock \bibinfo{journal}{\emph{Algorithmica}} \bibinfo{volume}{2},
  \bibinfo{number}{1} (\bibinfo{year}{1987}), \bibinfo{pages}{153}.
\newblock


\bibitem[\protect\citeauthoryear{Friedli and Velenik}{Friedli and
  Velenik}{2017}]%
        {Friedli2018}
\bibfield{author}{\bibinfo{person}{Sacha Friedli} {and} \bibinfo{person}{Yvan
  Velenik}.} \bibinfo{year}{2017}\natexlab{}.
\newblock \bibinfo{booktitle}{\emph{Statistical Mechanics of Lattice Systems: A
  Concrete Mathematical Introduction}}.
\newblock \bibinfo{publisher}{Cambridge University Press},
  \bibinfo{address}{Cambridge}.
\newblock


\bibitem[\protect\citeauthoryear{Gauci, Chen, Li, Dodd, and Gro{\ss}}{Gauci
  et~al\mbox{.}}{2014}]%
        {Gauci2014}
\bibfield{author}{\bibinfo{person}{Melvin Gauci}, \bibinfo{person}{Jianing
  Chen}, \bibinfo{person}{Wei Li}, \bibinfo{person}{Tony~J. Dodd}, {and}
  \bibinfo{person}{Roderich Gro{\ss}}.} \bibinfo{year}{2014}\natexlab{}.
\newblock \showarticletitle{Self-organized aggregation without computation}.
\newblock \bibinfo{journal}{\emph{The International Journal of Robotics
  Research}} \bibinfo{volume}{33}, \bibinfo{number}{8} (\bibinfo{year}{2014}),
  \bibinfo{pages}{1145--1161}.
\newblock


\bibitem[\protect\citeauthoryear{Hastings}{Hastings}{1970}]%
        {Hastings1970}
\bibfield{author}{\bibinfo{person}{Wilfred~K. Hastings}.}
  \bibinfo{year}{1970}\natexlab{}.
\newblock \showarticletitle{Monte Carlo Sampling Methods Using {M}arkov Chains
  and Their Applications}.
\newblock \bibinfo{journal}{\emph{Biometrika}} \bibinfo{volume}{57},
  \bibinfo{number}{1} (\bibinfo{year}{1970}), \bibinfo{pages}{97--109}.
\newblock


\bibitem[\protect\citeauthoryear{Hintermann, Kunz, and Wu}{Hintermann
  et~al\mbox{.}}{1978}]%
        {Hintermann1978}
\bibfield{author}{\bibinfo{person}{A. Hintermann}, \bibinfo{person}{H. Kunz},
  {and} \bibinfo{person}{F.~Y. Wu}.} \bibinfo{year}{1978}\natexlab{}.
\newblock \showarticletitle{Exact results for the {P}otts model in two
  dimensions}.
\newblock \bibinfo{journal}{\emph{Journal of Statistical Physics}}
  \bibinfo{volume}{19}, \bibinfo{number}{6} (\bibinfo{year}{1978}),
  \bibinfo{pages}{623--632}.
\newblock


\bibitem[\protect\citeauthoryear{Ising}{Ising}{1925}]%
        {Ising1925}
\bibfield{author}{\bibinfo{person}{Ernst Ising}.}
  \bibinfo{year}{1925}\natexlab{}.
\newblock \showarticletitle{Beitrag zur theorie des ferromagnetismus
  [Contribution to the Theory of Ferromagnetism]}.
\newblock \bibinfo{journal}{\emph{Zeitschrift f{\"u}r Physik}}
  \bibinfo{volume}{31}, \bibinfo{number}{1} (\bibinfo{year}{1925}),
  \bibinfo{pages}{253--258}.
\newblock


\bibitem[\protect\citeauthoryear{Jeanson, Rivault, Deneubourg, Blanco,
  Fournier, Jost, and Theraulaz}{Jeanson et~al\mbox{.}}{2005}]%
        {Jeanson2005}
\bibfield{author}{\bibinfo{person}{Raphael Jeanson}, \bibinfo{person}{Colette
  Rivault}, \bibinfo{person}{Jean-Louis Deneubourg}, \bibinfo{person}{Stephane
  Blanco}, \bibinfo{person}{Richard Fournier}, \bibinfo{person}{Christian
  Jost}, {and} \bibinfo{person}{Guy Theraulaz}.}
  \bibinfo{year}{2005}\natexlab{}.
\newblock \showarticletitle{Self-organized aggregation in cockroaches}.
\newblock \bibinfo{journal}{\emph{Animal Behaviour}} \bibinfo{volume}{69},
  \bibinfo{number}{1} (\bibinfo{year}{2005}), \bibinfo{pages}{169--180}.
\newblock


\bibitem[\protect\citeauthoryear{Jensen}{Jensen}{2004}]%
        {Jensen2004}
\bibfield{author}{\bibinfo{person}{Iwan Jensen}.}
  \bibinfo{year}{2004}\natexlab{}.
\newblock \showarticletitle{Improved lower bounds on the connective constants
  for two-dimensional self-avoiding walks}.
\newblock \bibinfo{journal}{\emph{Journal of Physics A: Mathematical and
  General}} \bibinfo{volume}{37}, \bibinfo{number}{48} (\bibinfo{year}{2004}),
  \bibinfo{pages}{11521--11529}.
\newblock


\bibitem[\protect\citeauthoryear{Jensen}{Jensen}{2009}]%
        {Jensen2009}
\bibfield{author}{\bibinfo{person}{Iwan Jensen}.}
  \bibinfo{year}{2009}\natexlab{}.
\newblock \showarticletitle{A parallel algorithm for the enumeration of
  benzenoid hydrocarbons}.
\newblock \bibinfo{journal}{\emph{Journal of Statistical Mechanics: Theory and
  Experiment}} \bibinfo{volume}{2009}, \bibinfo{number}{2}
  (\bibinfo{year}{2009}), \bibinfo{pages}{P02065}.
\newblock


\bibitem[\protect\citeauthoryear{Levin, Peres, and Wilmer}{Levin
  et~al\mbox{.}}{2009}]%
        {Levin2009}
\bibfield{author}{\bibinfo{person}{David~A. Levin}, \bibinfo{person}{Yuval
  Peres}, {and} \bibinfo{person}{Elizabeth~L. Wilmer}.}
  \bibinfo{year}{2009}\natexlab{}.
\newblock \bibinfo{booktitle}{\emph{{M}arkov chains and mixing times}}.
\newblock \bibinfo{publisher}{American Mathematical Society},
  \bibinfo{address}{Providence, RI, USA}.
\newblock


\bibitem[\protect\citeauthoryear{Lubetzky, Martinelli, Sly, and
  Toninelli}{Lubetzky et~al\mbox{.}}{2013}]%
        {Lubetzky2013}
\bibfield{author}{\bibinfo{person}{Eyal Lubetzky}, \bibinfo{person}{Fabio
  Martinelli}, \bibinfo{person}{Alan Sly}, {and} \bibinfo{person}{Fabio~Lucio
  Toninelli}.} \bibinfo{year}{2013}\natexlab{}.
\newblock \showarticletitle{Quasi-polynomial mixing of the {2D} stochastic
  {I}sing model with ``plus" boundary up to criticality}.
\newblock \bibinfo{journal}{\emph{Journal of the European Mathematical Society
  (JEMS)}} \bibinfo{volume}{15}, \bibinfo{number}{2} (\bibinfo{year}{2013}),
  \bibinfo{pages}{339–--386}.
\newblock


\bibitem[\protect\citeauthoryear{Lynch}{Lynch}{1996}]%
        {Lynch1996}
\bibfield{author}{\bibinfo{person}{Nancy Lynch}.}
  \bibinfo{year}{1996}\natexlab{}.
\newblock \bibinfo{booktitle}{\emph{Distributed Algorithms}}.
\newblock \bibinfo{publisher}{Morgan Kauffman}, \bibinfo{address}{San
  Francisco, CA, USA}.
\newblock


\bibitem[\protect\citeauthoryear{Martinelli and Toninelli}{Martinelli and
  Toninelli}{2010}]%
        {Martinelli2010}
\bibfield{author}{\bibinfo{person}{Fabio Martinelli} {and}
  \bibinfo{person}{Fabio~Lucio Toninelli}.} \bibinfo{year}{2010}\natexlab{}.
\newblock \showarticletitle{On the Mixing Time of the {2D} Stochastic {I}sing
  Model with ``Plus” Boundary Conditions at Low Temperature}.
\newblock \bibinfo{journal}{\emph{Communications in Mathematical Physics}}
  \bibinfo{volume}{296}, \bibinfo{number}{1} (\bibinfo{year}{2010}),
  \bibinfo{pages}{175--213}.
\newblock


\bibitem[\protect\citeauthoryear{Mlot, Tovey, and Hu}{Mlot
  et~al\mbox{.}}{2011}]%
        {Mlot2011}
\bibfield{author}{\bibinfo{person}{Nathan~J. Mlot}, \bibinfo{person}{Craig~A.
  Tovey}, {and} \bibinfo{person}{David~L. Hu}.}
  \bibinfo{year}{2011}\natexlab{}.
\newblock \showarticletitle{Fire ants self-assemble into waterproof rafts to
  survive floods}.
\newblock \bibinfo{journal}{\emph{Proceedings of the National Academy of
  Sciences}} \bibinfo{volume}{108}, \bibinfo{number}{19}
  (\bibinfo{year}{2011}), \bibinfo{pages}{7669--7673}.
\newblock


\bibitem[\protect\citeauthoryear{Moubarak and Ben-Tzvi}{Moubarak and
  Ben-Tzvi}{2012}]%
        {Moubarak2012}
\bibfield{author}{\bibinfo{person}{Paul Moubarak} {and} \bibinfo{person}{Pinhas
  Ben-Tzvi}.} \bibinfo{year}{2012}\natexlab{}.
\newblock \showarticletitle{Modular and reconfigurable mobile robotics}.
\newblock \bibinfo{journal}{\emph{Robotics and Autonomous Systems}}
  \bibinfo{volume}{60}, \bibinfo{number}{12} (\bibinfo{year}{2012}),
  \bibinfo{pages}{1648--1663}.
\newblock


\bibitem[\protect\citeauthoryear{P{\"{o}}nitz and Tittman}{P{\"{o}}nitz and
  Tittman}{2000}]%
        {Ponitz2000}
\bibfield{author}{\bibinfo{person}{Andr{\'{e}} P{\"{o}}nitz} {and}
  \bibinfo{person}{Peter Tittman}.} \bibinfo{year}{2000}\natexlab{}.
\newblock \showarticletitle{Improved Upper Bounds for Self-Avoiding Walks in
  $\mathbb{Z}^d$}.
\newblock \bibinfo{journal}{\emph{The Electronic Journal of Combinatorics}}
  \bibinfo{volume}{7}, \bibinfo{number}{R21} (\bibinfo{year}{2000}),
  \bibinfo{pages}{1--10}.
\newblock


\bibitem[\protect\citeauthoryear{Reid, Lutz, Powell, Kao, Couzin, and
  Garnier}{Reid et~al\mbox{.}}{2015}]%
        {Reid2015}
\bibfield{author}{\bibinfo{person}{Chris~R. Reid}, \bibinfo{person}{Matthew~J.
  Lutz}, \bibinfo{person}{Scott Powell}, \bibinfo{person}{Albert~B. Kao},
  \bibinfo{person}{Iain~D. Couzin}, {and} \bibinfo{person}{Simon Garnier}.}
  \bibinfo{year}{2015}\natexlab{}.
\newblock \showarticletitle{Army ants dynamically adjust living bridges in
  response to a cost--benefit trade-off}.
\newblock \bibinfo{journal}{\emph{Proceedings of the National Academy of
  Sciences}} \bibinfo{volume}{112}, \bibinfo{number}{49}
  (\bibinfo{year}{2015}), \bibinfo{pages}{15113--15118}.
\newblock


\bibitem[\protect\citeauthoryear{Restrepo, Shin, Tetali, Vigoda, and
  Yang}{Restrepo et~al\mbox{.}}{2013}]%
        {Restrepo2013}
\bibfield{author}{\bibinfo{person}{Ricardo Restrepo}, \bibinfo{person}{Jinwoo
  Shin}, \bibinfo{person}{Prasad Tetali}, \bibinfo{person}{Eric Vigoda}, {and}
  \bibinfo{person}{Linji Yang}.} \bibinfo{year}{2013}\natexlab{}.
\newblock \showarticletitle{Improving mixing conditions on the grid for
  counting and sampling independent sets}.
\newblock \bibinfo{journal}{\emph{Probability Theory and Related Fields}}
  \bibinfo{volume}{156} (\bibinfo{year}{2013}), \bibinfo{pages}{75--99}.
\newblock


\bibitem[\protect\citeauthoryear{Rivault and Cloarec}{Rivault and
  Cloarec}{1998}]%
        {Rivault1998}
\bibfield{author}{\bibinfo{person}{Colette Rivault} {and} \bibinfo{person}{Ann
  Cloarec}.} \bibinfo{year}{1998}\natexlab{}.
\newblock \showarticletitle{Cockroach aggregation: discrimination between
  strain odours in {B}lattella germanica}.
\newblock \bibinfo{journal}{\emph{Animal Behaviour}} \bibinfo{volume}{55},
  \bibinfo{number}{1} (\bibinfo{year}{1998}), \bibinfo{pages}{177--184}.
\newblock


\bibitem[\protect\citeauthoryear{Rubenstein, Cornejo, and Nagpal}{Rubenstein
  et~al\mbox{.}}{2014}]%
        {Rubenstein2014}
\bibfield{author}{\bibinfo{person}{Michael Rubenstein},
  \bibinfo{person}{Alejandro Cornejo}, {and} \bibinfo{person}{Radhika Nagpal}.}
  \bibinfo{year}{2014}\natexlab{}.
\newblock \showarticletitle{Programmable self-assembly in a thousand-robot
  swarm}.
\newblock \bibinfo{journal}{\emph{Science}} \bibinfo{volume}{345},
  \bibinfo{number}{6198} (\bibinfo{year}{2014}), \bibinfo{pages}{795--799}.
\newblock


\bibitem[\protect\citeauthoryear{Savoie, Cannon, Daymude, Warkentin, Li, Richa,
  Randall, and Goldman}{Savoie et~al\mbox{.}}{2018}]%
        {Savoie2018}
\bibfield{author}{\bibinfo{person}{William Savoie}, \bibinfo{person}{Sarah
  Cannon}, \bibinfo{person}{Joshua~J. Daymude}, \bibinfo{person}{Ross
  Warkentin}, \bibinfo{person}{Shengkai Li}, \bibinfo{person}{Andr\'ea~W.
  Richa}, \bibinfo{person}{Dana Randall}, {and} \bibinfo{person}{Daniel~I.
  Goldman}.} \bibinfo{year}{2018}\natexlab{}.
\newblock \showarticletitle{Phototactic Supersmarticles}.
\newblock \bibinfo{journal}{\emph{Artificial Life and Robotics}}
  \bibinfo{volume}{23}, \bibinfo{number}{4} (\bibinfo{year}{2018}),
  \bibinfo{pages}{459--468}.
\newblock


\bibitem[\protect\citeauthoryear{Shamos and Hoey}{Shamos and Hoey}{1976}]%
        {Shamos1976}
\bibfield{author}{\bibinfo{person}{Michael~I. Shamos} {and}
  \bibinfo{person}{Dan Hoey}.} \bibinfo{year}{1976}\natexlab{}.
\newblock \showarticletitle{Geometric intersection problems}. In
  \bibinfo{booktitle}{\emph{17th Annual Symposium on Foundations of Computer
  Science}} \emph{(\bibinfo{series}{SFCS '76})}. \bibinfo{publisher}{IEEE
  Computer Society}, \bibinfo{address}{Washington, DC, USA},
  \bibinfo{pages}{208--215}.
\newblock


\bibitem[\protect\citeauthoryear{Thakker, Kamat, Bharambe, Chiddarwar, and
  Bhurchandi}{Thakker et~al\mbox{.}}{2014}]%
        {Thakker2014}
\bibfield{author}{\bibinfo{person}{Rohan Thakker}, \bibinfo{person}{Ajinkya
  Kamat}, \bibinfo{person}{Sachin Bharambe}, \bibinfo{person}{Shital
  Chiddarwar}, {and} \bibinfo{person}{Kishor~M. Bhurchandi}.}
  \bibinfo{year}{2014}\natexlab{}.
\newblock \showarticletitle{ReBiS - Reconfigurable Bipedal Snake robot}. In
  \bibinfo{booktitle}{\emph{2014 IEEE/RSJ International Conference on
  Intelligent Robots and Systems}} \emph{(\bibinfo{series}{IROS '14})}.
  \bibinfo{publisher}{IEEE}, \bibinfo{pages}{309--314}.
\newblock


\bibitem[\protect\citeauthoryear{Walter, Tsai, and Amato}{Walter
  et~al\mbox{.}}{2005}]%
        {Walter2005}
\bibfield{author}{\bibinfo{person}{Jennifer~E. Walter},
  \bibinfo{person}{Elizabeth~M. Tsai}, {and} \bibinfo{person}{Nancy~M. Amato}.}
  \bibinfo{year}{2005}\natexlab{}.
\newblock \showarticletitle{Algorithms for Fast Concurrent Reconfiguration of
  Hexagonal Metamorphic Robots}.
\newblock \bibinfo{journal}{\emph{IEEE Transactions on Robotics}}
  \bibinfo{volume}{21}, \bibinfo{number}{4} (\bibinfo{year}{2005}),
  \bibinfo{pages}{621--631}.
\newblock


\bibitem[\protect\citeauthoryear{Woods, Chen, Goodfriend, Dabby, Winfree, and
  Yin}{Woods et~al\mbox{.}}{2013}]%
        {Woods2013}
\bibfield{author}{\bibinfo{person}{Damien Woods}, \bibinfo{person}{Ho-Lin
  Chen}, \bibinfo{person}{Scott Goodfriend}, \bibinfo{person}{Nadine Dabby},
  \bibinfo{person}{Erik Winfree}, {and} \bibinfo{person}{Peng Yin}.}
  \bibinfo{year}{2013}\natexlab{}.
\newblock \showarticletitle{Active self-assembly of algorithmic shapes and
  patterns in polylogarithmic time}. In \bibinfo{booktitle}{\emph{Proceedings
  of the 4th Innovations in Theoretical Computer Science Conference}}
  \emph{(\bibinfo{series}{ITCS '13})}. \bibinfo{publisher}{ACM},
  \bibinfo{address}{New York, NY, USA}, \bibinfo{pages}{353--354}.
\newblock


\end{thebibliography}

\end{document}